\documentclass[12pt,a4paper]{article}
\PassOptionsToPackage{hypertexnames=false,hyperfootnotes=false}{hyperref}
\usepackage{wl-art}
\newcommand{\tmath}[1]{\texorpdfstring{\(#1\)}{math}}
\newcommand{\indicator}[1]{\mathbb{I}\left\{#1\right\}}

\DeclareMathOperator{\Prob}{P}
\DeclareMathOperator{\Ep}{\mathrm{E}}

\DeclareMathOperator{\Var}{Var}
\DeclareMathOperator{\Cov}{Cov}
\newcommand{\dd}{\,\mathrm{d}}

\renewcommand{\epsilon}{\varepsilon}
\renewcommand{\emptyset}{\varnothing}

\newcommand{\pqty}[1]{\left(#1\right)}
\newcommand{\bqty}[1]{\left[#1\right]}
\newcommand{\cqty}[1]{\left\{#1\right\}}

\newcommand{\abs}[1]{\left\lvert#1\right\rvert}
\newcommand{\norm}[1]{\left\lVert#1\right\rVert}
\newcommand{\mathtext}[1]{\text{ #1 }}

\newcommand{\nb}[2]{\mathcal{N}_{#1}^{#2}}
\newcommand{\rcd}[1]{\lambda^{\mathcal{H}}_{\omega}\pqty{#1}}
\newcommand{\gga}{\Gamma^{i,j}}
\newcommand{\dist}{\mathbf{d}}
\newcommand{\lst}{\mathbf{i}}
\newcommand{\correction}[1]{#1}

\declaretheorem[name=Proposition,parent=section]{journalproposition}

\AtBeginDocument{%
  \bibcite{bradley2005BasicProperties}{{}{2005}{{Bradley}}{{}}}%
  \bibcite{prakasarao2009ConditionalIndependence}{{}{2009}{{Prakasa~Rao}}{{}}}%
  \bibcite{rio2017AsymptoticTheory}{{}{2017}{{Rio}}{{}}}%
}
\begin{document}
\title{Normal approximation for U-statistics with cross-sectional dependence}
\author{Weiguang Liu\thanks{Author email: weiguang.liu@ucl.ac.uk. The author thanks conference and seminar participants at the 2023 European Meeting of the Econometric Society, University of Cambridge, Yale University, Xiamen University, the Chinese University of Hong Kong (Shenzhen), and the Hong Kong University of Science and Technology for their valuable comments and feedback. We thank the generous funding from the UK Research and Innovation (UKRI) under the UK government's Horizon Europe funding guarantee (Grant Ref: EP/X02931X/1). Published in \textit{Journal of Econometrics}, 258 (2026), \href{https://doi.org/10.1016/j.jeconom.2026.106325}{106325}.}\\\small{University College London}}
\date{}
\maketitle
\begin{abstract}
We establish normal approximation in the Wasserstein metric for both non-degenerate and degenerate second-order U-statistics under cross-sectional dependence using Stein's method. For the non-degenerate case, our results extend recent studies on the asymptotic properties of sums of cross-sectionally dependent random variables. The degenerate case is more challenging due to the additional dependence induced by the nonlinearity of the U-statistic kernel. Through a specific implementation of Stein's method, we derive convergence rates under conditions on the mixing rate, the sparsity of the cross-sectional dependence structure, and the moments of the U-statistic kernel. Finally, we demonstrate the application of our theoretical results with a nonparametric specification test for data with cross-sectional dependence.
\end{abstract}
\onehalfspacing
\section{Introduction}\subsection{U-statistics}
Cross-sectional dependence is pervasive in economic data. Unlike time series data, where observations follow a natural temporal ordering, cross-sectional data consist of dependent observations without an inherent ordering, as individuals may be interconnected in complex ways.
In this paper, we study the asymptotic normal approximation for the distribution of an important class of statistics, known as second-order U-statistics, based on a cross-sectionally dependent sample \(\cqty{X_{i}: i \in \mathcal{I}_{n}}\). A second-order U-statistic can be written as
\begin{equation}
    S_{n} = \sum_{\{i,k \in \mathcal{I}_n : i \neq k\}} H_{n}\pqty{X_{i},X_{k}}.
\end{equation}
In the definition, \(H_{n}:\mathbb{R}^{d}\times \mathbb{R}^{d} \to \mathbb{R}\) is a symmetric measurable function known as the \textit{kernel} of the U-statistic \(S_{n}\).
With an independent and identically distributed (i.i.d.) sample and a fixed permutation symmetric kernel \(H\), this leads to an unbiased estimator of \(\theta = \Ep H(X_{1},X_{2})\) after normalisation by the inverse of \(n(n-1)\). Therefore U-statistics can be thought of as a generalisation of sample averages and many important statistics are U-statistics or can be approximated by U-statistics \citep{serfling1980ApproximationTheorems}.
Examples include the sample variance estimators, Kendall's \(\tau\), and Cramer-von-Mises statistics \citep{vandervaart2000AsymptoticStatistics}. Another motivating example is the nonparametric specification tests as in  \cite{fan1996ConsistentModel} and \cite{li2007NonparametricEconometrics} for regression models, which will be studied in more details in Section \ref{s1:nonpar test}.
% \begin{example}[Nonparametric specification test]\label{example:nonpar-specification}
%     Suppose that we have a regression model and a sample of \(\cqty{\pqty{Y_{i}, Z_{i}}: i\in \mathcal{I}_{n}}\), and we want to conduct the following specification test that given a parametric function \(g(z, \gamma)\) that is known for a given $\gamma \in \Gamma$,
%     \begin{align*}
%         \mathbb{H}_{0}&: \Ep(Y_{i}\mid Z_{i}) = g(Z_{i}, \gamma) \mathtext{a.s. for some} \gamma\in \Gamma_{0};\\
%         \mathbb{H}_{1}&: \Prob\bqty{\Ep\pqty{Y_{i}\mid Z_{i}} \neq g(Z_{i},\gamma)} > 0 \mathtext{for all} \gamma\in \Gamma_{0}.
%     \end{align*}
%     where \(\Gamma_{0} \subset \Gamma\) contains the parameter values consistent with the economic theory.
%     Let \(X_{i} = (u_{i}, Z_{i}')'\) where \(u_{i} = Y_{i} - g(Z_i, \gamma)\), a class of test statistics have been proposed for this specification test problem based on kernel smoothing method, for example,. The key component of the test statistics is the following second-order U-statistics,
%     \begin{equation*}
%         \mathbf{I}_{n,0} =\sum_i \sum_{k\neq i} H_n \pqty{X_{i}, X_{k}} = \sum_i \sum_{k\neq i} u_{i} u_{k} K\pqty{\frac{Z_{i} - Z_{k}}{b}},
%     \end{equation*}
%     where \(K(\cdot)\) is a \textit{kernel smoothing function} and \(b= b_n\) is a suitable sequence of bandwidths. The notation $\sum_i \sum_{k\neq i}$ is a shorthand for $\sum_{\cqty{i,k \in \mathcal{I}_n, i \neq k}}$.
% \end{example}

Although central limit theorems have been developed for U-statistics with independent or time series samples, existing theories do not apply under cross-sectional dependence and previous methods of proof are no longer applicable. This paper seeks to fill the gap and show conditions under which a normalised U-statistic \(S_n\) based on cross-sectionally dependent data converges in distribution to a normal random variable.

% We follow the recent developments on cross-sectional and network dependence and assume there is an economic distance \(\dist\) defined on \(\mathcal{I}_{n}\) that is potentially unobserved. The cross-sectional dependence is assumed to be weak in the sense that the strength of dependence between observations decreases fast enough as the distance increases.
% This approach generalises the notion of weak dependence in time series and random fields.

The difficulty of extending existing results on U-statistics to allow for cross-sectional dependence lies in the need to handle both the dependence induced by the nonlinear structure of U-statistics and the cross-sectional dependence in the underlying sample, where the lack of an ordering prohibits the use of classical methods such as the martingale method. In the following, we discuss these challenges and how we approach them.

\subsection{Nonlinearity and degeneracy}
The structure of U-statistics induces dependence among the summands due to the nonlinearity in the kernel \(H_n\).
The classical approach to handle this challenge is to apply the famous \textit{Hoeffding decomposition}. Let \(\hat{H}_{k}(x) = \Ep H(X_{k}, x)\) and \(\theta_{ik} = \Ep \hat{H}_{k}(X_{i}) = \Ep \hat{H}_{i}(X_{k})\), we can write \(S_{n}\) as
\begin{equation}\label{eq:hoeffding}
    S_{n} = \sum_{i}\sum_{k\neq i} \theta_{ik} + \hat{S}_{n} + {S}^{*}_{n},
\end{equation}
where \(\hat{S}_{n} = 2 \sum_{i}\sum_{k\neq i} (\hat{H}_{k}(X_{i}) - \theta_{ik})\), and \({S}^{*}_{n} = \sum_i \sum_{k\neq i} \pqty{H(X_{i}, X_{k}) - \hat{H}_{i}(X_{k}) - \hat{H}_{k}(X_{i}) + \theta_{ik}}.\)
In the decomposition, \(\hat{S}_{n}\) is referred to as the projection part and \(S^{*}_{n}\) the degenerate part of the U-statistic \(S_{n}\). We say that a U-statistic is itself degenerate if the first-order projection terms vanish, that is, \(\hat{H}_{k}(X_{i}) - \theta_{ik} = 0\) almost surely for all distinct \(i,k \in \mathcal{I}_{n}\).

When the U-statistic of interest is non-degenerate, the Hoeffding decomposition reduces it to a sum of random variables with an error term \(S^{*}_{n}\). It is then straightforward to obtain central limit theorems with Slutsky's Lemma by demonstrating that, under suitable normalisation, the projection part \(\hat{S}_{n}\) converges in law to a normal random variable, while the degenerate part \(S^{*}_{n}\) is asymptotically negligible.

However, many useful statistics are degenerate in nature, such as those appearing in nonparametric specification tests (Section \ref{s1:nonpar test}). For degenerate U-statistics, the analysis is more involved, since the linear approximation based on projection is now unavailable. In general the limiting distributions for degenerate U-statistics are Gaussian Chaos \citep{vandervaart2000AsymptoticStatistics} and, in the case of U-statistics of order \(2\), it is a weighted sum of centered \(\chi^{2}\) random variables with weights depending on the spectrum of the kernel function.

It is of interest to develop conditions under which asymptotic normality can be recovered for degenerate U-statistics.
One of the classic results due to \cite{hall1984CentralLimit} finds the following sufficient conditions for a CLT for degenerate U-statistics with a random sample of independent and identically distributed random variables \(\cqty{X_{i}: 1 \leq i \leq n}\). If the kernel function \(H_n\) has bounded second moment for each $n$, and
\begin{equation}\label{eq:hall's condition}
    \frac{{\Ep \Gamma_n^{2}(X_{1}, X_{2}) + \frac{1}{n} \Ep H_n^{4}(X_{1}, X_{2})}}{\pqty{\Ep H_n^{2}(X_{1}, X_{2})}^{2} }\to 0,
\end{equation}
where \(\Gamma_n(x, y) = \Ep\bqty{H_n(X_{1}, x) H_n(X_{1}, y)}\), then a degenerate U-statistic \(S_{n}\) with symmetric and centered kernel \(H_n\) is asymptotically normally distributed, with mean $0$ and variance $2n^2 \Ep H_n^2(X_1,X_2)$.
These conditions will be a benchmark for our results on the degenerate U-statistics.

It is worth noting that \cite{hall1984CentralLimit}'s proof relies on martingale representation of \(S_{n}\).
% as
% \begin{equation*}
%     S_{n} = 2\sum_{i=1}^n T_{i}, \quad T_i = \sum_{j < i} H(X_{i} , X_{j})
% \end{equation*}
% and \(\Ep\bqty{T_{i + 1} \mid \mathcal{F}_{i}} = 0\) for \(\mathcal{F}_{i} = \sigma(X_{1},\dots, X_{i})\) thanks to the degeneracy condition and the independence assumption.
Generalisations to time series settings, such as \cite{fan1999CentralLimit}, also rely on the ordering of the observations which are not available in the cross-sectional dependence settings.
We now discuss how we model the cross-sectional dependence and how to handle this additional complication.

\subsection{Cross-sectional dependence}

Cross-sectional dependence in the underlying sample \(\cqty{X_{i}: i\in \mathcal{I}_{n}}\) is a common phenomenon in economic data.
There has been a growing interest in modelling cross-sectional dependence as well as in the estimation and inference under cross-sectional dependence.
\cite{bolthausen1982CentralLimit} studies central limit theorems for stationary and mixing random fields, where the random variables are indexed by their location on the lattices \(\mathbb{Z}^{d}\).
\cite{anselin1988SpatialEconometrics} provides review on the spatial statistics.
\cite{conley1999GMMEstimation} further studies GMM estimation and inference when there exists an economic distance among the observations, which is allowed to be imperfectly measured.
\cite{jenish2009CentralLimit} and \cite{jenish2012SpatialProcesses} extend the previous results to allow for more general dependence structure, nonstationarity and potentially unbounded moments in the random fields.
There is also a long tradition of studying dependency graph in the literature of Stein's methods, see for example, \cite{chen1986RateConvergence}, \cite{baldi1989NormalApproximations}, \cite{rinott1996MultivariateCLT} and further extensions in \cite{chen2004NormalApproximation}.
Recent extensions allow for observations located on a stochastic network, for example \cite{kojevnikov2020LimitTheorems} considers estimation and inference under network dependence where they assume there is a random network and the economic distance is determined by the length of the shortest path connecting two nodes. \cite{vainora2020NetworkDependence} defines network stationarity and studies estimation and inference.

We adopt the approach to model the cross-sectional dependence by assuming that there is a measure of economic distance \(\dist\) between the observations and generalise current results on sums of cross-sectionally dependent samples to U-statistics. We will consider \textit{weak} cross-sectional dependence in the sense that two groups of observations will be approximately independent if their economic distance is large, which is closely related to the settings in \cite{jenish2009CentralLimit} which assumes the observations form random fields, and \cite{kojevnikov2020LimitTheorems} that specifies the distance to be the distance on a network. For our results to hold, the observation of the economic distance is not necessary, provided that there exists such a distance for which cross-sectional dependence is weak and the mixing rates apply. In practice, the observation of a network can be useful to check whether the conditions hold and for feasible inference (see methods proposed in \citealp{kojevnikov2020LimitTheorems} and \citealp{kojevnikov2021BootstrapNetwork}).

We handle the cross-sectional dependence with Stein's method. The Stein's method begins with the seminal paper \cite{stein1972BoundError}, which estimates the error of normal approximation for a sum of random variables with a certain dependence structure. It has since become one of the state-of-the-art ways to estimate the distance between two random variables under various probability metrics.
Several instances of the Stein's method have been developed under the general idea of ``auxiliary randomisation'' introduced by \cite{stein1986ApproximateComputation}, such as the exchangeable pairs approach, zero-bias and size-bias couplings, see \cite{ross2011FundamentalsSteins}. Using the approach of exchangeable pairs, \cite{rinott1997CouplingConstructions} and \cite{dobler2016QuantitativeJong} extend and provide convergence rates for the central limit theorems for U-statistics in \cite{hall1984CentralLimit} and \cite{dejong1987CentralLimit} respectively, with an independent sample \(\cqty{X_{i}: i \in \mathcal{I}}\). However, their construction of exchangeable pairs depends on the independence condition and is not straightforward to be extended to the case of cross-sectional dependence.

We base our analysis on a variant of Stein's coupling method and the decomposable method \citep{barbour1989CentralLimit, chen2010SteinCouplings}, as stated in \autoref{lemma:stein}. This approach is tailored to handle U-statistics, particularly degenerate ones, and differs from the other Stein's methods typically used for analysing sums of dependent random variables. We believe this method is of interest in its own right.

Our results are of interest to other strands of literature related to U-statistics, and we name a few examples here. \cite{athey2019GeneralizedRandom} shows that a generalised random forest has a U-statistic representation. The estimation and inference for multiway clustered data and exchangeable arrays depend crucially on U-statistics; see, for example, \cite{chiang2021InferenceHighdimensional}, \cite{cattaneo2022UniformInference}, \cite{chiang2023UsingTwoWay}, and \cite{chiang2024StandardErrors}. In some cases, the objective function in an M-estimation takes the form of a U-statistic; see \cite{honore1997PairwiseDifference}. In panel data settings, this also opens up a new avenue for showing limiting distributions along the cross-sectional dimension. \cite{anatolyev2020LimitTheorems} proposes a method that allows for unspecified cross-sectional dependence and utilizes independence or martingale properties along the time series dimension. \cite{zaffaroni2019FactorModels} considers conditional factor model with short time series. Our results can be used in the settings when the number of time series observations is limited, due to data availability or stationarity considerations.

\subsection*{Overview}
This paper is organised in the following way. Section \ref{sec:model} sets up the model and defines the key ingredients required for the theoretical results, with some examples given in Section \ref{sec:examples}. Section \ref{sec:theory} presents the main findings where we estimate the error of normal approximation to the non-degenerate and degenerate U-statistics. The convergence rates will depend on the mixing rates, the sparsity of the cross-sectional dependence, and moments of the kernel functions.
Section \ref{s1:nonpar test} provides an application of the CLT for degenerate U-statistics on the nonparametric specification tests. Section \ref{s1:conclusion} concludes.
Proofs for the theoretical results and further extensions are collected in the Appendix.

\subsection*{Notations}
For \(n\in \mathbb{N}\), \([n]\) denotes the set  \(\cqty{1, 2,\dots, n}\). \(\mathbb{I}(\cdot)\) is the indicator function.
For a vector of indices \(\mathbf{i} = (i_{1},\dots, i_{q})\), we will write \(i\in \mathbf{i}\) if \( i\in \cqty{i_{1},\dots, i_{q}}\), and we define the subvector \(\mathbf{i}_{-i_{k}} = ({i_{1},\dots, i_{k - 1}, i_{k + 1}, \dots, i_{q}})\).
For a collection of random vectors \(\cqty{X_{i}:i\in \mathcal{I}}\), \(\sigma\pqty{X_{i}: i\in \mathcal{I}}\) denotes the \(\sigma\)-algebra generated by \(\cqty{X_{i}}\), and \(\tilde{X}_{i}\) denotes a random vector that has the same distribution as \(X_{i}\).
\(\abs{\cdot}\) denotes the Euclidean norm and \(\norm{X}_{p} = \bqty{\Ep\abs{X}^{p}}^{{1}/{p}}\) denotes the \(L^{p}\)-norm of \(X\).
We will write \(\Ep(\cdot)\) for unconditional expectations and \(\Ep^{X} \pqty{Y}\) for conditional expectation \(\Ep\bqty{Y \mid X}\).
For sequences of positive numbers, \(a_{n} = O(b_{n})\) or \(a_{n}\lesssim b_{n}\) if \(a_{n} \leq C b_{n}\) for some constant \(C > 0\) and sufficiently large \(n\);
\(a_{n} =\Theta\pqty{b_{n}}\) if \(a_{n} = O(b_{n})\) and \(b_{n} = O(a_{n})\).
A sequence of random variables \(X_{n} \rightsquigarrow X\) if \(X_{n}\) converges in distribution to a random variable \(X\).

\section{Model setup}\label{sec:model}\subsection{Model for cross-sectional dependence}
Suppose we observe a sample \(\cqty{X_{i}: i\in \mathcal{I}_{n}}\), where each \(X_{i} :\pqty{\Omega, \mathcal{F}, \Prob} \to (\mathbb{R}^{d}, \mathcal{B}_{\mathbb{R}^{d}})\) is a random vector, \(\mathcal{I}_{n}\) is the index set for the entities of interest, and \(\mathcal{B}_{\mathbb{R}^{d}}\) is the Borel \(\sigma\)-algebra on \(\mathbb{R}^{d}\). We will also refer to an index \(i\in \mathcal{I}_{n}\) as a node, to subsets \(\mathcal{I}_{n}'\subset \mathcal{I}_{n}\) as groups of nodes, and we assume the sample size \(\abs{\mathcal{I}_{n}} = n\). To model cross-sectional dependence, we assume that \(\pqty{\mathcal{I}_{n}, \dist}\) is a metric space with distance \(\dist: \mathcal{I}_{n}\times \mathcal{I}_{n} \to \bqty{0, \infty}\). For example, \(\dist\) can measure physical distance, as in a spatial model, or the strength of economic links between entities. For two groups of nodes \(\mathcal{I}_{1}, \mathcal{I}_{2} \subset \mathcal{I}_{n}\), we naturally define the distance
\begin{equation*}
    \dist \pqty{\mathcal{I}_{1}, \mathcal{I}_{2}} = \inf_{i_{1}\in \mathcal{I}_{1}, i_{2}\in \mathcal{I}_{2}} \dist(i_{1}, i_{2}).
\end{equation*}

The strength of cross-sectional dependence among the observations \(\cqty{X_{i}:i\in \mathcal{I}_{n}}\) will be measured by the \(\beta\)-mixing coefficients.
The \(\beta\)-mixing coefficient between two \(\sigma\)-algebras \(\mathcal{F}_{1}, \mathcal{F}_{2}\) is defined as
\begin{equation*}
    \beta\pqty{\mathcal{F}_{1}, \mathcal{F}_{2}} = \sup \cqty{ \frac{1}{2} \sum_{i = 1}^{n_A}\sum_{j = 1}^{n_B}\abs{P(A_{i}\cap B_{j}) - P(A_{i}) P(B_{j})}},
\end{equation*}
where the supremum is taken over the set of all finite partitions \(\mathcal{A} = \cqty{A_i: 1 \leq i \leq n_A}\) and \(\mathcal{B}=\cqty{B_{j}: 1\leq j \leq n_B}\) of \(\Omega\) for some \(n_A, n_B < \infty\), with \( \mathcal{A}\subset \mathcal{F}_{1}\) and \(\mathcal{B}\subset \mathcal{F}_{2}\). It is well known that \(\beta\pqty{\mathcal{F}_{1}, \mathcal{F}_{2}} = 0\) if and only if \(\mathcal{F}_{1}\) and \(\mathcal{F}_{2}\) are independent \citep{doukhan1994Mixing}.

We will make use of the following sets, consisting of ordered pairs of groups of nodes that are at least distance \(m\) apart and satisfy restrictions on their sizes. This definition is common in the random-field literature \citep{jenish2009CentralLimit}.
\begin{equation}\label{eq: definition of P_n}
    \mathcal{P}_{n}(n_{1}, n_{2}, m)= \cqty{(\mathcal{I}_{1}, \mathcal{I}_{2}): \mathcal{I}_{1},\mathcal{I}_{2} \subset \mathcal{I}_{n}, \abs{\mathcal{I}_{1}} \leq n_{1}, \abs{\mathcal{I}_{2}} \leq n_{2}, \dist \pqty{\mathcal{I}_{1}, \mathcal{I}_{2}} \geq m}.
\end{equation}
The restrictions we impose on cross-sectional dependence are stated in the following assumptions.
\begin{assumption}\label{asmp:no-clustering-point}
    For all \(n\) and distinct \(i,j\in \mathcal{I}_{n}\), \(i\neq j\), the distance between them is bounded away from zero, \(\dist(i,j) \geq \underline{\delta} > 0\). Without loss of generality, we take \(\underline{\delta} = 1\).
\end{assumption}
\begin{assumption}\label{asmp:beta-mixing}
    The sample \(\cqty{X_{i}: i\in \mathcal{I}_{n}}\) is \textit{\(\beta\)-mixing} with coefficients \(\beta_{n}(n_{1}, n_{2}, m)\), so that
    \begin{equation}\label{eq:beta mixing}
        \sup_{(\mathcal{I}_{1}, \mathcal{I}_{2}) \in \mathcal{P}_{n}(n_{1},n_{2},m)}\beta(\sigma(\mathcal{I}_{1}), \sigma(\mathcal{I}_{2}))\leq \beta_{n}(n_{1}, n_{2}, m),
    \end{equation}
    where \(\beta(\sigma(\mathcal{I}_{1}), \sigma(\mathcal{I}_{2})) = \beta\pqty{\sigma\pqty{X_{i}: i \in \mathcal{I}_{1}}, \sigma\pqty{X_{i} : i \in \mathcal{I}_{2}}} \).
    Note that \(\beta_{n}(n_{1},n_{2},m)\) is non-decreasing in \(n_{1}\) and \(n_{2}\), and non-increasing in \(m\). Let \(\beta(n_1, n_{2}, m):= \sup_{n} \beta_{n}(n_{1},n_{2}, m)\). In our theoretical proofs, we require these \(\beta\)-mixing coefficients to approach zero sufficiently fast as \(m\) increases.
\end{assumption}

A few comments are in order.
The assumptions on weak dependence in the observations \(\cqty{X_{i}}\) are similar to standard assumptions in the random-field literature; see, for example, \cite{jenish2009CentralLimit}. The first lower-bound condition is imposed to rule out increasing concentration of observations into arbitrarily small neighbourhoods as the sample grows. What is essential for our counting arguments is that local multiplicity remains controlled.
More generally, the distance here should be interpreted as a device for controlling the strength of dependence. In spatial settings, even observations at the same physical location need not be perfectly dependent, so it is reasonable to allow for a strictly positive “dependence distance” in such cases.
The second assumption extends familiar mixing conditions from time series and random fields to the more general settings we consider.

We choose to state our conditions in terms of \(\beta\)-mixing, or absolute regularity, other than some other mixing concepts considered in the literature on limit theorems for sums of random variables (for example, \citealp{kojevnikov2020LimitTheorems} considers \(\psi\)-mixing) because \(\beta\)-mixing random variables have stronger coupling properties, as reflected in Berbee's Lemma (\autoref{lemma:berbee}), which is essential for handling the nonlinear U-statistics kernel. Berbee's Lemma can be extended to the case of \(\alpha\)-mixing under a Lipschitz continuity condition on the kernel, see \cite{bradley1983ApproximationTheorems} and \cite{dehling2010CentralLimit}. Therefore, our results can be readily generalised to allow for \(\alpha\)-mixing under Lipschitz continuity. Since the extension to \(\alpha\)-mixing complicates the notation without adding much insight, we maintain the assumption of \(\beta\)-mixing. Our results are first stated for the case where the distance \(\dist\) is deterministic, and in Section \ref{s2: conditional berbee} we discuss extensions that allow the underlying sample \(\cqty{X_{i}:i\in \mathcal{I}}\) to be conditionally \(\beta\)-mixing by providing a conditional version of Berbee's lemma. This extension is relevant in practice.

It is well known from the literature on random fields and network dependence that the topology of the index space \(\mathcal{I}_{n}\) plays an important role in determining the strength of cross-sectional dependence.
Indeed, from the definition of \(\mathcal{P}_{n}\pqty{n_{1},n_{2},m}\) in \eqref{eq: definition of P_n} and \autoref{asmp:beta-mixing} in \eqref{eq:beta mixing}, we see that cross-sectional dependence is stronger if the mixing coefficients \(\beta(\sigma(\mathcal{I}_{1}), \sigma(\mathcal{I}_{2}))\) are larger for each \((\mathcal{I}_{1}, \mathcal{I}_{2})\), or if more nodes in \(\mathcal{I}_{n}\) lie within distance \(m\) of one another and the indices are more clustered.
We therefore need further requirements on the topology of the index space \(\mathcal{I}_{n}\) with respect to the distance \(\dist\). In particular, we define the following quantities that measure the ``sparsity'' of the index space.

The normal approximation bounds below also use counts of index vectors that fall into different dependence configurations. For a vector of \(q\) indices \(\mathbf{i}=(i_1,\ldots,i_q)\) and \(m<\infty\), form a graph on the indices \(\{i_{1},\ldots,i_{q}\}\) by connecting positions \(i_{r}\) and \(i_{s}\) whenever \(\dist(i_r,i_s)\leq m\). We say that \(\mathbf{i}\) is \(m\)-connected if this graph is connected, and denote by \(\tau_q^m\) the number of vectors of \(q\) indices that are \(m\)-connected.
More generally, \(\tau_{q_1,\ldots,q_S}^m\) denotes the number of vectors of \(q\) indices whose \(m\)-connected components have sizes \(q_1\geq \cdots \geq q_S\), where \(\sum_{s=1}^S q_s=q\). For example, \(\tau_{2,2}^m\) counts vectors \((i_1,i_2,i_3,i_4)\) whose four positions can be partitioned into two \(m\)-connected components that are separated from each other by more than \(m\). And \(\tau_{1,1,1,1}^m\) counts vectors \((i_1,i_2,i_3,i_{4})\) where any two of the four indices have distance greater than \(m\). The full definition based on equivalent classes is given in the Appendix.

The distribution of \(\tau^{m}_{q_{1},\dots,q_{S}}/n^{q}\) over different profiles \((q_{1},\dots,q_{S})\) provides important information about the sparsity of the index space \(\mathcal{I}_{n}\).
Consider the case \(q = 2\) with fixed \(m>0\). Then \(\tau_{2}^{m}\) is the number of ordered pairs \((i,j)\in \mathcal{I}_{n}^{2}\) that lie in one another's \(m\)-neighbourhoods, \(i\stackrel{m}{\leftrightarrow} j\), while \(\tau_{1,1}^{m}\) is the number of pairs \((i,j)\) such that \(\dist(i,j) > m\).
When the index space is sparse, the proportion of pairs of nodes that are more than \(m\) units apart, \({\tau_{1,1}^{m}}/{n^{2}}\), is expected to be higher. When the index space is less sparse, the relative proportion \(\tau_{2}^{m} / n^{2}\) increases.

It is easy to compute the profile \(\pi_{m}(\mathbf{i})\) of a given vector of indices \(\mathbf{i}\) and each \(\tau^{m}_{q_{1},\dots,q_{S}}\) when the distance is observed.
Furthermore, a simple bound on the sparsity measure \(\tau^{m}_{q_{1},\dots, q_{S}}\) can be obtained if we have a bound on the growth rate of the maximum neighbourhood sizes. Let \(\eta_{m,i} = |\nb{i}{m}|\) be the size of the \(m\)-neighbourhood of \(i\) for each \(i \in \mathcal{I}_{n}\), and let \(\eta_{m} = \max_{i} \eta_{m,i}\), which measures the maximum size of \(m\)-neighbourhoods. Then we have the following lemma, which is useful in many applications.
\begin{lemma}\label{lemma:sparsity-bound}
    We have the following upper bound, for \(\sum_{s = 1}^{S} q_{s} = q\), where \(q_{1}\geq q_{2}\geq \dots \geq q_{S} \in \mathbb{N}^{+}\),
    \begin{equation*}
        \tau_{q_{1},\dots, q_{S}}^{m} \leq C n^{S} \eta_{m}^{q - S},
    \end{equation*}
    where \(C\) is a constant that only depends on \((q_{1},\dots, q_{S})\).
\end{lemma}

Our main results provide bounds on the Wasserstein distance between a standard normal random variable and a suitably normalised U-statistic \(S_{n}\) constructed from a weakly cross-sectionally dependent collection of observations \(\cqty{X_{i}:i\in \mathcal{I}_{n}}\).
We provide conditions under which the Wasserstein distance converges to \(0\) as the sample size increases.
Since convergence in Wasserstein distance implies convergence in distribution, the central limit theorems follow from these general results.
The Wasserstein distance between two random variables \(V_{1}\) and \(V_{2}\) is defined by
\begin{equation}
    \dist_{W}\pqty{V_{1},V_{2}} := \sup_{g\in \mathbb{L}_{1}} \abs{\Ep g(V_{1}) - \Ep g(V_{2})}.
\end{equation}
where \(\mathbb{L}_{1} =  \{g: \mathbb{R} \to \mathbb{R} : |g(x) - g(y)| \leq |x - y|\}\) is the class of \(1\)-Lipschitz continuous functions.
\subsection{Examples}\label{sec:examples}

We discuss some examples of cross-sectional dependence structures that can be incorporated into our framework and how their sparsity can be measured.
% \begin{example}[Independence]
%     For a random sample of independent observations \(\cqty{X_{i}: i \in \mathcal{I}_{n}}\), where \(\mathcal{I}_{n} = \cqty{1, 2,\dots, n}\), the distance can be defined as \(\dist\pqty{i,j} = \abs{i - j}\) and \(\beta(n_{1},n_{2},m) = 0\) if \(m > 0\).
%     The equivalence relation of interest is $\stackrel{0}{\leftrightsquigarrow}$, that is $i \stackrel{0}{\leftrightsquigarrow} j \iff i = j$, the set of vectors of 4 indices can be partitioned into the sets $\cqty{T_{4}^{0}, T_{3,1}^{0}, T_{2,2}^{0}, T_{2,1,1}^{0}, T_{1^{4}}^{0}}$, with sizes in Table \ref{table:1}. The value of \(\tau_{q_{1},\dots,q_{s}}\) in the table combined with the main theorem in \ref{sec:theory} will imply the conditions of \cite{hall1984CentralLimit}.

% \end{example}

\begin{example}[Time series and random fields]
    In the study of random fields, \(\mathcal{I}_{n}\) is a subset of the lattice \(\mathbb{Z}^{d}\), which includes time series models as a special case when \(d = 1\). We can define the distance as \(\dist(i,j) = \max_{1\leq k \leq d} \abs{i_{k} - j_{k}}\); see \cite{bolthausen1982CentralLimit}. The requirement that \(\mathcal{I}_{n}\) be a subset of \(\mathbb{Z}^{d}\) can be relaxed by considering irregular lattices. If \(\mathcal{I}_{n}\subset \mathbb{R}^{d}\) and \autoref{asmp:no-clustering-point} holds, then \(\eta_{m} = O(m^{d})\).
\end{example}

\begin{example}[Simple networks]
    For an unweighted, undirected network \(\mathcal{G}\) defined over the nodes in \(\mathcal{I}_{n}\), where \(\mathcal{G}\pqty{i,j} = 1\) if \(i\) and \(j\) are directly connected, we can define the distance \(\dist(i,j)\) as the minimum length \(\abs{\mathcal{T}} -1\), where \(\mathcal{T} = \cqty{i_{0}, i_{1}, \dots, i_{l}}\) is a path connecting \(i\) and \(j\) in \(\mathcal{G}\), that is, \(i_{0}= i\), \(i_{l} = j\), and \(\mathcal{G}(i_{k}, i_{k+1}) =1\) for all \(k= 0, 1,\dots, l-1\). We set \(\dist(i,j) = \infty\) if there exists no path connecting \(i\) and \(j\). Time series dependence can be cast in this framework by considering a linear network \(\mathcal{G}(i,j) = 1\) if \(\abs{i - j} = 1\). In general, \(\tau^{m}\) depends on the specific network, but \autoref{lemma:sparsity-bound} can be used to obtain a bound when the neighbourhood sizes are bounded. If \(\beta_n(n_{1}, n_{2}, m) = 0\) for all \(m > 1\), then \(\mathcal{G}\) defines a dependency graph \citep{baldi1989NormalApproximations}, and \(X_{i}\) and \(X_{j}\) are independent whenever they are not connected.
\end{example}
\begin{example}[Clustering structure]
    Consider the two-way clustering structure, which can be generalised to multiway clustering \citep{menzel2021BootstrapClusterDependence}. Let \(C_{1}\) and \(C_{2}\) denote the numbers of clusters along the first and second dimensions. Each \(\mathbf{j} \in \cqty{(j_{1},j_{2}): j_{1} \leq C_{1}, j_{2} \leq C_{2}}\) represents a cell in the clustering structure. Let \(N_{\mathbf{j}}\) denote the number of observations in cell \(\mathbf{j}\), and let \(\cqty{X_{l, \mathbf{j}}, l = 1,\dots, N_{\mathbf{j}}}\) denote the observations in cell \(\mathbf{j}\). It is typically assumed that \(X_{l_{1},\mathbf{j}}\) is independent of \(X_{l_{2}, \mathbf{k}}\) if \(j_{1}\neq k_{1}\) and \(j_{2}\neq k_{2}\) for all \(l_{1}, l_{2}\).
    If we pool the observations from different clusters using a bijective mapping \(\mathcal{C}: [n] \to (l, \mathbf{j})\), where \(n\) is the total number of observations, so that \(X_{i} = X_{\mathcal{C}(i)}\), then \(\dist(i,k) = 2\) if \(X_{\mathcal{C}(i)}\) is independent of \(X_{\mathcal{C}(k)}\), and \(\dist(i,k)= 1\) otherwise. We can bound the neighbourhood sizes for \(m=1\) by \(\eta_{m,i} = N_{(j_{1},j_{2})} + \sum_{k\neq j_{2}}^{C_{2}} N_{(j_{1}, k)} + \sum_{k\neq j_{1}}^{C_{1}} N_{(k,j_{2})}\). If further \(\bar{N} = \max_{\mathbf{j}}N_{\mathbf{j}}\), then \(\eta_{m} \leq \pqty{C_{1} + C_{2} - 1} \bar{N}\).
\end{example}
\begin{example}[Random network formation]
    We can also consider the setup in \cite{kojevnikov2020LimitTheorems}, where the network is randomly formed according to a network formation model,
    \begin{equation*}
        \mathcal{G}(i,j) = \indicator{\varphi_{ij} \geq \epsilon_{ij}},
    \end{equation*}
    where \(\varphi_{ij}\) is a random variable that may depend on the observed and unobserved characteristics of \(i\) and \(j\), and \(\epsilon_{ij}\) is independent of \(\varphi_{ij}\). \cite{kojevnikov2020LimitTheorems} measures the sparsity of the network with the conditional maximal expected neighbourhood size
    \begin{equation*}
        \pi_{n} = \max_{i} \sum_{j\neq i} \Prob\bqty{\varphi_{ij}\geq \epsilon_{ij} \mid \varphi_{ij}, i,j\in \mathcal{I}_{n}},
    \end{equation*}
    which replaces the deterministic maximal neighbourhood size \(\eta_{m}\) in the case when the network is random.
    This type of random network formation is relevant in empirical applications. In Section \ref{s2: conditional berbee}, we discuss how our results can be extended to allow for conditioning on the random generation of the network.
\end{example}

\subsection{U-statistics kernel}

In addition to assumptions on the mixing rate and the sparsity of the index space, we also require moment restrictions on the kernel function $H_{n}$. For simplicity, we omit the subscript $n$ when the dependence of $H_{n}$ on the sample size is clear.
Let \(\tilde{X}_{j}\) be an independent copy of \(X_{j}\). We assume the following regularity conditions on the kernel function.
\begin{assumption}\label{asmp:kernel}
    The kernel function \(H(x,y)\) is symmetric \(H(x,y) = H(y,x)\) and for all \(i\neq j \in \mathcal{I}\) and \(n\), \(\Ep H_{n}^{2}(X_{i}, \tilde{X}_{j}) \leq C < \infty\).
\end{assumption}
The requirement of a symmetric kernel can be relaxed by considering \(\frac{1}{2}(H(x,y) + H(y,x))\). The square-integrability assumption ensures that the Hoeffding decomposition is well defined in the non-degenerate case and is important for the CLT for degenerate U-statistics, as in \cite{hall1984CentralLimit}.
We define the following quantities, analogous to Hall's condition \eqref{eq:hall's condition}, with $H_{ij} := H_{n}(X_{i}, X_{j})$ and \(H_{i \tilde{j}} = H(X_{i}, \tilde{X}_{j})\), for each \(i\neq j\),
\begin{equation*}
    \mathbf{H}_{p} =\sup_{i,j \in \mathcal{I}, i \neq j} \cqty{\norm{H_{ij}}_{p}, \norm{H_{i\tilde{j}}}_{p}},\quad\Gamma^{i,j}(x,y) := \Ep H(X_{i}, x)H(X_{j},y),
\end{equation*}
and let \(\tilde{X}_{k_{2}}\) be independent of \(X_{k_{1}}\) and identically distributed as \(X_{k_{2}}\),
\begin{equation*}
    \mathbf{\Gamma}_{m,p} := \sup\cqty{ \norm{{\Gamma^{i,j}}(X_{k_{1}}, \tilde{X}_{k_{2}})}_{p}: i,j,k_{1},k_{2}\in \mathcal{I}_{n}, \, j\in \nb{i}{m}, \text{ and } \dist\pqty{k_{1},k_{2}} > m}.
\end{equation*}
Our bound on the normal approximation error will be a consequence of the interplay between sparsity, mixing rates, and moments.

\section{Theoretical results}\label{sec:theory}\subsection{Non-degenerate case}
For non-degenerate U-statistics with cross-sectionally dependent underlying processes, we use the following Hoeffding decomposition.
% The items in the decomposition below will not be orthogonal as in the classical i.i.d. setting under cross-sectional dependence.
Let \(\hat{H}_{k}(x) = \Ep H(X_{k}, x)\) and \(\theta_{ik} = \Ep \hat{H}_{k}(X_{i}) = \Ep \hat{H}_{i}(X_{k})\), for \(x\) in the support of \(\cqty{X_{i} : i\in \mathcal{I}}\). Then
\begin{equation}\label{eq:hoeffding-depdendent}
    S_{n} = \sum_i \sum_{k \neq i} \theta_{ik} + \hat{S}_{n} + {S}^{*}_{n},
\end{equation}
where, for \(h_{i} = (n-1)^{-1} \sum_{k\neq i} (\hat{H}_{k}(X_{i}) - \theta_{ik})\),
\begin{equation*}
    \hat{S}_{n} = 2 (n-1) \sum_{i} h_{i} \quad \text{and} \quad {S}^{*}_{n} = \sum_i \sum_{k\neq i} \pqty{H(X_{i}, X_{k}) - \hat{H}_{i}(X_{k}) - \hat{H}_{k}(X_{i}) + \theta_{ik}}.
\end{equation*}

It is now straightforward to obtain the following Law of Large Numbers for non-degenerate U-statistics. It generalises existing results for independent observations; see \cite{serfling1980ApproximationTheorems} and Lemma 3.1 in \cite{powell1989SemiparametricEstimation}.
\begin{theorem}\label{thm:lln for non-degenerate}
    If there exists a sequence of positive numbers \(m = m(n) \to \infty\) and for some \(\delta> 0\), \(\tau_{2}^{m} \mathbf{H}_{2}^{2} = o(n^{2})\) and \(\mathbf{H}_{2+\delta}^{2}\beta(1,3,m)^{\delta / (2+\delta)} = o(1)\), then
    \begin{equation*}
        \frac{1}{n(n-1)} S_{n} = \frac{1}{n(n-1)} \sum_i \sum_{k\neq i} \theta_{ik} + o_{p}(1).
    \end{equation*}
\end{theorem}

Stein's method can then be applied to obtain a central limit theorem by showing that the projection part converges in distribution to a Gaussian random variable and that the remainder term is asymptotically negligible.
By constructing a specific Stein's coupling in \autoref{lemma:stein} and using Berbee's Lemma (\autoref{lemma:berbee}), we obtain the following normal approximation result.
\begin{theorem}\label{thm:non-degenerate}
    Under Assumptions \ref{asmp:no-clustering-point}, \ref{asmp:beta-mixing}, and \ref{asmp:kernel}, let \(\nu_{n}^{2} = \Var (\sum_{i} h_{i}) = \sum_{i\in I_n}\sum_{j\in I_n} \Cov(h_i,h_j)\). Then, for any \(\delta > 0\) and \(m \geq 1\), the Wasserstein distance between the normalised non-degenerate U-statistic \(W_{n} = \frac{1}{2(n-1)\nu_{n}} (S_{n} - \sum_i \sum_{k\neq i} \theta_{ik})\) and \(Z \sim N(0,1)\) satisfies
    \begin{align}
        \dist_{W}(W_{n} , Z) &\lesssim  \frac{1}{\nu_{n}n}\bqty{\pqty{\tau_{4}^{m} + \tau_{2,2}^{m}} \mathbf{H}_{2}^{2} + n^{4} \mathbf{H}_{2 + \delta}^{2} \beta(1,3,m)^{\frac{\delta}{2 + \delta}}}^{\frac{1}{2}}
        + \frac{n}{\nu_{n}} \mathbf{H}_{1 + \delta} \beta(1,\infty, m)^{\frac{\delta}{1 + \delta}} \nonumber\\
        &+ \frac{n^{2}}{\nu_{n}^{2}} \mathbf{H}_{2 + \delta}^{2} \beta(1,1,m)^{\frac{\delta}{2 + \delta}} + \frac{\tau_{2}^{m}}{\nu_{n}^{2}} \mathbf{H}_{2 + \delta}^{2} \beta(2,\infty, m)^{\frac{\delta}{2 + \delta}} + \frac{\tau_{3}^{m}}{\nu_{n}^{3}} \mathbf{H}_{3}^{3}.
    \end{align}
\end{theorem}
Similar to the classical central limit theorems for time series and random fields, there is a trade-off between the moment restrictions and the mixing rates. When \(\beta(n_{1},n_{2},m)\) tends to \(0\) faster as \(m\to \infty\), we can take a smaller \(\delta\) and thus relax the moment assumptions. There is also a trade-off, through the choice of \(m\), between the sparsity of the cross-sectional dependence and the mixing rates. When we choose a larger \(m\), the mixing rate \(\beta(n_{1},n_{2},m)\) gets smaller at the cost of a larger \(\tau^{m}_{q_{1},\dots,q_{s}}\).
In the independent case, one may view distinct nodes as having dependence distance \(\mathrm d(i,j)=\infty\). As a result, for any fixed finite \(m\), the mixing terms vanish, \(\eta_m=1\), and \autoref{thm:non-degenerate} yields the usual \(n^{-1/2}\) normal-approximation rate.

As a corollary, if we assume bounded third moments and bounds on the maximum neighbourhood sizes, then the following central limit theorem follows by combining \autoref{thm:non-degenerate} with \autoref{lemma:sparsity-bound} and taking \(\delta = 1\).
\begin{theorem}\label{thm:non-degenerate-clt}
    Under the assumptions of \autoref{thm:non-degenerate}, suppose \(\nu_{n}^{2} = \Theta\pqty{n}\) and \(\mathbf{H}_{3} \leq C < \infty\) for some generic constant \(C\). If, for some sequence \(m = m(n)\) such that \(\eta_{m}^{4} = o\pqty{n}\), and \(n\beta\pqty{1,3,m}^{\frac{1}{3}} + {n}\beta(2, \infty, m) = o\pqty{1}\), then \(W_{n} \rightsquigarrow N(0,1)\) as \(n \to \infty\).
\end{theorem}
We make a few remarks about the theorems.
Firstly, \autoref{thm:non-degenerate} provides a Berry--Esseen type result under \(\beta\)-mixing; see Theorem 2.7 of \cite{dobler2015NewBerryEsseen} for independent observations. Because of \autoref{lemma:stein} and the decomposable approach we use to construct the coupling variables, only bounds on the third moment are necessary, in contrast to \cite{kojevnikov2020LimitTheorems}, which requires moments of order higher than \(4\) and \(\psi\)-dependence.
Secondly, we have used only \(\eta_{m}\), the maximum \(m\)-neighbourhood size, to describe the sparsity. In the literature on random fields and network dependence, summability conditions involving the sizes of ``neighbourhood shells'' \(\mathcal{S}_{i}^{m} = \{j: m < \dist(i,j) \leq m+1\}\), multiplied by the mixing rates, are often assumed \citep{bolthausen1982CentralLimit,kojevnikov2020LimitTheorems}. These conditions usually provide a more precise description of the dependence structure, leading to a weaker set of conditions on the mixing rates, since \(|\mathcal{S}_{i}^{m}| \leq |{\mathcal{N}_{i}^{m}}|\).
% We choose to state the results in terms of \(\eta_{m}\) because, for the more difficult and interesting degenerate U-statistics below, we find it difficult to come up with suitable conditions on the neighbourhood shells that would improve the results.
In any case, we have the freedom to choose a suitable sequence \(m_{n}\), so we do not lose much applicability by stating \autoref{thm:non-degenerate} and \autoref{thm:non-degenerate-clt} in terms of \(\eta_{m}\) with the benefits of simpler notations.

Next, we consider a HAC estimator of the ``long-run'' variance \(\nu_{n}^{2}\). For non-degenerate U-statistics, thanks to the Hoeffding decomposition, this problem is similar to estimating the variance of sample means under cross-sectional dependence. For example, \cite{kojevnikov2020LimitTheorems} and \cite{kojevnikov2021BootstrapNetwork} propose a network HAC estimator for the variance of the sample mean and a bootstrap procedure under network dependence, provided that knowledge of the underlying network is available.

Consider a tapering function \(\kappa: \mathbb R_{+} \to [0,1]\) that is nonincreasing, with \(\kappa(0)=1\) and \(\kappa(z)=0\) for \(z>1\). For some positive sequence \(b_n\), let \(\kappa_{ij}(b_n)=\kappa\!\left(\frac{\mathbf d(i,j)}{b_n}\right)\), \(Q_i := \frac{1}{n-1}\sum_{k\neq i} H_{ik}\), and \(\bar Q := \frac{1}{n}\sum_{i\in I_n} Q_i\). We consider the following HAC estimator:
\[
    \hat \nu_n^2
    :=
    \sum_{i\in I_n}\sum_{j\in I_n}
    \kappa_{ij}(b_n)
    (Q_i-\bar Q)(Q_j-\bar Q).
\]

To show consistency of the estimator, we impose the following additional conditions.
\begin{assumption}\label{asmp:hac-non-degenerate}
    \begin{enumerate*}[label=(\roman*)]
    \item For some \(\delta > 0\), \(\mathbf{H}_{4+\delta} < \infty\).
    \item The sequence of bandwidths \(b_n\) satisfies \(b_n \to \infty\), \(\eta_{b_n}^{4} = o(n)\), and \(n\beta_4(b_n)^{\delta/(4+\delta)} \to 0\), where \(\beta_4(m) := \max_{q_1+q_2=4} \beta(q_1,q_2,m)\).
    \item
        \(\mathbf{H}_{2+\delta}^{2} \sum_{r=1}^{\infty} (1-\kappa\!(\frac{r}{b_n})) (\tau_2^r-\tau_2^{r-1}) \beta(1,1,r-1)^{\delta/(2+\delta)} = o(n)\).
    \item \(\mu_i := \frac{1}{n-1}\sum_{k\neq i} \theta_{ik} =\mu\) for \(i\in I_n\).
    \end{enumerate*}
\end{assumption}
These conditions are similar to those in \cite{kojevnikov2020LimitTheorems} for the network HAC estimator of the variance of the sample mean.
The first two conditions are a standard moment restriction on the kernel function \(H\) and a bandwidth condition requiring the underlying dependence structure to be sufficiently sparse and weak. The third condition is the analogue of Assumption 4.1(ii) in \cite{kojevnikov2020LimitTheorems}, which bounds the contribution to the bias of the HAC estimator from each pair of nodes in a neighbourhood shell \(\{(i,j): m < \dist(i,j) \leq m+1\}\). The last condition is a homogeneity condition on the mean, so that we can eliminate the bias arising from variation in the mean across \(i\in I_n\); this condition holds, for example, when the underlying process is stationary.

\begin{theorem}\label{thm:rowsum-hac-consistency}
    Assume the conditions of \autoref{thm:non-degenerate-clt} and \autoref{asmp:hac-non-degenerate} hold. Then the HAC estimator ${\hat \nu_n^2}/{\nu_n^2} \to_p 1.$
\end{theorem}
Hence, under suitable regularity conditions, we can use the HAC estimator when the underlying distance \(\dist\) is observed.

\subsection{Degenerate case}

For the degenerate case, \(S_{n} = \sum_i \sum_{k\neq i} H_{ik}\) and \(\hat{H}_{k}(x) = \Ep H(X_{k},x) = 0\).
The following theorem follows from the more general result \autoref{thm:degenerate in detail} in \autoref{sec:proofs} and \autoref{lemma:sparsity-bound}, which generalises \cite{hall1984CentralLimit}'s result to the case of cross-sectional dependence for degenerate U-statistics.
\begin{theorem}\label{thm:degenerate}
    Under Assumptions \ref{asmp:no-clustering-point}, \ref{asmp:beta-mixing}, and \ref{asmp:kernel}, for a sequence \(m = m(n)\), let \(\eta_{m} = \max_{i\in \mathcal{I}}\abs{\mathcal{N}_{i}^m}\) and \(\beta(m) = \beta(2,\infty,m)\). Suppose \(s_{n} = \Theta(n)\) and \(\liminf_{n\to\infty}\mathbf H_2>0\). Then, for any \(\delta > 0\),
    \begin{equation}
        \dist_{W}\pqty{W_{n}, N(0,1)} \lesssim  \sqrt{n^{2}\eta_{m}^{2} \mathbf{H}_{4 + \delta}^{4} \beta^{\frac{\delta}{4 + \delta}}(m)} + \sqrt{\frac{\eta_{m}^{7} + \eta_{4m}^{3}}{n} \mathbf{H}_{4}^{4}} + \sqrt{\eta_{m}^{4} \mathbf{\Gamma}_{m,2}^{2}}.
    \end{equation}
    where \(W_{n} = s_{n}^{-1}(S_{n}- \Ep S_{n})\), \(\eta_{4m}\) is the maximum neighbourhood size at radius \(4m\), and \(s_{n}^{2} = \Var\pqty{S_{n}}\).
\end{theorem}
There is also a trade-off between the mixing and moment conditions. If the mixing coefficient converges to zero more rapidly, one can choose a smaller value of \(\delta\), thereby weakening the required moment assumptions. The choice of \(m\) reflects a further trade-off: a larger \(m\) improves the mixing coefficient \(\beta(m)\), but at the cost of increasing the neighbourhood size \(\eta_m\).
In the special case of an independent sample, \(\beta(m)=0\) for all \(m\geq 1\).  Taking \(m=1\) gives \(N_i^m=\{i\}\) and hence \(\eta_m=1\). The bound in Theorem~3.5 then recovers the classical independent-sample benchmark conditions for degenerate U-statistics, as in \cite{hall1984CentralLimit}.

As a corollary, the following CLT holds under a geometric mixing rate and slightly stronger moment conditions than those in \cite{hall1984CentralLimit}. In practice, once we have a specific choice of the kernel function \(H\), the conditions depend on the mixing rate and moment restrictions, which can be checked empirically, as in classical time-series settings. The sparsity restriction can be checked if the distance function is known; otherwise, it may be justified through economic reasoning.
\begin{journalproposition}\label{cor:degenerate}
    Suppose the assumptions in \autoref{thm:degenerate} hold and \(\beta(m) = O\pqty{\rho^{m}}\) for some \(0 < \rho < 1\). If, for some \(C > 0\), \(m = C \log n\), there exists \(K < \infty\) such that \(\eta_{m} = O(m^{K})\), and if, for some \(\epsilon > 0\) and \(\delta > 0\),
    \begin{equation}
        \frac{1}{n} \mathbf{H}_{4}^{4} + \mathbf{\Gamma}_{m,2}^{2} + n^{2 + \mu}\mathbf{H}_{4 + \delta}^{4} = o\pqty{n^{-\epsilon}},
    \end{equation}
    with \(\mu = \frac{C \delta}{4 + \delta}\log\rho\), then we have \({W}_{n} \rightsquigarrow N(0,1)\) as sample size \(n\to \infty\).
\end{journalproposition}

\autoref{thm:degenerate} and the corollary appear to be suitable for index spaces that are relatively ``uniform'', in the sense that the \(\eta_{m,i}\) are of the same order across \(i\in \mathcal{I}\), so that using the maximum cardinalities of neighbourhoods \(\eta_{m}\) in \autoref{lemma:sparsity-bound} is a reasonable choice. For more asymmetric networks, it may be more appropriate to consider the general result and use a more precise counting method in the estimation, as in \autoref{thm:degenerate in detail}.

In practice, we also need to estimate the long-run variance $s_{n}^{2}$ for degenerate U-statistics. In the time-series setting, \cite{fan1999CentralLimit} show that a variance estimator constructed as though the sample were independent can nonetheless be consistent, and demonstrate this in the context of nonparametric specification testing. More generally, bootstrap methods for U-statistics have been extensively studied under independence and time-series dependence. For independent samples, \cite{bickel1981AsymptoticTheory} develop bootstrap theory for U-statistics, while \cite{arcones1992BootstrapStatisticsa} extend these results to degenerate U-statistics. Under weak dependence, \cite{dehling2010CentralLimit} establish bootstrap validity for U-statistics of mixing processes, and \cite{leucht2013DependentWild} propose a dependent wild bootstrap for degenerate U-statistics.

Extending bootstrap methods for U-statistics to settings with cross-sectional dependence is an important direction that we intend to explore in future work, but it seems beyond the scope of this paper. Instead, we provide a complementary result to the central limit theorem for degenerate U-statistics (\autoref{thm:degenerate}) by simplifying the variance expression. This result is used in the proof of \autoref{thm:nonpar test} to show that a variance estimator $\hat{s}_n^2$, constructed as if the sample were independent, remains consistent for the variance of the test statistic under suitable moment and sparsity conditions, in the same spirit as \cite{fan1999CentralLimit}.
\begin{theorem}\label{thm:degenerate u variance}
    Under Assumptions \ref{asmp:no-clustering-point}, \ref{asmp:beta-mixing} and \ref{asmp:kernel}, if we have
    \begin{equation*}
        \tau_{4}^{m} \mathbf{H}_{2}^{2} + \pqty{ \tau_{3,1}^m + \tau_{2,1,1}^m + \tau_{1,1,1,1}^m + \tau_{2,2}^{m}} \mathbf{H}_{2+\delta}^{2}\beta_{4}(m)^{\frac{\delta}{2 + \delta}} + \tau_{2,2}^{m} \gamma_{m,1} = o(\sigma_{n}^{2}),
    \end{equation*}
    where \(\beta_{4}(m) := \max_{q_{1}+q_{2}=4}\beta(q_{1},q_{2},m)\) and $$\gamma_{m,1} := \sup\cqty{\abs{\Ep \Gamma^{i,j}_{k_{1},k_{2}}} : i,j,k_{1},k_{2}\in \mathcal{I}_{n},\ j\in \nb{i}{m},\ k_{1}\neq k_{2}},$$
    then \(s_{n}^{2} = \Var\pqty{S_{n}} = (2 + o(1))\sigma_{n}^{2}\), where \(\sigma_{n}^{2} = \sum_i \sum_{k\neq i} \Ep H_{i \tilde{k}}^{2}\).
\end{theorem}

\section{Application: nonparametric specification test}\label{s1:nonpar test}In statistics and econometrics, there has always been a trade-off between model complexity and statistical inference. For a statistical model \(\cqty{\Prob_{\gamma}: \gamma\in \Gamma}\), where \(\Prob_{\gamma}\) denotes the probability law that governs the data generating process indexed by a parameter \(\gamma\) in the parameter space \(\Gamma\), we face the risk of model misspecification if we estimate a restricted model \({\Gamma}_0\), while the true parameter lies outside the models we consider, \(\gamma_{0} \notin {\Gamma}_0\). On the other hand, with a more complex model, we may lose identification power and efficiency. This trade-off necessitates close examination of how we specify our models and calls for the development of specification tests.  We will focus on the test for the specification of a regression model, \(Y_{i} = g(Z_{i}) + u_{i}\), where we observe \(\cqty{(Z_{i}, Y_{i}):  i\in \mathcal{I}}\) allowing for cross-sectional dependence among the observations.

A vast literature exists on such specification tests. Bearing in mind the trade-off mentioned, this paper will consider the specification test where we have the null hypothesis that a regression function lies in a parametric family indexed by some finite-dimensional parameter against a general nonparametric alternative, i.e. \(H_{0}: g(z) = g(z,\gamma)\) for a known parametric function indexed by a parameter \(\gamma \in \Gamma\subset \mathbb{R}^{d}\) against the alternative that \(g(z)\) is a smooth nonparametric function. Several such tests have been proposed based on the kernel smoothing method; for example, \cite{hardle1993ComparingNonparametric} compares the \(L_{2}\) distance between the parametric and nonparametric fit, while \cite{zheng1996ConsistentTest}, \cite{fan1999CentralLimit}, and \cite{li2007NonparametricEconometrics} propose tests based on the idea that if the parametric model is correctly specified, then a kernel smoothing fit for the residuals should be approximately zero. Such tests avoid the random denominator problem. Additionally, \cite{fan1999CentralLimit} dealt with absolutely regular time series data, which was later generalised to broader mixing concepts by \cite{gao2008CentralLimit} and others. Finally, \cite{horowitz2001AdaptiveRateOptimal} and \cite{horowitz2002AdaptiveRateOptimal} propose adaptive and rate optimal tests that are uniformly consistent against alternatives by considering multiple bandwidths.

For these test statistics, the dominating parts are second-order U-statistics and the Central Limit Theorems for the test statistics are derived under conditions of \cite{hall1984CentralLimit}, \cite{dejong1987CentralLimit}, and their extensions to the time series settings. With the normal approximation theory developed in the previous section, we are now able to generalise the nonparametric specification test to allow for cross-sectional dependence.

\subsection{Test statistics}
Suppose we have a collection of cross-sectional observations \(\cqty{(Y_{i}, Z_{i}) : i\in \mathcal{I}}\), where \(\mathcal{I}\) is the index set and \(n =\abs{\mathcal{I}}\) is the sample size, \(Y_{i}\) is a random scalar, and \(Z_{i}\in \mathbb{R}^{d}\). Suppose for \(i \in \mathcal{I}\), \(\Ep(Y_{i}\mid Z_{i}) = g(Z_{i})\).
We wish to test the null hypothesis of a parametric model specification for \(g(\cdot, \gamma)\) where \(\gamma\in \Gamma_0\) against a nonparametric alternative,
\begin{align*}
    &\mathbb{H}_{0}: \Ep\bqty{Y_{i} \mid Z_{i}} = g(Z_{i} , \gamma_{0}) \mathtext{a.s. for some} \gamma_{0} \in \Gamma_0, \\
    &\mathbb{H}_{1}: \Prob\bqty{\Ep\pqty{Y_{i}\mid Z_{i}} \neq g(Z_{i}, \gamma)} > 0 \mathtext{for all} \gamma\in \Gamma_0.
\end{align*}
Let \(u_{i} = Y_{i} - g(Z_i, \gamma_0)\), we will allow for cross-sectional dependence in \(X_{i} = \pqty{u_{i}, Z_{i}}\) under the null hypothesis \(H_{0}\).
Let \(\hat{\gamma}\) be the nonlinear regression estimate for \(\gamma\) which is assumed to be \(\sqrt{n}\)-consistent, and \(\hat{u}_{i} = Y_{i} - g\pqty{Z_{i}, \hat{\gamma}}\).
Following \cite{fan1996ConsistentModel} and \cite{li2007NonparametricEconometrics}, we base our test on the kernel estimator for \(\Ep\bqty{u_{i}\Ep\bqty{u_{i}\mid Z_{i}} f_{i}(Z_{i})}\), which equals zero under the null hypothesis \(H_{0}\), where \(f_{i}(\cdot)\) is the marginal density function of \(Z_{i}\). Consider the following statistic,
\begin{equation}\label{eq:test statistic base}
    I_{n} = \sum_{i}\sum_{j\neq i} b^{- \frac{d}{2}} \hat{u}_{i} \hat{u}_{j} K\pqty{\frac{Z_{i} - Z_{j}}{b}},
\end{equation}
where \(b\) is a bandwidth we choose. We will write \(K_{ij} = K\pqty{\frac{Z_{i} - Z_{j}}{b}}\) for a kernel weighting function \(K\). For later use, let $I_{n}^{\circ} = \sum_{i}\sum_{j\neq i} b^{- \frac{d}{2}} u_{i}u_{j} K\pqty{\frac{Z_{i} - Z_{j}}{b}}$. This oracle statistic \(I_{n}^{\circ}\) is a U-statistic of order \(2\) with samples \(\cqty{\pqty{u_i, Z_i} : i \in \mathcal{I}}\) and kernel
\begin{equation*}
    H(x,y) = b^{- \frac{d}{2}} x_{1}y_{1} K\pqty{\frac{x_{2} - y_{2}}{b}},
\end{equation*}
where \(x = (x_{1}, x_{2}')'\) and \(y = (y_{1}, y_{2}')'\in \mathbb{R}^{1}\times \mathbb{R}^{d}\). The feasible statistic \(I_{n}\) in \eqref{eq:test statistic base} is the plug-in version obtained by replacing \(u_i\) with \(\hat u_i\).
Let \(\hat{s}_{n}^{2} = 2 b^{- d} \sum_{i}\sum_{j\neq i} \hat{u}^{2}_{i}\hat{u}^{2}_{j} K_{ij}^{2}\), we will consider the following test statistic
\begin{equation}\label{eq:test stats}
    \mathbf{T}_{n} = \frac{I_{n}}{\hat{s}_{n}}.
\end{equation}

We make the following assumptions on the data generating process. Firstly, we allow for cross-sectional dependence in the sample but require the sample to be \(\beta\)-mixing.
\begin{assumption}\label{asmp:nonpar dgp}
    We observe a sample \(\cqty{\pqty{Y_{i}, Z_{i}}: i\in \mathcal{I}_n}\) which is \(\beta\)-mixing with coefficients \(\beta(n_{1},n_{2},m)\) with respect to some distance \(\dist:\mathcal{I}_{n} \times \mathcal{I}_{n} \to [0, \infty]\). For \(m > 0\), let \(\nb{i}{m}\) be the \(m\)-neighbourhood of \(i\), \(\nb{i}{m} = \cqty{j : \dist(i,j) \leq m}\), and \(\eta_{m} = \max_{i\in \mathcal{I}} \abs{\nb{i}{m}} = O(m^{C_{1}})\) for some \(0 \leq C_{1} < \infty\).
\end{assumption}

And we make the following assumptions on the moments of the error terms, and the smoothness of the density function of \(Z_{i}\)'s and the regression function.
\begin{assumption}\label{asmp:nonpar moments}
    Assume that \(\Ep\bqty{u_{i}\mid Z_{i}} = 0\), and there exist constants \(0 < \underline{\mu} \leq \overline{\mu} < \infty\) such that
    \[
        \underline{\mu}
        \leq
        \mu_{i, 2}(z) := \Ep\bqty{u_{i}^{2} \mid Z_{i} = z}
        \leq
        \overline{\mu}
    \]
    for all \(z\) in the support of \(Z_{i}\) and all \(i \in \mathcal{I}_n\). For all \(i\neq j \in \mathcal{I}\), \(\norm{u_{i}^{4}}_{\frac{1}{1-\xi}}\) and \(\norm{u_{i}^{4}u_{j}^{4}}_{\frac{1}{1-\xi}}\) are bounded for some \(0 < \xi < 1\), and \(\norm{u_{i}u_{j}u_{k}u_{l}}_{\frac{1}{1 - \xi_1}} < \infty\) for some \(\xi_1\in (\frac{1}{2}, 1)\) and \(i,j,k,l \in \mathcal{I}\).
\end{assumption}
\begin{assumption}\label{asmp:z}
    \begin{enumerate}[(a)]
        \item Let \(f_{i_{1}, \dots, i_{l}}\) be the density of \((Z_{i_{1}}, Z_{i_{2}}, \dots, Z_{i_{l}})\). Then there exist \(\mathcal{Z}\subset \mathbb{R}^{d}\) with positive measure and a constant \(\underline{f} > 0\) such that \(f_{i}(z) \geq \underline{f}\) for all \(z \in \mathcal{Z}\) and all \(i \in \mathcal{I}_{n}\). For all \(i_{1}, \dots, i_{l} \in \mathcal{I}_{n}\), \(l\leq 4\), the densities \(f_{i_{1},\dots, i_{l}}\) are bounded and satisfy the following Lipschitz condition, \(\abs{f_{i_{1}, \dots, i_{l}}(z_{1} + v_{1}, \dots, z_{l} + v_{l}) - f_{i_{1}, \dots, i_{l}}(z_{1}, \dots, z_{l})} < D_{l}(z_{1}, \dots, z_{l}) \abs{v},\)
            where \(D_{l}(z_{1}, \dots, z_{l})\) are integrable, \(\int D_{l}(z_{1}, \dots, z_{l}) f_{i_{1}, \dots, i_{l}}(z_{1}, \dots, z_{l}) < \infty\). For all \(i\neq j \in \mathcal{I}_n\), the densities of \(Z_i-Z_j\) are uniformly bounded.
        \item
            \(g(z,\cdot)\) is twice differentiable and the norms of \(\frac{\partial}{\partial \gamma} {g(z,\cdot)}\) and \(\frac{\partial^{2}}{\partial \gamma \partial \gamma'}g(z, \cdot)\) are bounded by a Lipschitz continuous function \(M_{g}(z)\) with finite second moment in a neighbourhood of \(\gamma_{0}\), such that \(\int M_{g}^4(z) f_{i}(z)\dd z < \infty, \int M_{g}^4(z) f_{i,j}(z,z)\dd z < \infty,\int M_{g}^4(z) D_{l}(z, \dots, z)\dd z < \infty,\) for all \(i,j \in \mathcal{I}_{n}\), \(l \leq 4\).
    \end{enumerate}
\end{assumption}
This assumption that the supports of \(Z_{i}\)'s are not disjoint will be satisfied if we assume stationarity.
The following assumptions on the kernel smoothing function \(K\) are standard in the literature and satisfied for commonly chosen kernel smoothing functions, such as normal kernel, Epanechnikov kernel, and uniform kernel.
\begin{assumption}\label{asmp:nonpar kernel}
    \(K(v)\) is a product kernel smoothing function \(K(v) = \prod_{k=1}^{d}\tilde{K}(v_{(k)})\) where \(v_{(k)}\) denotes the \(k\)-th element of the \(d\)-dimensional vector \(v\) and \(\tilde{K}\) is a nonnegative, bounded, symmetric function \(\tilde{K}\) with \(\int \tilde{K}(v) = 1\) and \(\int \abs{v}^{2}K(v) \dd v < \infty\).
\end{assumption}

Our main result is the following distribution theory for the test statistics.
\begin{theorem}\label{thm:nonpar test}
    Suppose Assumptions \ref{asmp:nonpar dgp}--\ref{asmp:nonpar kernel} hold, and that \(\sqrt{n}\pqty{\hat{\gamma} - \gamma_{0}} = O_{p}(1)\) under \(\mathbb H_{0}\). Suppose also that \(\beta(m) = O(\rho^{m})\) for some \(0 < \rho < 1\), and the bandwidth \(b^{-d} = \Theta(n^{\alpha})\) for some \(0 < \alpha < \frac{1}{2-\xi}\), then
    \begin{equation}\label{eq:nonpar test}
        \mathbf{T}_{n} \rightsquigarrow_{H_{0}} N(0,1),
    \end{equation}
    as the sample size \(n\to \infty\).
\end{theorem}

\subsection{Simulations}
In this section, we conduct simulations to show the finite-sample performance of the proposed specification test with cross-sectionally dependent data.\footnote{The codes for implementing the nonparametric specification test and the simulations are available at \href{https://github.com/lwg342/specification-test-csc-code}{https://github.com/lwg342/specification-test-csc-code}.} The simulation specification is close to the one in \cite{horowitz2001AdaptiveRateOptimal}. For each simulation we generate \(n = 2000\) samples of \((z_{i}, y_{i})\), where \(z_{i} \sim N\pqty{0,25}\) and the null-hypothesis model is
\begin{equation*}
    y_{i} = \gamma_{0} + \gamma_{1}z_{i} + u_{i},
\end{equation*}
where we take the true \(\gamma_{0} = \gamma_{1} = 1\), and the alternative model is
\begin{equation*}
    y_{i} = 1 + z_{i} + \psi\cdot{\frac{5}{\tau}} \phi\pqty{\frac{z_{i}}{\tau}} + u_{i},
\end{equation*}
where \(\phi\) is the standard normal density function and \(\tau\) is the bandwidth parameter where \(\tau = 1\) or \(0.25\), it controls the shape of deviation from the null model with a smaller \(\tau\) implying a more ``spiked'' deviation. The \(\psi\) parameter controls the overall magnitude of deviation of the true model from the null models. The error terms \(u_{i}\) are generated according to three types of models. The first is the benchmark i.i.d. model, where \(u_{i} \sim N\pqty{0,1}\). The second is an AR(1) model with \(u_{i} = \rho u_{i - 1} + \upsilon_{i}\) where \(\upsilon_{i} \sim N\pqty{0,1}\). The third is a two-way clustering structure. Assume there is a bijective mapping \(\pi: [n] \to [n_{1}]\times [n_{2}]\), such that \(i,j\) are independent if \(\pi(i) = \pqty{\pi_{1}(i), \pi_{2}(i)}\) and \(\pi(j)\) share no common coordinates. In particular the model we consider is \(u_{i} = \frac{1}{\sqrt{2}} \bqty{\lambda_{\pi_{1}(i)} + F_{\pi_{2}(i)}}\) where each \(\lambda, F\) are independent standard normal random variables. For the two-way clustering design, we use \((n_{1}, n_{2}) = (40, 50)\) and \((100, 20)\), so that the total sample size remains \(n = 2000\) in each case.

For each simulated dataset \(\pqty{z_{i}, y_{i}}\), we compute the residual \(\hat{u}_{i}\) from the linear regression of \(y_{i}\) on \(z_{i}\). We then compute the test statistic \(\mathbf{T}_{n}\) using the residuals \(\hat{u}_{i}\), with fixed bandwidth \(b_{n} = n^{-1/4}\), and reject the null when \(\mathbf{T}_{n} > 1.645\). We repeat the simulation 2000 times and compute the empirical size and power of the test. The results are shown in Table \ref{tab:spec test}.

% Generated by python3 u-stats-sim/spec_test_main_pdf.py --n 2000 --nsim 2000 --seed 31415 --bandwidth-const 0.1 --fixed-bandwidth 0.14953487812212204 --rejection-rule upper --two-way-scenarios 40x50,100x20 --critical-value 1.645 --jobs 10
\begin{table}[tbp]
    \centering
    \setlength{\tabcolsep}{4pt}
    \resizebox{\textwidth}{!}{%
        \begin{tabular}{llccccc}
            \toprule
            & & \multicolumn{1}{c}{Null model} & \multicolumn{4}{c}{Alternative model}\\
            \cmidrule(lr){3-3} \cmidrule(lr){4-5} \cmidrule(lr){6-7}
            & &  & \multicolumn{2}{c}{ $\tau = 0.25$} & \multicolumn{2}{c}{ $\tau = 1$} \\
            \cmidrule(lr){4-5} \cmidrule(lr){6-7}
            Model of $\epsilon$ & Parameters &  & \(\psi\) = 0.2 & \(\psi\) = 0.1 & \(\psi\) = 0.2 & \(\psi\) = 0.1 \\
            \midrule
            Normal & N(0,1) & 0.046 & 1.0 & 0.949 & 0.837 & 0.209 \\
            \addlinespace
            \multirow{3}{*}{AR} & $\rho = 0.1$ & 0.046 & 1.0 & 0.940 & 0.847 & 0.208 \\
            & $\rho = 0.5$ & 0.049 & 1.0 & 0.847 & 0.690 & 0.157 \\
            & $\rho = 0.9$ & 0.047 & 0.860 & 0.224 & 0.163 & 0.067 \\
            \addlinespace
            \multirow{2}{*}{Two-way clustering} & $n_1 = 40$, $n_2 = 50$ & 0.044 & 1.0 & 0.949 & 0.856 & 0.211 \\
            & $n_1 = 100$, $n_2 = 20$ & 0.058 & 1.0 & 0.953 & 0.839 & 0.232 \\
            \bottomrule
        \end{tabular}
    }
    \caption{Empirical rejection probability of the specification test with cross-sectionally dependent data.}
    \label{tab:spec test}
\end{table}

It can be seen from the table that when the null hypothesis \(\mathbb{H}_{0}\) is true, the empirical probability of the test statistic rejecting the null hypothesis is close to the nominal level of 0.05, even when we allow for cross-sectional dependence. When the alternative hypothesis \(\mathbb{H}_{1}\) is true, the power of the test statistic depends on the deviation from the null as well as the cross-sectional dependence structure. In general when the deviation from the null is large, then the power of the test is reasonably good. When the deviation gets closer to being local, the effect of cross-sectional dependence is more noticeable. With stronger cross-sectional dependence, the power generally decreases, as we use the consistent variance estimator in \autoref*{thm:degenerate u variance} which requires weak dependence and sparsity. In this case, perhaps we could consider more sophisticated variance estimators that take into account the cross-sectional dependence structure.

\subsection{Empirical application: Airbnb pricing}

We apply the nonparametric specification test to a dataset on Airbnb listings in Dublin to test whether the log of listing prices is linear in the log-distance to Dublin city centre, conditional on a set of listing and host characteristics.
Airbnb listings are likely to exhibit local spatial dependence because nearby listings may share unobserved neighbourhood quality, common demand shocks, and local amenities.
The data are obtained from Inside Airbnb\footnote{The data can be downloaded from \url{https://insideairbnb.com}} and are based on a Dublin extract that contains only the quarterly listings snapshot scraped on September 16, 2025.
We restrict attention to listings classified as ``Entire home or apartment.'' After retaining observations with positive prices, the sample contains 2958 listings.
Excluding observations with missing regression controls yields a final testing sample of 2494 listings.

The outcome variable is \(Y_i=\log(\text{price}_i)\), and the main regressor of interest is the log-distance to the city centre, as in the literature \citep[for example,][]{combes2019CostsAgglomeration,gupta2022FlatteningCurve}.
Specifically, we define \(Z_i=\log(d_i+0.25)\), where \(d_i\) denotes the straight-line distance, measured in kilometers, from listing \(i\) to a central Dublin reference point at \((53.35^\circ\text{N},\; 6.26^\circ\text{W})\), which is the location of the Spire of Dublin.
We add the small constant $0.25$, following \cite{gupta2022FlatteningCurve}, to regularise the transformation for listings located very close to the city centre.
Robustness checks with alternative distance transformations are reported in the Appendix, and the main conclusions are not sensitive to this choice.

The null hypothesis is \(H_0:\ \mathbb{E}[Y_i\mid Z_i,W_i]=\alpha+\beta Z_i+W_i'\delta\), where \(W_i\) is a vector of observable listing and host characteristics.
We estimate the null model by OLS and then compute the proposed test statistic with Epanechnikov kernel over \(Z_i\). As a benchmark choice, we use the rule-of-thumb bandwidth \(b_0=1.06\,\hat{\sigma}(Z)\,n^{-1/5}\), which equals \(b_0=0.212\) in our sample.

Table~\ref{tab:dublin_main} reports the OLS estimation results (Panel A) and the specification test statistics (Panel B). In Panel B, we report the benchmark test statistic computed as in Eq.~\eqref{eq:test stats} together with its asymptotic $p$-value under the simple variance estimator. Because Airbnb listings are likely to be spatially dependent, we also report results from dependence-robust bootstrap procedures. We implemented a neighbourhood block wild bootstrap, grid block wild bootstraps, and \emph{Leucht-style} spatial dependent multiplier bootstraps evaluated at several dependence scales.
The neighbourhood and grid procedures are in the spirit of wild and block bootstrap methods for heteroskedastic and dependent data \citep{mammen1993BootstrapWild,zhu2007BootstrappingEmpirical,cameron2011RobustInference}. The multiplier procedure is instead motivated by \mbox{\cite{leucht2013DependentWild}}, although our spatial implementation is an adaptation to cross-sectional proxy distance rather than a literal application of their theorem.
The residual dependence diagnostics reported in the Appendix indicate that dependence is local. In the Dublin dataset, the mean residual cross-product within 1 km is about 0.0055, while beyond 2 km it is essentially zero. This pattern informs our choice of the hyperparameters for the bootstrap procedures.

First, we compute a neighbourhood block wild bootstrap. Let \(g(i)\) denote the neighbourhood block of listing \(i\), based on the variable \textit{neighbourhood}. In each bootstrap replication, we draw a Rademacher multiplier \(\xi_{g(i)}\in\{-1,+1\}\), and the bootstrap outcome is \(Y_i^*=\hat m_0(X_i)+\xi_{g(i)}\hat u_i\), where \(\hat m_0(X_i)\) is the fitted value under the null. We then refit the same null model using \(Y_i^*\), obtain bootstrap residuals \(\hat u_i^*\), and recompute the studentized bootstrap statistic. The bootstrap \(p\)-value is the empirical upper-tail probability of \(T_n^*(b)\) relative to the observed \(T_n(b)\). In the Dublin file, however, the variable \textit{neighbourhood} yields only four non-missing blocks, so this procedure should be interpreted only as a coarse diagnostic.

Second, we compute finer grid block wild bootstraps. The bootstrap construction is identical, namely \(Y_i^*=\hat m_0(X_i)+\xi_{g(i)}\hat u_i\), but the block indicator \(g(i)\) is now defined by projected square grid cells. This gives a much finer partition of space: the 0.5 km grid yields 713 blocks and the 1.0 km grid yields 359 blocks.

Third, we report a \emph{Leucht-style} spatial dependent multiplier bootstrap. Unlike the two block bootstraps, this procedure does not resample the outcomes. Instead, it applies listing-level multipliers \(\xi_i\) whose dependence decays with projected geographic distance, with the decay controlled by the scale parameter \(\ell\). The bootstrap statistic is computed analogously to the original test statistic, but with \(I_n^{*}(b) = \sum_i \sum_{k\neq i}b^{-1/2}\hat u_i\hat u_k K\!\left(\frac{Z_i-Z_k}{b}\right)\xi_i\xi_k\), where the multiplier field is generated as follows. Let \(d(i,j)\) denote projected geographic distance, let \(\eta_j \overset{iid}{\sim} N(0,1)\), and define
\[
    \xi_i=\sum_{j\in \mathcal N_i(\ell)} w_{ij}(\ell)\eta_j,
    \qquad
    w_{ij}(\ell)=
    \frac{\rho\!\left(d(i,j)/\ell\right)}
    {\left(\sum_{m\in \mathcal N_i(\ell)}\rho\!\left(d(i,m)/\ell\right)^2\right)^{1/2}},
\]
where \(\mathcal N_i(\ell)=\{j:\rho(d(i,j)/\ell)>0\}\). In the reported calculations, we use the Bartlett decay function \(\rho(u)=\max\{1-u,0\}\), so only listings within distance \(\ell\) receive positive weight. This creates a spatially structured perturbation field. We report the results for \(\ell\in\{0.1,0.25,0.5,1.0\}\) km in \autoref{tab:dublin_main}.

Table~\ref{tab:dublin_main} reports the OLS estimates under the null hypothesis together with the corresponding specification test results. In the Dublin dataset, both the asymptotic test and all dependence-robust bootstrap procedures yield \(p\)-values well above conventional significance levels.
For Dublin listings, these findings suggest that a linear model provides an adequate approximation to the relationship between price and location in logarithms. The non-rejection is consistent across all dependence-robust procedures.
Additional robustness checks, reported in the Appendix, corroborate the same qualitative conclusion.

\begin{table}[!htbp]
    \small
    \centering
    \caption{Airbnb pricing in Dublin: OLS estimation and specification test.}
    \label{tab:dublin_main}
    \begin{tabular}{lrrrr}
        \hline
        Variable & Estimate & Std. err. & $t$-stat & $p$-value \\
        \hline
        \multicolumn{5}{l}{\rule{0pt}{2.6ex}\textbf{Panel A. OLS regression under $H_0$}} \\
        \multicolumn{5}{l}{\textit{Property characteristics}} \\
        \hspace{1em}Accommodates (log) & 0.750 & 0.038 & 19.983 & 0.000 \\
        \hspace{1em}Bedrooms & 0.082 & 0.016 & 5.277 & 0.000 \\
        \hspace{1em}Bathrooms & 0.065 & 0.013 & 4.878 & 0.000 \\
        \hspace{1em}Beds & $-0.020$ & 0.008 & $-2.465$ & 0.014 \\
        \hspace{1em}Minimum nights (log) & $-0.121$ & 0.012 & $-9.895$ & 0.000 \\
        \multicolumn{5}{l}{\textit{Host characteristics}} \\
        \hspace{1em}Superhost & 0.078 & 0.020 & 3.982 & 0.000 \\
        \hspace{1em}Identity verified & $-0.033$ & 0.049 & $-0.681$ & 0.496 \\
        \hspace{1em}Instant bookable & 0.007 & 0.021 & 0.311 & 0.756 \\
        \hspace{1em}Response rate & 0.002 & 0.0005 & 4.576 & 0.000 \\
        \hspace{1em}Acceptance rate & $-0.001$ & 0.0004 & $-1.758$ & 0.079 \\
        \multicolumn{5}{l}{\textit{Reviews}} \\
        \hspace{1em}Number of reviews (log) & $-0.029$ & 0.008 & $-3.782$ & 0.000 \\
        \hspace{1em}Reviews per month & $-0.043$ & 0.005 & $-9.139$ & 0.000 \\
        \hspace{1em}Rating & 0.125 & 0.026 & 4.851 & 0.000 \\
        \hspace{1em}Rating missing & 0.043 & 0.031 & 1.382 & 0.167 \\
        \multicolumn{5}{l}{\textit{Other}} \\
        \hspace{1em}Listings count (log) & $-0.021$ & 0.008 & $-2.673$ & 0.008 \\
        \hspace{1em}$Z$: Log-distance to city centre & $-0.151$ & 0.010 & $-15.555$ & 0.000 \\
        \hline
        \multicolumn{5}{l}{\rule{0pt}{2.8ex}\textbf{Panel B. Specification test}} \\
        % \multicolumn{5}{l}{\hspace{1em}Observed test statistic: } \\
        \multicolumn{2}{l}{Inference method} & \multicolumn{2}{l}{Test Statistic / Hyperparameters} & $p$-value \\
        \hline
        \multicolumn{2}{l}{\hspace{1em}Test Statistic} & \multicolumn{2}{l}{$T_n = 0.855$} & 0.196 \\
        \multicolumn{2}{l}{\hspace{1em}Neighbourhood block bootstrap} & \multicolumn{2}{l}{4 neighbourhoods} & 0.498 \\
        \multicolumn{2}{l}{\hspace{1em}Grid block bootstrap} & \multicolumn{2}{l}{Grid width 0.5 km} & 0.552 \\
        \multicolumn{2}{l}{\hspace{1em}Grid block bootstrap} & \multicolumn{2}{l}{Grid width 1.0 km} & 0.744 \\
        \multicolumn{2}{l}{\hspace{1em}Leucht-style multiplier bootstrap} & \multicolumn{2}{l}{$\ell = 0.1$ km} & 0.416 \\
        \multicolumn{2}{l}{\hspace{1em}Leucht-style multiplier bootstrap} & \multicolumn{2}{l}{$\ell = 0.25$ km} & 0.676 \\
        \multicolumn{2}{l}{\hspace{1em}Leucht-style multiplier bootstrap} & \multicolumn{2}{l}{$\ell = 0.5$ km} & 0.846 \\
        \multicolumn{2}{l}{\hspace{1em}Leucht-style multiplier bootstrap} & \multicolumn{2}{l}{$\ell = 1.0$ km} & 0.928 \\
        \hline
    \end{tabular}
    \vspace{0.2em}

    \parbox{0.92\linewidth}{\footnotesize
    Notes: The sample consists of 2494 Dublin Airbnb listings classified as ``Entire home or apartment'' with positive prices and non-missing controls. In Panel A, the dependent variable is \(Y_i=\log(\text{price}_i)\), and the regressor of interest is \(Z_i=\log(d_i+0.25)\), where \(d_i\) is the straight-line distance in kilometers from the Spire of Dublin. Panel A reports OLS estimates under the null linear specification. The matched Dublin regression includes property-type fixed effects (15 estimated dummies, omitted from the table). The OLS regression has \(R^2=0.5662\). Panel B reports the observed value of the specification test statistic together with the asymptotic \(p\)-value based on the simple variance estimator. All reported bootstrap \(p\)-values use the studentized statistic \(T_n\).}
\end{table}

\section{Conclusion}\label{s1:conclusion}
This paper develops normal approximation results for second-order U-statistics when the underlying observations exhibit cross-sectional dependence. Using Stein's coupling method together with a decomposable argument, we obtain Wasserstein bounds for both non-degenerate and degenerate U-statistics. The results extend the classical theory of U-statistics beyond independent and time-series settings, and also complement recent limit theorems for sums of cross-sectionally or network-dependent random variables.

The main difficulty is that U-statistics combine two sources of dependence: the dependence generated by the nonlinear kernel and the dependence already present in the underlying sample. Our results show that asymptotic normality can still be obtained under conditions on the mixing rate, the sparsity of the dependence structure, and the moments of the kernel. For non-degenerate U-statistics, the projection component is dominant and feasible inference can be conducted using a HAC-type variance estimator when the underlying distance is observed. For degenerate U-statistics, we provide sufficient conditions under which normal approximation remains valid and discuss variance estimation in the spirit of existing specification-testing procedures.

We also illustrate the usefulness of the theory through a nonparametric specification test that allows for cross-sectional dependence. The empirical application to Airbnb pricing demonstrates how the proposed framework can be used in settings where spatial or network dependence is a natural concern.

Several directions remain for future research. First, it would be useful to develop uniform laws of large numbers for U-statistics under cross-sectional dependence. Such results would be relevant for estimators based on pairwise comparisons, such as the pairwise-difference estimators studied by \cite{honore1997PairwiseDifference}. Second, bootstrap inference for U-statistics under cross-sectional dependence is an important open direction. While bootstrap methods are well understood in classical settings for non-degenerate U-statistics and sample averages, and network block bootstrap methods have been proposed for sample means under network dependence by \cite{kojevnikov2021BootstrapNetwork}, the bootstrap theory for degenerate U-statistics under general cross-sectional dependence appears substantially more challenging. We leave this problem for future work.

\section*{Declaration of competing interest}
The authors declare the following financial interests/personal relationships which may be considered as potential competing interests:

Weiguang Liu reports financial support was provided by UK Research and Innovation (UKRI). If there are other authors, they declare that they have no known competing financial interests or personal relationships that could have appeared to influence the work reported in this paper.

\section*{Supplementary material}
Supplementary material associated with this article can be found in the online version at \href{https://doi.org/10.1016/j.jeconom.2026.106325}{10.1016/j.jeconom.2026.106325}.

{\small\setlength{\bibsep}{4pt plus 1pt}
}
\clearpage
\appendix
% Preserve the printed supplement's numbering and target the corresponding
% result in the published article when a reader follows the link.
\setcounter{footnote}{0}
\makeatletter
\newcommand{\supplementnumber}[5]{%
  \expandafter\def\csname r@thm:degenerate u variance\endcsname{{3.7}{#2}{#3}{#4}{#5}}}
\ifcsname r@thm:degenerate u variance\endcsname
  \expandafter\expandafter\expandafter\supplementnumber\csname r@thm:degenerate u variance\endcsname
\fi
\makeatother
\begin{center}
  \Large\textbf{Supplementary Material for ``Normal Approximation for U-Statistics with Cross-Sectional Dependence''}
\end{center}

\section{Proofs}\label{sec:proofs}
\subsection{Lemmas}The following lemma, based on Stein's Coupling \citep{chen2010SteinCouplings} and the decomposable methods \citep{barbour1989CentralLimit}, is a variant of the Stein's method and the key to our analysis.
\begin{lemma}\label{lemma:stein}
    Given a random variable \(W\in L^{2}\pqty{\Prob}\), suppose there exist square-integrable random variables \((G_{i},D_{ij}, D_{ij}')\) for \(i \in \mathcal{I}\) and \(j \in \mathcal{J}_{i} \subset \mathcal{I}\), such that \(\sum_{i\in \mathcal{I}} G_{i} = W\), then we have the following bound on the Wasserstein distance between \(\correction{W}\) and \(N(0,1)\):
    \begin{equation}
        \dist_{W}(\correction{W}, N(0,1)) \leq  2 A_{1} + \sqrt{\frac{2}{\pi}} A_{2} +  \sqrt{\frac{2}{\pi}}A_{3} +  2A_{4} + A_{5},
    \end{equation}
    where
    \begin{align*}
        A_{1} &= \Ep\sum_{i}\abs{\Ep\bqty{ G_{i}\mid W - D_{i}}},
        &A_{2} &= \Ep\abs{1 - \sum_{i}\sum_{j \in \mathcal{J}_{i}} M_{ij}},\\
        A_{3} &= \Ep\sum_{i}\sum_{j \in \mathcal{J}_{i}}\abs{\Ep \pqty{M_{ij}- G_{i}D_{ij} \mid W - D_{ij}'}},
        &A_{4} &= \Ep\sum_{i}\sum_{j \in \mathcal{J}_{i}} \abs{\pqty{M_{ij} - G_{i}D_{ij}}D_{ij}'}, \\
        A_{5} &= \sum_{i} \Ep\abs{G_{i}D_{i}^{2}},
    \end{align*}
    with \(D_{i} = \sum_{j\in \mathcal{J}_{i}} D_{ij}\), and \(M_{ij} = \Ep \bqty{G_{i} D_{ij}}\).
\end{lemma}
This lemma reduces the problem of estimating the Wasserstein distance to that of estimating moments as soon as we have constructed the coupling variables \((G_{i}, D_{ij}, D_{ij}')\). It remains to construct suitable coupling variables for U-statistics with cross-sectionally dependent samples. It is a result of Theorem 2.1 in \cite{chen2010SteinCouplings}, but we provide a proof here for the special cases of interest to clarify the items' dependence on the indices.

The proof is based on the following lemma, see \cite{ross2011FundamentalsSteins} Proposition 2.4 and Lemma 2.5.
\begin{lemma}\label{lemma:stein-lemma}
    Let \(Z\sim N(0,1)\) and \(W\) be a random variable. Given a class of absolutely continuous test functions \(\mathbb{G}\), for any \(g\in \mathbb{G}\), there exists \(f_{g}\) that solves the equation:
    \begin{equation*}
        f_{g}'(w) - wf_{g}(w) = g(w) - \Ep g(Z),
    \end{equation*}
    such that,
    \begin{equation*}
        \norm{f_{g}} \leq 2\norm{g'}, \quad \norm{f'_{g}}\leq \sqrt{\frac{2}{\pi}}\norm{g'}, \quad \norm{f''_{g}} \leq 2\norm{g'},
    \end{equation*}
    where \(\norm{f} = \sup_{x\in \mathcal{X}}\abs{f(x)}\), \(\mathcal{X}\) being the domain of \(f\). As a result,
    \begin{equation}\label{eq:stein - lemma}
        d_{\mathbb{G}}(W,Z) := \sup_{g\in \mathbb{G}}\abs{\Ep g(W) - \Ep g(Z)}  = \sup_{g\in \mathbb{G}} \abs{\Ep\pqty{f_{g}'(W) - Wf_{g}(W)}}.
    \end{equation}
\end{lemma}

\begin{proof}[Proof of \autoref{lemma:stein}]
    As a result of \autoref{lemma:stein-lemma},
    \begin{equation*}
        \dist_{W}(W, N(0,1)) = \sup_{g\in \mathbb{L}_{1}}\abs{\Ep\bqty{f_{g}'(W) - Wf_{g}(W)}}.
    \end{equation*}
    From \cite{chen2010SteinCouplings}, we know \(f_{g}\) is twice differentiable and \(f_{g}'\) is continuous. For any twice differentiable \(f\) with continuous derivative,
    \begin{align*}
        f'(W) - Wf(W) &= f'(W) - \sum_{i}G_{i}\bqty{f(W) - f(W_{i}') + f(W_{i}')} \\
        &= - \sum_{i} G_{i} f(W_{i}') + f'(W) \bqty{1 - \sum G_{i}D_{i}} + \frac{1}{2}\sum_{i} G_{i}D_{i}^{2}f''(\tilde{W}_{i}'),
    \end{align*}
    where we construct \(G_{i}\) such that \(\sum_{i\in \mathcal{I}} G_{i} = W\), \(D_{i} = W - W_{i}'\) and \(\tilde{W}_{i}'\) is some point between \(W\) and \(W_{i}'\). We refine the analysis of the second term by further decomposing each \(D_{i}\), that is we construct \(D_{ij}\) such that \(\sum_{j} D_{ij} = D_{i}\), and let \(M_{ij} = \Ep G_{i} D_{ij}\),
    \begin{align*}
        f'(W)\bqty{1 - \sum_{i} G_{i}D_{i}} &= f'(W) \bqty{1 - \sum_{i,j} M_{ij} + \sum_{i,j} \pqty{M_{ij} - G_{i}D_{ij}}} \\
        &= f'(W) \bqty{1 - \sum_{i,j} M_{ij}} + \sum_{i,j} f'(W)\pqty{M_{ij} - G_{i}D_{ij}}
    \end{align*}

    For each \(i \in \mathcal{I}\) and \(j \in \mathcal{J}_{i}\), we then construct \(W_{ij}'\) and \(D_{ij}' = W - W_{ij}'\), and with \(f'(W) = f'(W_{ij}') + f''\pqty{\tilde{W}_{ij}'} D_{ij}'\), we have,
    \begin{align*}
        \sum_{i,j} f'(W)\pqty{M_{ij} - G_{i}D_{ij}} = {\sum_{ij}(M_{ij} - G_{i}D_{ij})\bqty{f'(W_{ij}') + {D}_{ij}' f''(\tilde{W}_{ij}')}}.
    \end{align*}
    Putting together we have,
    \begin{align*}
        f'(W) - Wf(W)
        &= {-\sum_{i} G_{i}f(W_{i}')}
        + {f'(W)(1 - \sum_{ij}M_{ij}) }\nonumber \\
        &\quad + {\sum_{ij} f'(W_{ij}')(M_{ij} - G_{i}D_{ij}) }
        + {\sum_{ij}(M_{ij} - G_{i}D_{ij}){D}_{ij}' f''(\tilde{W}_{ij}')} \nonumber\\
        &\quad + {\frac{1}{2} \sum_{i} G_{i}D_{i}^{2} f''(\tilde{W}_{i})}.
    \end{align*}

    We know from \cite{ross2011FundamentalsSteins},
    \begin{equation*}
        \norm{f_{g}} \leq 2 \norm{g'}, \quad \norm{f_{g}'}\leq \sqrt{\frac{2}{\pi}} \norm{g'} , \quad \norm{f_{g}''}\leq 2 \norm{g'}
    \end{equation*}
    and since \(g\) is Lipschitz with constant \(1\), \(\norm{g'}\leq 1\).
    For the first term we have,
    \begin{equation*}
        \abs{\Ep\sum_{i} G_{i} f_{g}\pqty{W_{i}'}} \leq \Ep \sum_{i}\abs{f_{g}\pqty{W_{i}'} \Ep^{W_{i}'}G_{i}} \leq \sum_{i}\Ep \norm{f_{g}} \abs{\Ep^{W_{i}'}G_{i}} \leq 2\sum_{i}\Ep\abs{\Ep^{W_{i}'} G_{i}}
    \end{equation*}
    by taking conditional expectation and moving absolute value inside. The other terms can be bounded in a similar way.
\end{proof}

We will apply \autoref{lemma:stein} to gauge the error we make when approximating the distribution of U-statistics with a normal distribution by constructing the coupling variables \((G_{i}, D_{ij}, D_{ij}')\) for non-degerate and degenerate U-statistics respectively, when the underlying sample is cross-sectionally dependent.

Our estimates of the moments and conditional moments are based on the following coupling result, famously known as Berbee's lemma. It allows us to replace a dependent random variable with an independent copy and to bound the error probability with the \(\beta\)-mixing coefficients.

\begin{lemma}[Berbee's Coupling Lemma]\label{lemma:berbee}
    Let \((\Omega, \mathcal{F}, P)\) be a probability space and \(\mathcal{A}\) be a sub-\(\sigma\)-algebra of \(\mathcal{F}\), and suppose \(X: (\Omega, \mathcal{F}, P) \to \mathcal{X}\) is a random variable taking values in a Polish space \(\mathcal{X}\), then there exists a random variable \(\tilde{X}\) which is independent of \(\mathcal{A}\) and has the same law as \(X\) such that
    \begin{equation*}
        P\pqty{X \neq \tilde{X}} = \beta\pqty{\mathcal{A}, \sigma\pqty{X}}.
    \end{equation*}
\end{lemma}
\begin{proof}
    See \cite{rio2017AsymptoticTheory} Lemma 5.1 and \cite{doukhan1994Mixing}.
\end{proof}

In the next lemma, we use the notation that for two sub-\(sigma\)-algebra \(\mathcal{F}_{1}, \mathcal{F}_{2}\subset \mathcal{F}\), \(\mathcal{F}_{1} \oplus \mathcal{F}_{2} = \sigma\pqty{\mathcal{F}_{1}\cup \mathcal{F}_{2}}\).
\begin{lemma}\label{lemma:independent}
    Suppose a sub-\(\sigma\)-algebra \(\mathcal{A}\) is independent of \(\mathcal{F}_{1} \oplus \mathcal{F}_{2}\), then
    \begin{equation*}
        \beta\pqty{\mathcal{F}_{1}, \mathcal{A} \oplus \mathcal{F}_{2}} \leq \beta(\mathcal{F}_{1}, \mathcal{F}_{2}).
    \end{equation*}
\end{lemma}
\begin{proof}
    Let \(\mathcal{C} = \cqty{\emptyset, \Omega}\) denote the trivial \(\sigma\)-algebra, then \(\mathcal{C}\oplus \mathcal{A}\) is independent of \(\mathcal{F}_{1} \oplus \mathcal{F}_{2}\), and by Theorem 1 in \cite{doukhan1994Mixing}, also \cite{bradley2005BasicProperties},
    \begin{equation*}
        \beta(\mathcal{C}\oplus \mathcal{F}_{1}, \mathcal{A}\oplus \mathcal{F}_{2}) \leq \beta(\mathcal{F}_{1}, \mathcal{F}_{2}) + \beta(\mathcal{C}, \mathcal{A}),
    \end{equation*}
    and the result follows.
\end{proof}

We now formally define the notion of \(m\)-connectedness and \(m\)-profile, which are key to our analysis of the dependence structure of the observations indexed by \(\mathbf{i}\) and to our construction of the coupling variables.
Let \(\mathbf{i} = (i_{1}, \dots, i_{q})\) denote a vector of indices of length \(q \in \mathbb{N}^{+}\). We write \(i\in \mathbf{i}\) if \(i\) belongs to the set \(\cqty{i_{1}, \dots,i_{q}}\). Let \(\nb{i}{m} = \cqty{j\in \mathcal{I}_{n} : \dist\pqty{i,j} \leq m}\) be the \(m\)-neighbourhood of \(i\) for each \(i \in \mathcal{I}_{n}\).
If \(j\in \nb{i}{m}\), then \(i\in \nb{j}{m}\) and we write \(i \stackrel{m}{\leftrightarrow} j\), which is a reflexive and symmetric relation.
The main quantity we use to measure the sparsity of the index space \(\mathcal{I}_{n}\) depends on the following definition of an equivalence relation built from the relation \(\stackrel{m}{\leftrightarrow}\).
For a fixed \(m > 0\) and a given \(\mathbf{i}\in \mathcal{I}_{n}^{q}\), we say \(i\in \mathbf{i}\) and \(j \in \mathbf{i}\) are \textit{\(m\)-connected in \(\mathbf{i}\)}, written as \(i \stackrel{\mathbf{i},m}{\leftrightsquigarrow} j\), if there exist nodes \(i'_{1}, \dots, i_{p}' \in \mathbf{i}\), such that
\begin{equation*}
    i \stackrel{m}{\leftrightarrow} i'_{1} \stackrel{m}{\leftrightarrow} i'_{2} \stackrel{m}{\leftrightarrow} \cdots \stackrel{m}{\leftrightarrow} i'_{p} \stackrel{m}{\leftrightarrow} j.
\end{equation*}
Intuitively, \(i\) and \(j\) are \(m\)-connected in \(\mathbf{i}\) if we can go from \(i\) to \(j\) while moving at most \(m\) units at a time and using only intermediate nodes in \(\mathbf{i}\).

For each \(\mathbf{i} = (i_{1}, \dots, i_{q})\), \(m\)-connectedness defines an equivalence relation on the set \(\cqty{i_{1}, \dots, i_{q}}\), and the corresponding quotient set is \(\cqty{i_{1}, \dots, i_{q}}\slash \stackrel{\mathbf{i},m}{\leftrightsquigarrow} = \cqty{Q_{1}, \dots, Q_{S}}\), with \(S\) equal to the number of equivalence classes in \(\mathbf{i}\). We can then compute, for each \(s = 1,2,\dots, S\),
\begin{equation*}
    \Theta_{s}^{m}(\mathbf{i}) = \sum_{j = 1}^{q} \indicator{i_{j} \in Q_{s}} ,
\end{equation*}
which is the number of occurrences of indices in each \(Q_{s}\), counting duplicates, and \(\sum_{s = 1}^{S} \Theta_{s}^{m}(\mathbf{i}) = q\).
We will arrange \(\Theta_{s}^{m}\pqty{\mathbf{i}}\) in descending order such that \(\Theta_{1}^{m}\pqty{\mathbf{i}} \geq \Theta_{2}^{m}\pqty{\mathbf{i}} \geq\dots \geq \Theta_{S}^{m}\pqty{\mathbf{i}}\). This constructive definition shows that every vector of indices \(\mathbf{i}\) has a unique \textit{profile} in terms of the \(\theta\)'s.
\begin{definition}
    For a fixed \(m > 0\) and \(q \in \mathbb{N}^{+}\), the \textit{\(m\)-profile} of \(\mathbf{i}\in \mathcal{I}_{n}^{q}\) is defined as
    \begin{equation*}
        \pi_{m} \pqty{\mathbf{i}} = \pqty{\Theta_{1}^{m}(\mathbf{i}), \Theta_{2}^{m}\pqty{\mathbf{i}}, \dots, \Theta_{S}^{m}\pqty{\mathbf{i}}}.
    \end{equation*}
\end{definition}

The \(m\)-profile encodes the dependence structure of the observations indexed by \(\mathbf{i}\).
We will treat functionals of \((X_{i}: i \in \mathbf{i})\) differently according to the \(m\)-profile of \(\mathbf{i}\) in our theoretical analysis.
For given integers \(q_{1}\geq \cdots \geq q_{S} > 0\) satisfying \(\sum_{s =1}^{S} q_{s} = q\), we define the following set, which consists of all vectors of indices of length \(q\) with a given \(m\)-profile.
\begin{equation*}
    T_{q_{1}, \dots, q_{S}}^{m} = \cqty{\mathbf{i} =(i_{1}, \dots, i_{q}) : \pi_{m}(\mathbf{i}) = (q_{1}, \dots, q_{S})}, \mathtext{with} \tau_{q_{1}, \dots, q_{S}}^{m} = \abs{T_{q_{1}, \dots, q_{S}}^{m}}.
\end{equation*}
The space \(\mathcal{I}_{n}^{q}\) can be partitioned into sets of vectors of indices with the same \(m\)-profile.\footnote{For ease of notation, in the subsequent sections we write $q_{j}^{l}$ for $q_{j},q_{j+1},\dots, q_{j+l}$ when $q_{j} = q_{j+1} = \dots = q_{j+l}$. For example, \(\tau_{1^{4}}^{m} = \tau^{m}_{1,1,1,1}\).}
\begin{equation*}
    \mathcal{I}_{n}^{q} = \bigsqcup_{\stackrel{q_{1}\geq \dots \geq q_{S} > 0}{\sum q_{s} = q}} T_{q_{1}, \dots, q_{S}}^{m},
\end{equation*}
where the symbol \(\sqcup\) denotes a union of disjoint sets.
As a result, for any fixed \(m\),
\begin{equation*}
    \sum_{\stackrel{q_{1}\geq \dots \geq q_{S} > 0}{\sum_{s=1}^{S} q_{s} = q}} \tau_{q_{1},\dots,q_{S}}^{m} = n^{q}.
\end{equation*}

The profile \(\pi_{m}\pqty{\mathbf{i}}\) contains information about the dependence structure of the observations indexed by the components of \(\mathbf{i}\).
One particularly useful case is when one index \(i_{k}\) in the vector of indices \(\mathbf{i}\) is more than \(m\)-distance away from the rest, we call such index an \(m\)-free index in the vector \(\mathbf{i}\).
\begin{definition}[\(m\)-free index]\label{def:m-free}
    Given a vector of indices \( \lst = (i_{1},\dots, i_{q}) \in \mathcal{I}^{q}_{n}\),  we say \(i_{k}\) is an \textit{\(m\)-free index} in \(\lst\), if
    \begin{equation*}
        \dist\pqty{i_{k}, \mathbf{i}_{- i_{k}}} = \correction{\min_{\ell\neq k}\dist(i_k,i_\ell)} > m.
    \end{equation*}
    A vector of indices \(\mathbf{i}\) has an \(m\)-free index if and only if the last digit of its \(m\)-profile is \(1\), that is,
    \begin{equation*}
        \pi_{m}(\mathbf{i}) = (\theta_{1}^{m}\pqty{\mathbf{i}}, \dots, 1).
    \end{equation*}
\end{definition}
When \(i_{k}\) is an \(m\)-free index in \(\mathbf{i}\), and \(\cqty{X_{i}:i \in \mathcal{I}_{n}}\) is \(\beta\)-mixing with respect to \(\dist\), we can bound the error we make if we replace \(X_{i_{k}}\) by an independent copy \(\tilde{X}_{i_{k}}\) using \autoref{lemma:berbee}.

In our treatment of both the non-degenerate and degenerate cases, we will need bounds on the moments of products of the form \(\prod_{j}^{J} H_{i_{j}k_{j}}\). The following lemma is particularly useful which provides a bound on the expectation of the products when one of the indices is \(m\)-free.

\begin{lemma}\label{lemma:m-free}
    Suppose \(\pqty{X_{i}:i\in \mathcal{I}}\) is \(\beta\)-mixing, \(H\) is a degenerate kernel, and \(\mathbf{H}_{J+\delta}<\infty\) for some \(\delta > 0\). If there exists an \(m\)-free index \(i^{*}\) in \(\lst = \pqty{i_{1},k_{1},\dots,i_{J},k_{J}}\), then
    \begin{equation*}
        \abs{\Ep\bqty{\prod_{j=1}^{J} H_{i_{j} k_{j}}}}
        \leq 2 \mathbf{H}_{J+ \delta}^{J} \beta(1, 2J-1, m)^{\frac{\delta}{J + \delta}}.
    \end{equation*}
\end{lemma}

\begin{proof}
    Without loss of generality, let \(k_{J}\) be \(m\)-free, and write \(\mathcal{G}_{-k_{J}} = \sigma\pqty{X_{i}: i\in \mathbf{i}_{- k_{J}}}\). By \autoref{lemma:berbee}, we can construct \(\tilde{X}_{k_{J}}\) that is independent of \(\mathcal{G}_{-k_{J}}\) and has the same distribution as \(X_{k_{J}}\), such that
    \begin{equation*}
        P(X_{k_{J}} \neq \tilde{X}_{k_{J}}) = \beta\pqty{\sigma(X_{k_{J}}), \sigma\pqty{X_{i}: i \in \mathbf{i}_{- k_{J}}}}.
    \end{equation*}
    As \(H\) is degenerate, we have by \autoref{lemma:taking-out-cond-set},
    \begin{align*}
        \Ep\bqty{H_{i_{1}k_{1}}\dots H_{i_{J - 1} k_{J - 1}} H_{i_{J}\tilde{k}_{J}}}
        &= \Ep\bqty{H_{i_{1}k_{1}} \dots H_{i_{J - 1} k_{J - 1}} \Ep\bqty{H_{i_{J} \tilde{k}_{J}} \mid \mathcal{G}_{-k_{J}}}} \\
        &= \Ep\bqty{H_{i_{1}k_{1}} \dots H_{i_{J - 1} k_{J - 1}} \Ep\bqty{H_{i_{J} \tilde{k}_{J}} \mid X_{i_{J}}}}
        = 0.
    \end{align*}
    As a result,
    \begin{align*}
        \Ep\prod_{j}^{J} H_{i_{j}k_{j}}
        &= \Ep\bqty{\prod_{j}^{J-1} H_{i_{j}k_{j}}\cdot\pqty{H_{i_{J}k_{J}} - H_{i_{J}\tilde{k}_{J}}}} \\
        &= \Ep\bqty{\prod_{j}^{J-1} H_{i_{j}k_{j}}\cdot\pqty{H_{i_{J}k_{J}} - H_{i_{J}\tilde{k}_{J}}} \indicator{X_{k_{J}}\neq \tilde{X}_{k_{J}}}}.
    \end{align*}
    Hence we have
    \begin{align*}
        \abs{\Ep\bqty{\prod_{j}^{J} H_{i_{j}k_{j}}}}
        &\leq \norm{\prod_{j}^{J-1} H_{i_{j}k_{j}}\cdot\pqty{H_{i_{J}k_{J}} - H_{i_{J}\tilde{k}_{J}}} \indicator{X_{k_{J}}\neq \tilde{X}_{k_{J}}}}_{1}\\
        &\leq \prod_{j}^{J-1} \norm{H_{i_{j}k_{j}}}_{J+ \delta} \cdot \pqty{\norm{H_{i_{J} k_{J}}}_{J+\delta} + \norm{H_{i_{J} \tilde{k}_{J}}}_{J+\delta}} \cdot \norm{\indicator{X_{k_{J}} \neq \tilde{X}_{k_{J}}}}_{\frac{J + \delta}{\delta} }.
    \end{align*}
    As \(\norm{\indicator{X_{k_{J}} \neq \tilde{X}_{k_{J}}}}_{\frac{J + \delta}{\delta} } = \beta\left(\sigma\left(X_{k_{J}}\right), \sigma\pqty{X_{i}: i\in \mathbf{i}_{- k_{J}}}\right)^{\delta/(J +\delta)}\), we have
    \begin{equation*}
        \abs{\Ep\bqty{\prod_{j}^{J} H_{i_{j}k_{j}}}}\leq 2 \mathbf{H}_{J +\delta}^{J} \beta(1, 2J - 1 , m )^{\frac{\delta}{J + \delta}},
    \end{equation*}
    using the moment and mixing assumptions.
\end{proof}

\begin{lemma}\label{lemma:aux-one-index}
    Let \(g_{i} = g_{i}(X_{i})\) be a centered one-index term with \(\Ep g_{i} = 0\) and \(\norm{g_{i}}_{1+\delta}<\infty\). Let \(\mathcal{G}\subset \mathcal{F}\) be a sub-\(\sigma\)-algebra, and let \(\tilde{X}_{i}\) be a copy of \(X_{i}\) that is independent of \(\mathcal{G}\). Write \(\tilde{g}_{i} = g_{i}(\tilde{X}_{i})\). Then
    \begin{equation*}
        \Ep\abs{\Ep\bqty{g_{i} \mid \mathcal{G}}} \leq 2 \norm{g_{i}}_{1+\delta}\Prob\pqty{X_{i}\neq \tilde{X}_{i}}^{\frac{\delta}{1+\delta}}.
    \end{equation*}
    If \(g_{j} = g_{j}(X_{j})\) is another centered one-index term with \(\norm{g_{i}}_{2+\delta}, \norm{g_{j}}_{2+\delta}<\infty\), and if \(\tilde{X}_{i}\) is independent of \(\sigma(X_{j})\), then
    \begin{equation*}
        \abs{\Ep\bqty{g_{i}g_{j}}} \leq 2 \norm{g_{i}}_{2+\delta}\norm{g_{j}}_{2+\delta}\Prob\pqty{X_{i}\neq \tilde{X}_{i}}^{\frac{\delta}{2+\delta}}.
    \end{equation*}
    In particular, if \(j\notin \nb{i}{m}\), the second bound applies with \(\Prob(X_{i}\neq \tilde{X}_{i})\leq \beta(1,1,m)\); and if \(\mathcal{G} = \sigma\cqty{X_{k}: k\notin \nb{i}{m}}\), the first bound applies with \(\Prob(X_{i}\neq \tilde{X}_{i})\leq \beta(1,\infty,m)\).
\end{lemma}
\begin{proof}
    Since \(\Ep \tilde{g}_{i} = \Ep g_{i} = 0\) and \(\tilde{g}_{i}\) is independent of \(\mathcal{G}\),
    \begin{equation*}
        \Ep\abs{\Ep\bqty{g_{i}\mid \mathcal{G}}}
        = \Ep\abs{\Ep\bqty{g_{i}-\tilde{g}_{i}\mid \mathcal{G}}}
        \leq \Ep\abs{\pqty{g_{i}-\tilde{g}_{i}}\indicator{X_{i}\neq \tilde{X}_{i}}}.
    \end{equation*}
    By H\"older's inequality,
    \begin{equation*}
        \Ep\abs{\pqty{g_{i}-\tilde{g}_{i}}\indicator{X_{i}\neq \tilde{X}_{i}}}
        \leq \norm{g_{i}-\tilde{g}_{i}}_{1+\delta}\Prob\pqty{X_{i}\neq \tilde{X}_{i}}^{\frac{\delta}{1+\delta}}
        \leq 2\norm{g_{i}}_{1+\delta}\Prob\pqty{X_{i}\neq \tilde{X}_{i}}^{\frac{\delta}{1+\delta}}.
    \end{equation*}

    For the covariance bound, since \(\tilde{g}_{i}\) is independent of \(g_{j}\) and both are centered,
    \begin{equation*}
        \Ep\bqty{\tilde{g}_{i}g_{j}} = 0.
    \end{equation*}
    Hence
    \begin{align*}
        \abs{\Ep\bqty{g_{i}g_{j}}}
        &= \abs{\Ep\bqty{\pqty{g_{i}-\tilde{g}_{i}}g_{j}}} \\
        &= \abs{\Ep\bqty{\pqty{g_{i}-\tilde{g}_{i}}g_{j}\indicator{X_{i}\neq \tilde{X}_{i}}}} \\
        &\leq \norm{g_{i}-\tilde{g}_{i}}_{2+\delta}\norm{g_{j}}_{2+\delta}\Prob\pqty{X_{i}\neq \tilde{X}_{i}}^{\frac{\delta}{2+\delta}} \\
        &\leq 2 \norm{g_{i}}_{2+\delta}\norm{g_{j}}_{2+\delta}\Prob\pqty{X_{i}\neq \tilde{X}_{i}}^{\frac{\delta}{2+\delta}}.
    \end{align*}
\end{proof}

\begin{lemma}\label{lemma:aux-remainder-centered}
    Let \(R_{ik} = H_{ik} - \hat{H}_{k}(X_{i}) - \hat{H}_{i}(X_{k}) + \theta_{ik}\) be the remainder kernel in \eqref{eq:hoeffding-depdendent}, and suppose \(H\) is symmetric. If \(\tilde{X}_{k}\) has the same distribution as \(X_{k}\) and is independent of \(X_{i}\), then
    \begin{equation*}
        \Ep\bqty{R_{i\tilde{k}} \mid X_{i}} = 0 \qquad \text{a.s.}
    \end{equation*}
    Likewise, if \(\tilde{X}_{i}\) has the same distribution as \(X_{i}\) and is independent of \(X_{k}\), then
    \begin{equation*}
        \Ep\bqty{R_{\tilde{i}k} \mid X_{k}} = 0 \qquad \text{a.s.}
    \end{equation*}
\end{lemma}
\begin{proof}
    Using symmetry of \(H\), independence, and the definitions of \(\hat{H}_{k}\) and \(\theta_{ik}\),
    \begin{align*}
        \Ep\bqty{R_{i\tilde{k}} \mid X_{i}}
        &= \Ep\bqty{H(X_{i},\tilde{X}_{k}) \mid X_{i}} - \hat{H}_{k}(X_{i}) - \Ep\bqty{\hat{H}_{i}(\tilde{X}_{k}) \mid X_{i}} + \theta_{ik} \\
        &= \hat{H}_{k}(X_{i}) - \hat{H}_{k}(X_{i}) - \theta_{ik} + \theta_{ik} = 0.
    \end{align*}
    The second claim is identical after exchanging the roles of \(i\) and \(k\).
\end{proof}

\begin{lemma}\label{lemma:gamma-prod}
    Let \(H\) be a degenerate kernel, suppose \(\dist\pqty{\cqty{i,j}, \cqty{k_{1},k_{2} }} > m\), and \( 0 < \dist(i,j) \leq m\), \(0< \dist(k_{1},k_{2}) \leq m\), then for \(\delta > 0\),
    \begin{equation*}
        \abs{\Ep \pqty{H_{ik_{1}}H_{jk_{2}}}} \leq C \mathbf{H}_{2 + \delta}^{2} \beta(2,2,m)^{\delta/(2+\delta)} + \mathbf{\gamma}_{m,1}.
    \end{equation*}
    where  \(\mathbf{\gamma}_{m,1} = \sup \cqty{\abs{\Ep\gga_{k_{1}k_{2}}} : i,j,k_{1},k_{2}\in \mathcal{I}_{n}, j\in \nb{i}{m}, k_{1}\neq k_{2}}\).
\end{lemma}
\begin{proof}
    By Berbee's lemma, we can construct \(\pqty{\tilde{X}_{i}, \tilde{X}_{j}}\) that has the same distribution as \(\pqty{{X}_{i}, X_{j}}\) and independent of \(\sigma(X_{k_{1}}, X_{k_{2}})\) such that \(P\pqty{\correction{\pqty{\tilde{X}_{i}, \tilde{X}_{j}} \neq \pqty{X_{i}, X_{j}}}} \leq \beta(2,2,m)\),
    \begin{align*}
        \abs{\Ep H_{ik_{1}}H_{jk_{2}}} &\leq
        \abs{
            \Ep\pqty{H_{ik_{1}}H_{jk_{2}} - H_{\tilde{i}k_{1}}H_{\tilde{j}k_{2}}}
            \indicator{\correction{\pqty{\tilde{X}_{i}, \tilde{X}_{j}} \neq \pqty{X_{i}, X_{j}}}}
        }
        + \abs{\Ep H_{\tilde{i}k_{1}}H_{\tilde{j}k_{2}}} \\
        &\leq \pqty{\norm{H_{ik_{1}} H_{jk_{2}}}_{1 + \delta/2} + \norm{H_{\tilde{i}k_{1}}H_{\tilde{j}k_{2}}}_{1 + \delta/2}} \beta(2,2,m)^{\delta/(2 + \delta)}
        + \mathbf{\gamma}_{m,1}\\
        &\leq C \mathbf{H}_{2 + \delta}^{2} \beta(2,2,m)^{\delta/(2+ \delta)} + \mathbf{\gamma}_{m,1}.  \qedhere
    \end{align*}
\end{proof}

Recall the definition, \(\gga_{k_{1}k_{2}} = \Ep[H(\tilde{X}_{i}, X_{k_{1}}) H(\tilde{X}_{j}, X_{k2})\mid X_{k_{1}}, X_{k_{2}}]\).
Before treating the terms relating to \(\gga_{k_{1}k_{2}}\), we first prove a technical lemma that basically says we can discard redundant information when doing conditional expectation.
\begin{lemma}\label{lemma:taking-out-cond-set}
    Suppose \(f(X,Y)\in L^{1}(\mathbb{R}^{d}\times \mathbb{R}^{d})\), and \(Y\) is independent of \(\sigma(X,Z)\), then
    \begin{equation*}
        \Ep\bqty{f(X,Y) \mid {X}} = \Ep\bqty{f(X,Y)\mid {X,Z}} \mathtext{a.s.}
    \end{equation*}
\end{lemma}
\begin{proof}
    For any \(A \in \sigma(X,Z)\), with a slight abuse of notation, we also denote \(A = \indicator{\omega \in A}\),
    \begin{align*}
        \Ep\pqty{f(X,Y)A} &= \int \Ep\bqty{f(X, y) A} \dd{P_{Y}(y)} \\
        &= \int \Ep\bqty{f(X, y) \Ep\pqty{A \mid \sigma(X)}} \dd{P_{Y}(y)} \\
        &= \Ep\bqty{f(X,Y)\Ep\bqty{A\mid \sigma(X)}}\\
        &= \Ep\bqty{\Ep\bqty{f(X,Y)\mid \sigma(X)} \Ep\bqty{A \mid \sigma(X)}} \\
        &= \Ep\cqty{\Ep\bqty{\Ep\bqty{f(X,Y)\mid \sigma(X)}A \mid \sigma(X)}} \\
        &= \Ep\bqty{\Ep\bqty{f(X,Y)\mid \sigma(X)} A}.
    \end{align*}
    And the result follows by the definition of conditional expectation.
\end{proof}

With the previous lemma in hand, we are ready to prove that \(\gga_{k_{1}k_{2}}\) and products \(\gga_{k_{1}k_{2}}\gga_{k_{3}k_{4}}\) have mean \(0\) as long as one of the inputs is independent of the rest.
\begin{lemma}\label{lemma:gamma-mean}
    Suppose \(\tilde{X}_{k_{4}}\) is independent of \(\pqty{X_{k_{1}},X_{k_{2}}, X_{k_{3}}}\), then
    \begin{equation*}
        \Ep\bqty{\gga_{k_{1}k_{2}}\gga_{k_{3}\tilde{k}_{4}}}= 0, \mathtext{and} \Ep\gga_{k_{3}\tilde{k}_{4}} = 0.
    \end{equation*}
\end{lemma}
\begin{proof}
    We only prove the more difficult first statement, the second statement is straightforward.
    \begin{align*}
        \Ep\bqty{\gga_{k_{1}k_{2}}\gga_{k_{3}\tilde{k}_{4}}} &= \Ep \pqty{\gga_{k_{1}k_{2}} \Ep\bqty{\gga_{k_{3}\tilde{k}_{4}}\mid k_{1},k_{2}, k_{3}}} \\
        &= \Ep\pqty{\gga_{k_{1}k_{2}} \Ep\bqty{\gga_{k_{3}\tilde{k}_{4}} \mid k_{3}}} && \mathtext{because} \tilde{k}_{4} \perp k_{1}, k_{2}, k_{3}\\
        &= \Ep\pqty{\gga_{k_{1}k_{2}} \Ep\bqty{\Ep\pqty{H_{\tilde{i}k_{3}}H_{\tilde{j}\tilde{k}_{4}}\mid k_{3}, \tilde{k}_{4}} \mid k_{3}}} &&\mathtext{where} \tilde{i},\tilde{j} \perp k_{3},\tilde{k}_{4}\\
        &= \Ep\pqty{\gga_{k_{1}k_{2}} \Ep\bqty{{H_{\tilde{i}k_{3}}H_{\tilde{j}\tilde{k}_{4}}} \mid k_{3}}}\\
        &= \Ep\pqty{\gga_{k_{1}k_{2}} \Ep\bqty{\Ep\pqty{H_{\tilde{i}k_{3}}H_{\tilde{j}\tilde{k}_{4}}\mid \tilde{i}, \tilde{j}, k_{3}} \mid k_{3}}} \\
        &= \Ep\pqty{\gga_{k_{1}k_{2}} \Ep\bqty{H_{\tilde{i}k_{3}} \Ep\pqty{ H_{\tilde{j}\tilde{k}_{4}}\mid \tilde{i}, \tilde{j}, k_{3}} \mid k_{3}}}
    \end{align*}
    because \(\tilde{k}_{4} \perp (\tilde{i}, \tilde{j}, k_{3})\) by \autoref{lemma:independent}, we have \(\Ep\bqty{H_{\tilde{j}\tilde{k}_{4}}\mid \tilde{i}, \tilde{j}, k_{3}} = \Ep\bqty{H_{\tilde{j}\tilde{k}_{4}}\mid \tilde{j}} = 0\) almost surely by \autoref{lemma:taking-out-cond-set}.
\end{proof}

Now we are ready to bound higher moments of \(\gga_{k_{1}k_{2}}\) using \autoref{lemma:berbee}.
\begin{lemma}\label{lemma:gamma-cov}
    For each \(i\in \mathcal{I}\) and \(j\in \nb{i}{m}\) and \(k_{1},k_{3}\neq i, k_{2},k_{4}\neq j\), with the definition \(\gga(x,y) = \Ep \bqty{H\pqty{X_{i}, x}H\pqty{X_{j}, y}}\) and \(\gga_{k_{1}k_{2}} = \gga(X_{k_{1}}, X_{k_{2}})\), we have the bounds on the norm \(\norm{\gga_{k_{1}k_{2}}}_{p} \leq \mathbf{H}_{2p}^{2}\), and the covariance between \(\gga_{k_{1}k_{2}}\) and \(\gga_{k_{3}k_{4}}\).
    \begin{enumerate}
        \item
            If there exists an \(m\)-free index in \((k_{1},k_{2},k_{3},k_{4})\), then
            \begin{equation*}
                \abs{\Cov\pqty{\gga_{k_{1}k_{2}}, \gga_{k_{3}k_{4}}}} \leq 4 \mathbf{H}_{4 + \delta}^{4} \beta(1,3,m)^{\frac{\delta}{4 + \delta}};
            \end{equation*}

        \item
            If \((k_{1},k_{2},k_{3},k_{4}) \in T_{2,2}^{m}\), and \(\dist\pqty{k_{1},k_{2}} \leq m\), \(\dist\pqty{k_{3},k_{4}}\leq m\), then
            we have
            \begin{equation*}
                \abs{ \Cov\pqty{\gga_{k_{1},k_{2}}, \gga_{k_{3},k_{4}}}} \leq 2 \mathbf{H}_{4 + \delta}^{4} \beta(2,2,m)^{\frac{\delta}{4 + \delta}}
            \end{equation*}

        \item
            If \((k_{1},k_{2},k_{3},k_{4}) \in T_{2,2}^{m}\) and \(\dist\pqty{k_{1},k_{3}}\leq m\) or \(\dist\pqty{k_{1},k_{4}}\leq m\), then
            we have
            \begin{equation*}
                \abs{\Cov\pqty{\gga_{k_{1}k_{2}}, \gga_{k_{3}k_{4}}}} \leq \mathbf{\Gamma}_{m,2}^{2} + C \mathbf{H}_{4 + \delta}^{4} \beta(2,2,m)^{\frac{\delta}{4 + \delta}}
            \end{equation*}
    \end{enumerate}
\end{lemma}
\begin{proof}
    It is straightforward to find that, for \(p \geq 1\), and \(\pqty{\tilde{X}_{i}, \tilde{X_{j}}}\) that has the same distribution as \(\pqty{X_{i}, X_{j}}\) and independent of \(\pqty{X_{k_{1}}, X_{k_{2}}}\),
    \begin{align*}
        \Ep\bqty{\abs{\gga_{k_{1},k_{2}}}^{p}} &= \Ep\bqty{\abs{\Ep\bqty{H\pqty{\tilde{X}_{i}, X_{k_{1}}} H\pqty{\tilde{X}_{j}, X_{k2}}\mid X_{k_{1}}, X_{k_{2}}}}^{p}}\\
        &\leq \Ep\bqty{\Ep\pqty{\abs{H^{p}\pqty{\tilde{X}_{i}, X_{k_{1}}} H^{p}\pqty{\tilde{X}_{j}, X_{k_{2}}}} \mid X_{k_{1}}, X_{k_{2}}}}\\
        &\leq \sqrt{\Ep H_{\tilde{i},k_{1}}^{2p}\Ep H_{\tilde{j},k_{2}}^{2p}} \\
        &\leq \mathbf{H}_{2p}^{2p}.
    \end{align*}

    \correction{In the first case, suppose without loss of generality that \(k_4\) is \(m\)-free. Let \(\tilde X_{k_4}\) be its Berbee copy and \(D=\{X_{k_4}\neq\tilde X_{k_4}\}\), so \(\Prob(D)\leq\beta(1,3,m)\). By \autoref{lemma:gamma-mean},}
        \begin{align*}
            \correction{\Cov\pqty{\gga_{k_1k_2},\gga_{k_3k_4}}}
            &\correction{=\Ep\bqty{\gga_{k_1k_2}\pqty{\gga_{k_3k_4}-\gga_{k_3\tilde k_4}}\indicator D}}\\
            &\correction{\quad-\Ep\gga_{k_1k_2}\,
            \Ep\bqty{\pqty{\gga_{k_3k_4}-\gga_{k_3\tilde k_4}}\indicator D}.}
        \end{align*}
    \correction{Applying H\"older's inequality gives \(\abs{\Cov\pqty{\gga_{k_1k_2},\gga_{k_3k_4}}}\leq4\mathbf H_{4+\delta}^4\beta(1,3,m)^{\frac{\delta}{4+\delta}}\).}

    \correction{In the second case, take a joint Berbee copy \(\pqty{\tilde X_{k_3},\tilde X_{k_4}}\) independent of \(\pqty{X_{k_1},X_{k_2}}\), and let \(D_{34}\) be the event on which the pair differs from its copy. Then}
        \begin{align*}
            \correction{\abs{\Cov\pqty{\gga_{k_1k_2},\gga_{k_3k_4}}}}
            &\correction{=\abs{\Ep\bqty{\gga_{k_1k_2}
            \pqty{\gga_{k_3k_4}-\gga_{\tilde k_3\tilde k_4}}\indicator{D_{34}}}}}\\
            &\correction{\leq2\mathbf H_{4+\delta}^4\beta(2,2,m)^{\frac{\delta}{4+\delta}}.}
    \end{align*}

    As for the third case, in general we can't separate the covariance into two parts, and we resort to the bound we have \(\mathbf{\Gamma}_{m,p} = \sup_{k_{1},k_{2}}\sup_{i\in \mathcal{I},j\in \nb{i}{m}} \norm{\gga\pqty{X_{k_{1}}, \tilde{X}_{k_{2}}}}_{p}\) where \(\tilde{X}_{k_{2}}\) has the same distribution of \(X_{k_{2}}\) but independent of \(\sigma(X_{k_{1}})\). Let's consider the case when \(\dist\pqty{k_{1}, k_{3}} \leq m\) and \(\dist\pqty{k_{2},k_{4}}\leq m\).
    \begin{align*}
        \abs{\Cov\pqty{\gga_{k_{1}k_{2}}, \gga_{k_{3}k_{4}}}} &\leq \abs{\Ep\bqty{\gga_{k_{1}k_{2}}\gga_{k_{3}k_{4}}}} + \abs{\Ep\gga_{k_{1},k_{2}} \Ep\gga_{k_{3}k_{4}}}\nonumber \\
        &\leq \mathbf{\Gamma}_{m,2}^{2} +\abs{\Ep\bqty{\gga_{k_{1}k_{2}}\gga_{k_{3}k_{4}} - \gga_{\tilde{k}_{1} k_{2}} \gga_{\tilde{k}_{3}k_{4}}}} + \abs{\Ep\gga_{k_{1},k_{2}} \Ep\gga_{k_{3}k_{4}}}\nonumber \\
        &\leq \mathbf{\Gamma}_{m,2}^{2} + 6\mathbf{H}_{4 + \delta}^{4}\beta(2,2,m)^{\frac{\delta}{4 + \delta}}.
    \end{align*}
    The last inequality comes from the estimates that, by constructing \(\correction{\tilde{X}_{k_{1}} \perp X_{k_{2}}}\), \(\Ep\gga_{\tilde{k}_{1}k_{2}} = 0\), and
    \begin{equation*}
        \Ep\abs{\gga_{k_{1}k_{2}} - \gga_{\tilde{k}_{1}k_{2}}} \leq 2\mathbf{H}_{2 + \delta}^{2} \beta(1,1,m)^{\frac{\delta}{2 + \delta}}.
    \end{equation*}
    Therefore, we have \(\abs{\Ep\gga_{k_{1}k_{2}}\Ep\gga_{k_{3}k_{4}}} \leq 4 \mathbf{H}_{2 + \delta}^{4}\beta(1,1,m)^{\frac{2\delta}{2 + \delta}}\), which is of smaller order than the middle term.
\end{proof}

\subsection{Non-degenerate U-statistics}\subsubsection{Proof for \autoref{thm:lln for non-degenerate}}

We start from \eqref{eq:hoeffding-depdendent}. We have
\begin{equation*}
    \frac{1}{n(n-1)} \sum_i \sum_{k\neq i} \pqty{H_{ik} - \theta_{ik}} = \frac{2}{n} \sum_{i} h_{i} + \frac{1}{n(n-1)} \sum_i \sum_{k\neq i} \pqty{H_{ik} - \hat{H}_{i}(X_{k}) - \hat{H}_{k}(X_{i}) + \theta_{ik}}.
\end{equation*}
Now \(\Ep h_{i} = 0\) for all \(i\), and for \(p\geq 1\),
\begin{equation*}
    \norm{h_{i}}_{p} \leq \frac{1}{n-1}\sum_{k\neq i} \norm{\hat{H}_{k}(X_{i}) - \theta_{ik}}_{p} \leq 2\mathbf{H}_{p}.
\end{equation*}
Noticing that \((i,j) \in T_{2}^{m} \cup T_{1,1}^{m}\), with \(\abs{T_{1,1}^{m}} = \tau_{1,1}^{m} \leq n^{2}\), \autoref{lemma:aux-one-index} gives, for \(j\notin \nb{i}{m}\),
\begin{equation*}
    \abs{\Ep\bqty{h_{i} h_{j}}} \lesssim \mathbf{H}_{2+\delta}^{2} \beta(1,1,m)^{\frac{\delta}{2+\delta}} \leq \mathbf{H}_{2+\delta}^{2} \beta(1,3,m)^{\frac{\delta}{2+\delta}},
\end{equation*}
where the last inequality uses monotonicity of the mixing coefficient in the block size. Hence
\begin{align*}
    \Ep \big(\sum_{i} h_{i}\big)^{2}
    &=  \sum_{i}\sum_{j} \Ep h_{i} h_{j} \\
    &\lesssim \tau_{2}^{m} \sup_{i} \norm{h_{i}}_{2}^{2} + n^{2} \sup_{i,j: j\notin \nb{i}{m}} \abs{\Ep h_{i} h_{j}} \\
    &\lesssim \tau_{2}^{m} \mathbf{H}_{2}^{2} + n^{2} \mathbf{H}_{2+\delta}^{2} \beta(1,3,m)^{\frac{\delta}{2+\delta}}.
\end{align*}
Hence \(\frac{1}{n} \sum_{i}h_{i} = o_{p}(1).\)

For the remainder term, write
\begin{equation*}
    R_{ik} = H_{ik} - \hat{H}_{i}(X_{k}) - \hat{H}_{k}(X_{i}) + \theta_{ik}.
\end{equation*}
As shown in the proof of \autoref{thm:non-degenerate} below, \(\norm{R_{ik}}_{p} \leq 4 \mathbf{H}_{p}\). If \((i_{1},k_{1},i_{2},k_{2})\in T_{4}^{m}\cup T_{2,2}^{m}\), there are at most \(\tau_{4}^{m}+\tau_{2,2}^{m}\) such terms and each term is bounded by \(16\mathbf{H}_{2}^{2}\). If \((i_{1},k_{1},i_{2},k_{2})\in T_{3,1}^{m}\cup T_{2,1,1}^{m}\cup T_{1^{4}}^{m}\), there is an \(m\)-free index among \((i_{1},k_{1},i_{2},k_{2})\). Coupling that index with an independent copy by Berbee's lemma, using \autoref{lemma:aux-remainder-centered}, and the same H\"older argument as in \autoref{lemma:m-free} yields
\begin{equation*}
    \abs{\Ep\bqty{R_{i_{1}k_{1}}R_{i_{2}k_{2}}}} \lesssim \mathbf{H}_{2+\delta}^{2} \beta(1,3,m)^{\frac{\delta}{2+\delta}}.
\end{equation*}
Therefore,
\begin{equation*}
    \Ep \bqty{\pqty{\sum_i \sum_{k\neq i} \pqty{H_{ik} - \hat{H}_{i}(X_{k}) - \hat{H}_{k}(X_{i}) + \theta_{ik}}}^{2}} \lesssim {\pqty{\tau_{4}^{m}+ \tau_{2,2}^{m}} \mathbf{H}_{2}^{2} + n^{4} \mathbf{H}_{2+\delta}^{2} \beta(1,3,m)^{\frac{\delta}{2+\delta}}} = o(n^{4}),
\end{equation*}
as \(\correction{\tau_{4}^{m}+\tau_{2,2}^{m}\lesssim n^{2}\tau_{2}^{m}}\) and \(\tau_{3,1}^{m} + \tau_{2,1,1}^{m} + \tau_{1^{4}}^{m} \leq n^{4}\). \qed

\subsubsection{Proof for \autoref{thm:non-degenerate}}
The rest of the section proves \autoref{thm:non-degenerate}. Start from \eqref{eq:hoeffding-depdendent}, let \(\nu_{n}^{2} = \Var\bqty{\sum_{i} h_{i}}\), and
\begin{equation*}
    \hat{W}_{n} = \frac{1}{2(n-1)\nu_{n}}\hat{S}_{n} = \frac{1}{\nu_{n}}\sum_{i\in \mathcal{I}_{n}} h_{i}.
\end{equation*}
Notice that
\begin{equation*}
    \dist_{W} \pqty{W_{n}, Z} \leq \dist_{W}\pqty{\hat{W}_{n}, Z} + \dist_{W}\pqty{\hat{W}_{n}, W_{n}},
\end{equation*}
where
\begin{equation*}
    \dist_{W}\pqty{\hat{W}_{n}, W_{n}} = \sup_{f\in \mathbb{L}_{1}}\abs{\Ep f\pqty{\hat{W}_{n}} - \Ep f(W_{n})} \leq \Ep \abs{\hat{W}_{n} - W_{n}} \leq \frac{1}{2(n-1)\nu_{n}}\sqrt{\Ep S_{n}^{* 2}}.
\end{equation*}
We first bound \(\Ep S_{n}^{* 2}\) and then bound \(\dist_{W}\pqty{\hat{W}_{n}, Z}\) using Stein's method.

\paragraph{Degenerate part \tmath{S^{*}_{n}}.}
We first consider the degenerate part. For \(R_{ik} = H_{ik} - \hat{H}_{k}(X_{i}) - \hat{H}_{i}(X_{k}) + \theta_{ik}\),
\begin{equation}\label{eq:S-star}
    \Ep S^{*2}_{n} =  \sum_{i_{1}\neq k_{1}}\sum_{i_{2}\neq k_{2}} \Ep \bqty{R_{i_{1}k_{1}}R_{i_{2}k_{2}}}.
\end{equation}
Notice that for \(p \geq 1\), with conditional Jensen's inequality, by constructing a generic \(\tilde{X}_{k}\) that has the same distribution as \(X_{k}\) for some \(k\in \mathcal{I}\) while being independent of \(X_{i}\) using \autoref{lemma:berbee}, we have
\begin{equation}\label{eq:norm of h}
    \norm{\hat{H}_k(X_{i})}_{p} = \pqty{\Ep\abs{\hat{H}_k(X_{i})}^{p}}^{\frac{1}{p}} = \pqty{\Ep\abs{\Ep\bqty{H(X_{i}, \tilde{X}_{k}) \mid X_{i}}}^{p}}^{\frac{1}{p}} \leq \pqty{\Ep \abs{H_{i, \tilde{k}}}^{p}}^{\frac{1}{p}} \leq\mathbf{H}_{p},
\end{equation}
and hence \(\norm{R_{ik}}_{p} = \norm{H_{ik} - \hat{H}_{i}(X_{k}) - \hat{H}_{k}(X_{i}) + \theta_{ik}}_{p} \leq 4 \mathbf{H}_{p}\).

The summation in \eqref{eq:S-star} can be separated into the following cases. For \((i_{1},k_{1},i_{2},k_{2})\in T_{4}^{m}\cup T_{2,2}^{m}\), we have fewer than \(\tau_{4}^{m} + \tau_{2,2}^{m}\) terms and each term is bounded by \(16 \mathbf{H}_{2}^{2}\). For \((i_{1},k_{1},i_{2},k_{2})\in T_{3,1}^{m}\cup T_{2,1,1}^{m}\cup T_{1^{4}}^{m}\), there is an \(m\)-free index. Coupling that index to an independent copy by \autoref{lemma:berbee}, the coupled product has mean zero by the auxiliary decoupled-centering lemma in Appendix A.1, and the same argument as in \autoref{lemma:m-free} gives
\begin{equation*}
    \abs{\Ep \bqty{R_{i_{1}k_{1}}R_{i_{2}k_{2}}}} \lesssim \mathbf{H}_{2 + \delta}^{2} \beta\pqty{1,3,m}^{\frac{\delta}{2 + \delta}}.
\end{equation*}
\correction{Write}
\[
    \correction{\hat{\tau}_{4}^{m}:=\tau_{3,1}^{m}+\tau_{2,1,1}^{m}+\tau_{1^{4}}^{m}.}
\]
Hence we have
\begin{equation*}
    \Ep S^{*2}_{n} \lesssim \pqty{\tau_{4}^{m} + \tau_{2,2}^{m}} \mathbf{H}_{2}^{2} + \hat{\tau}_{4}^{m} \mathbf{H}_{2 + \delta}^{2} \beta(1,3,m)^{\frac{\delta}{2 + \delta}}.
\end{equation*}

For the projection part, we use Stein's method to bound \(\dist_{W}\pqty{\hat{W}_{n}, Z}\) for a standard normal random variable \(Z\). Recall \(\hat{W}:= \hat{W}_{n} = \frac{1}{\nu_{n}} \sum_{i\in \mathcal{I}} h_{i}\), our approximate Stein's coupling is constructed with the following quantities, \(G_{i} = \frac{1}{\nu_{n}}h_{i}\), \(D_{i} =\hat{W}- W_{i}' = \frac{1}{\nu_{n}}\sum_{j\in \nb{i}{m}} h_{j}\), \(D_{ij} = \frac{1}{\nu_{n}} h_{j}\) where \(j \in \nb{i}{m}\), and \(D_{ij}' = \sum_{k\in\nb{i}{m}\cup \nb{j}{m}}\frac{1}{\nu_{n}} h_{k}\). In the following, we treat each term in \autoref{lemma:stein}.

\paragraph{Bounds on \tmath{A^{N}_{1}}.}
For the term \(A^{N}_{1}\), we have \(A^{N}_{1} \leq \frac{1}{\nu_{n}} \sum_{i\in \mathcal{I}} \Ep\abs{\Ep^{W_{i}'} h_{i}}\), where \(\Ep^{X}\bqty{Y} = \Ep\bqty{Y \mid X}\).
The properties for \(h_{i} = \frac{1}{n-1}\sum_{k\neq i} (\hat{H}_{k}(X_{i}) - \theta_{ik})\) include \(\Ep h_{i} = 0\) and
\begin{equation*}
    \norm{h_{i}}_{p} \leq \frac{1}{(n-1)}\sum_{k\neq i} \norm{\hat{H}_{k}(X_{i}) - \theta_{ik}}_{p} \leq 2 \mathbf{H}_{p}.
\end{equation*}
Notice that \(W_{i}'\in \mathcal{F}_{- i}^{m} = \sigma\cqty{X_{j}: j\notin \nb{i}{m}}\). By \autoref{lemma:aux-one-index} with \(g_{i}=h_{i}\) and \(\mathcal{G}=\mathcal{F}_{- i}^{m}\), for all \(i\in \mathcal{I}\),
\begin{equation*}
    \Ep\abs{\Ep^{W_{i}'} h_{i}} \leq 2 \norm{h_{i}}_{1 + \delta} \beta(1,\infty, m)^{\frac{\delta}{1 + \delta}} \lesssim \mathbf{H}_{1 + \delta} \beta(1,\infty, m)^{\frac{\delta}{1 + \delta}},
\end{equation*}
therefore we have
\begin{equation}
    A^{N}_{1} \lesssim \frac{n}{\nu_{n}}\mathbf{H}_{1 + \delta} \beta(1,\infty, m)^{\frac{\delta}{1 + \delta}}.
\end{equation}

\paragraph{Bounds on \tmath{A^{N}_{2}}.}

For the second term, under the assumption that \(\Ep h_{i} = 0\) for all \(i\), we have
\begin{equation*}
    \nu_{n}^{2} = \Var\pqty{\sum_{i\in \mathcal{I}} h_{i}} = \sum_{i\in \mathcal{I}}\sum_{j\in \mathcal{I}} \Ep\bqty{h_{i}h_{j}}.
\end{equation*}
Then for \(A^{N}_{2} = \abs{1 - \sum_{i} \sum_{j \in \nb{i}{m}} M_{ij}}\), where \(M_{ij} = \Ep G_{i}D_{ij} = \frac{1}{\nu_{n}^{2}}\Ep\bqty{h_{i}h_{j}}\),
\begin{equation*}
    A^{N}_{2} = \frac{1}{\nu_{n}^{2}} \abs{\nu_{n}^{2} - \sum_{i}\sum_{j\in \nb{i}{m}}\Ep h_{i}h_{j}} = \frac{1}{\nu_{n}^{2}}\abs{\sum_{i}\sum_{j\notin \nb{i}{m}} \Ep h_{i}h_{j}}.
\end{equation*}
By \autoref{lemma:aux-one-index}, \(\abs{\Ep(h_{i} h_{j})}\lesssim \mathbf{H}_{2 + \delta}^{2} \beta\pqty{1,1,m}^{\frac{\delta}{2 + \delta}}\) for \(j\notin \nb{i}{m}\), therefore,
\begin{equation}
    A^{N}_{2}\lesssim \frac{n^{2}}{\nu_{n}^{2}}\mathbf{H}_{2 + \delta}^{2} \beta(1,1,m)^{\frac{\delta}{2 + \delta}}.
\end{equation}

\paragraph{Bounds on \tmath{A^{N}_{3}}.}

For the third term, with \(M_{ij} = \frac{1}{\nu_{n}^{2}}\Ep\bqty{h_{i}h_{j}}\), we have,
\begin{align*}
    A^{N}_{3} = \Ep\sum_{i} \sum_{j\in \nb{i}{m}} \abs{\Ep^{W_{ij}'} \pqty{M_{ij} - G_{i}D_{ij}}} = \frac{1}{\nu_{n}^{2}} \Ep\sum_{i} \sum_{j\in \nb{i}{m}} \abs{\Ep^{W_{ij}'} \pqty{\Ep h_{i}h_{j} - h_{i}h_{j}}}.
\end{align*}
Let \((\tilde{X}_{i}, \tilde{X}_{j})\) be a copy of \((X_{i}, X_{j})\) that is independent of \(\mathcal{F}_{- i, - j}^{m} = \sigma(X_{k}: k\notin \nb{i}{m}\cup \nb{j}{m})\), we have
\begin{equation*}
    \Ep\bqty{h_{\tilde{i}}h_{\tilde{j}} \mid \mathcal{F}_{- i, - j}^{m}} = \Ep h_{\tilde{i}}h_{\tilde{j}}  \mathtext{a.s.}
\end{equation*}
Therefore, by arguments similar to \autoref{lemma:m-free},
\begin{align*}
    \Ep\abs{\Ep^{W_{ij}'} (\Ep h_{i}h_{j} - h_{i}h_{j})} &= \Ep\abs{\Ep^{W_{ij}'} \pqty{\Ep h_{i}h_{j} - h_{i}h_{j} - \bqty{\Ep h_{\tilde{i}}h_{\tilde{j}} - h_{\tilde{i}}h_{\tilde{j}}}}} \\
    &\leq \Ep\abs{\pqty{\Ep h_{i}h_{j} - h_{i}h_{j} - \bqty{\Ep h_{\tilde{i}}h_{\tilde{j}} - h_{\tilde{i}}h_{\tilde{j}}}} \indicator{\pqty{X_{i}, X_{j}} \neq \pqty{\tilde{X}_{i}, \tilde{X}_{j}}}}\\
    &\leq \norm{\Ep h_{i}h_{j} - h_{i}h_{j} - \bqty{\Ep h_{\tilde{i}}h_{\tilde{j}} - h_{\tilde{i}}h_{\tilde{j}}}}_{1 + \frac{\delta}{2}} \norm{\indicator{\pqty{X_{i}, X_{j}} \neq \pqty{\tilde{X}_{i}, \tilde{X}_{j}}}}_{\frac{2 + \delta}{\delta}}.
\end{align*}
As \(\norm{\mathbb{I}\{(X_{i}, X_{j}) \neq (\tilde{X}_{i}, \tilde{X}_{j})\}}_{\frac{2 + \delta}{\delta}} \correction{\leq} \beta(2,\infty,m)^{\frac{\delta}{2 + \delta}}\), we have,
\begin{equation}
    A^{N}_{3} \lesssim \frac{\tau_{2}^{m}}{\nu_{n}^{2}}\mathbf{H}_{2 + \delta}^{2}\beta(2,\infty, m)^{\frac{\delta}{2 + \delta}}.
\end{equation}

\paragraph{Bounds on \tmath{A^{N}_{4}} and \tmath{A^{N}_{5}}.}

For the fourth term,
\begin{align*}
    A^{N}_{4} &= \Ep\sum_{i} \sum_{j\in \nb{i}{m}} \abs{\pqty{M_{ij} - G_{i}D_{ij}} D_{ij}'}\\
    &\leq \frac{1}{\nu_{n}^{3}} \sum_{i}\sum_{j\in \nb{i}{m}} \sum_{k \in \nb{i}{m}\cup\nb{j}{m}}\bqty{\abs{\Ep h_{i}h_{j}}\Ep\abs{h_{k}} + \Ep\abs{h_{i}h_{j}h_{k}}}.
\end{align*}

The summands are \(O(\mathbf{H}_{3}^{3})\), and for any \(j\in\nb{i}{m}\) and \(k\in \nb{i}{m}\cup\nb{j}{m}\), \((i,j,k)\in T_{3}^{m}\), so there are \(O(\tau_{3}^{m})\) terms in the summation. Therefore,
\begin{equation}
    A^{N}_{4} \lesssim \frac{\tau_{3}^{m} }{\nu_{n}^{3}} \mathbf{H}_{3}^{3}.
\end{equation}
For the last term, we have similarly,
\begin{equation}
    A^{N}_{5} =\sum_{i} \Ep\abs{G_{i} D_{i}^{2}} = \frac{1}{\nu_{n}^{3}} \sum_{i} \sum_{j_{1}\in \nb{i}{m}, j_{2}\in \nb{i}{m}} \Ep\abs{h_{i}h_{j_{1}}h_{j_{2}}} \lesssim \frac{1}{\nu_{n}^{3}}\tau_{3}^{m} \mathbf{H}_{3}^{3}.
\end{equation}

\autoref{thm:non-degenerate} follows by combining the bounds for \(A^{N}_{1}, \dots, A^{N}_{5}\) with \(\hat{\tau}_{4}^{m} \leq n^{4}\).\qed

\subsection{Degenerate U-statistics}The main result of this subsection is the following theorem.
\begin{theorem}\label{thm:degenerate in detail}
    For degenerate U-statistics \(S_{n}\) defined as before, and \(s_{n}^{2} = \Var(S_{n})\), let \(W_{n} = \frac{1}{s_{n}}(S_{n} - \Ep S_{n})\) and \(Z \sim N(0,1)\)\correction{. For any \(m\geq1\) and \(\delta>0\),}
    \begin{equation*}
        d_{W}\pqty{W_{n} , Z} \correction{\lesssim A_{0} + A_{1} + A_{2} + A_{3} + A_{4} + A_{5}},
    \end{equation*}
    where the components are,
    \begin{align*}
        A_{0} &\lesssim \frac{1}{s_{n}}\sqrt{\tau^{4m}_{4} \mathbf{H}_{2}^{2} + \tau_{2,2}^{4m} \mathbf{H}_{2 + \delta}^{2} \beta_{4}(4m)^{\frac{\delta}{2 + \delta}}}, \quad A_{1} \lesssim  \frac{n^{2}}{s_{n}} \mathbf{H}_{1 + \delta} {\beta}(1,\infty,m)^{\frac{\delta}{1 + \delta}},\\
        A_{2} &\lesssim \frac{1}{s_{n}^{2}}\bqty{\tau_{4}^{4m} \mathbf{H}_{2}^{2} + n^{4}\mathbf{H}_{2 + \delta}^{2}\beta_{4}(m)^{\frac{\delta}{2 + \delta}}}, \\
        A_{3} &\lesssim \frac{\tau_{2}^{m}}{s_{n}^{2}}\pqty{n^{2}\mathbf{H}_{2+\delta}^{2} \beta(2,\infty, m)^{\frac{\delta}{2+\delta}} + \sqrt{\tau_{4}^{m} \mathbf{H}_{4}^{4} + (\hat{\tau}_{4}^{m} + \tau_{2,2}^{m}) \mathbf{H}_{4 + \delta}^{4}\beta_{4}(m)^{\frac{\delta}{4+\delta}} + \tau_{2,2}^{m}\mathbf{\Gamma}^{2}_{m,2} }}, \\
        A_{4}+ A_{5} &\lesssim \frac{\sqrt{\tau_{2}^{m}\tau_{4}^{m} + n \tau_{5}^{m}}}{s_{n}^{3}} \bqty{(\tau_{4}^{m} + \tau_{2,2}^{m}) \mathbf{H}_{4}^{4} + \hat{\tau}_{4}^{m} \mathbf{H}_{4 + \delta}^{4} \beta(1,6,m)^{\frac{\delta}{4 + \delta}}}^{\frac{1}{2}}\bqty{\tau_{2}^{m} \tilde{\mathbf{H}}_{2}^{2} + n^{2} \mathbf{H}_{2 + \delta}^{2}\beta_{4}(m)^{\frac{\delta}{2 + \delta}}}^{\frac{1}{2}},
    \end{align*}
    with \(\hat{\tau}^{m}_{4} = \tau_{3,1}^{m} + \tau_{2,1,1}^{m} + \tau_{1^{4}}^{m}\), \(\beta_{q}(m) = \max_{q_{1},q_{2}} \cqty{\beta(q_{1},q_{2},m) : q_{1},q_{2}\in \mathbb{N}^{+}, q_{1} + q_{2} = q}\) and \(\tilde{\mathbf{H}}_{2} = \sup_{i\neq k} \norm{H_{i,\tilde{k}}}_{2}\).
\end{theorem}

\paragraph{Bounds on \tmath{A_{0}}.}

\correction{Define \(H_{ii}=0\), \(S^{\backslash}:=\sum_i\sum_{k\in\nb{i}{4m},\,k\neq i}H_{ik}\), \(S^{\circ}:=S_n-S^{\backslash}\), and \(W^{\circ}:=s_n^{-1}S^{\circ}\). Since the index quadruples in \(\Var(S^{\backslash})\) belong to \(T_4^{4m}\cup T_{2,2}^{4m}\), Cauchy--Schwarz and the \(\beta(2,2,4m)\) covariance inequality give}
\begin{equation}
    \correction{\Var\pqty{S^{\backslash}} \lesssim \tau^{4m}_{4} \mathbf{H}_{2}^{2} + \tau_{2,2}^{4m} \mathbf{H}_{2 + \delta}^{2} \beta(2,2, 4m)^{\frac{\delta}{2 + \delta}}.}
\end{equation}
\correction{Let \(A_0:=s_n^{-1}\sqrt{\Var(S^{\backslash})}\). Then \(d_W(W_n,W^{\circ})\leq s_n^{-1}\Ep\abs{(S^{\backslash}-\Ep S^{\backslash})-\Ep S^{\circ}}\leq A_0+s_n^{-1}\abs{\Ep S^{\circ}}\).}
Next, we apply the following construction of approximate Stein's Coupling for \correction{\(W^{\circ}\)} as in \autoref{lemma:stein}. For each \(i \in \mathcal{I}\), we construct,
\begin{align*}
    G_{i} = \frac{1}{s_{n}} \sum_{k\notin \nb{i}{4m}}H_{ik}, \quad
    D_{i} = \frac{1}{s_{n}}\sum_{j\in \nb{i}{m}}\sum_{k \notin \nb{j}{4m}} [H_{jk}+H_{kj}], \quad
    W_{i}' = \correction{W^{\circ}} - D_{i},
\end{align*}
and for each \(i\) and \(j \in \nb{i}{m}\),
\begin{equation*}
    D_{ij} = \frac{1}{s_{n}}\sum_{k \notin \nb{j}{4m}} [H_{jk} + H_{kj}], \quad
    D_{ij}' = \frac{1}{s_{n}}\sum_{i' \in \nb{i}{m}\cup \nb{j}{m}}\sum_{k\notin \nb{i'}{4m}}[H_{i'k} + H_{ki'}].
\end{equation*}
\correction{Set \(W_{ij}'=W^\circ-D_{ij}'\).}
\correction{Since \(\sum_iG_i=W^{\circ}\), we have} \correction{\(d_{W}(W^{\circ}, N(0,1)) \leq  2 A_{1} + \sqrt{\frac{2}{\pi}} A_{2} +  \sqrt{\frac{2}{\pi}}A_{3} +  2A_{4} + A_{5}\)},
where
\begin{align*}
    A_{1} &= \Ep\sum_{i}\abs{\Ep^{W_{i}'} G_{i}},
    &A_{2} &= \Ep\abs{1 - \sum_{i}\sum_{j \in \nb{i}{m}} M_{ij}},\\
    A_{3} &= \Ep\sum_{i}\sum_{j \in \nb{i}{m}}\abs{\Ep^{W_{ij}'} \pqty{M_{ij}- G_{i}D_{ij}}},
    &A_{4} &= \Ep\sum_{i}\sum_{j \in \nb{i}{m}} \abs{\pqty{M_{ij} - G_{i}D_{ij}}D_{ij}'}, \\
    A_{5} &= \sum_{i} \Ep\abs{G_{i}D_{i}^{2}}.
\end{align*}
\correction{Moreover,}
    \[
        \correction{s_n^{-1}\abs{\Ep S^{\circ}}=\abs{\Ep W^{\circ}}\leq\sum_i\abs{\Ep G_i}\leq\sum_i\Ep\abs{\Ep^{W_i'}G_i}=A_1.}
    \]
    \correction{Together with the deletion and Stein bounds, this contribution is absorbed into \(A_1\) in the stated order bound.}

\paragraph{Bounds on \tmath{A_{1}}.}

Firstly, we notice that \(W_{i}' \in \mathcal{F}_{- i}^{m} = \sigma\cqty{X_{j}:j\notin \nb{i}{m}}\), because those ordered summands \(H_{jk}\) for which
\(j \in \mathcal N_i^m\) or \(k \in \mathcal N_i^m\) will be contained in \(D_{i}\) and thus removed from \correction{\(W^{\circ}\)} to get \(W_{i}'\).
By taking conditional expectation, we have,
\begin{equation*}
    {A_{1}} = \sum_{i} \Ep\abs{\Ep^{W_{i}'} G_{i}} = \frac{1}{s_{n}}\sum_{i} \Ep\abs{\Ep^{W_{i}'} \sum_{k\notin \nb{i}{4m}} \Ep\bqty{H_{ik}\mid \mathcal{F}_{-i}^{m}}}.
\end{equation*}

For each \(i\) we have constructed a random variable \(\tilde{X}_{i}\) that is independent of \(\mathcal{F}_{- i}^{m}\) and has the same distribution as \(X_{i}\), such that \(P(X_{i} \neq \tilde{X}_{i}) \leq \beta(1,n - 1,m)\leq \beta(1, \infty, m)\) by \autoref{lemma:berbee}. We also take advantage of the degeneracy condition that \(\Ep\bqty{H_{\tilde{i}k} \mid \mathcal{F}_{- i}^{m}} = \Ep\bqty{H_{\tilde{i}k}\mid X_k} = 0\) almost surely to get,
\begin{align}
    A_{1}&=\frac{1}{s_{n}}\sum_{i} \Ep\abs{\Ep^{W_i'}\sum_{k\notin \nb{i}{4m}} \Ep\bqty{\pqty{H_{ik} - H_{\tilde{i}k} } \indicator{X_{i} \neq \tilde{X}_{i}} \mid \mathcal{F}_{-i}^{m}}} \nonumber \\
    &\leq \frac{1}{s_{n}}\sum_{i} \sum_{k\notin \nb{i}{4m}} \Ep\abs{\pqty{H_{ik} - H_{\tilde{i}k} } \indicator{X_{i} \neq \tilde{X}_{i} } } \nonumber \\
    &\leq  \frac{2n^{2}}{s_{n}} \mathbf{H}_{1 + \delta} \beta(1, \infty, m)^{\frac{\delta}{1 + \delta}}
\end{align}
\paragraph{Bounds on \tmath{A_{2}}.}

The part \(A_{2}\) can be bounded by
\begin{align*}
    A_{2} & = \abs{1 - \frac{1}{s_{n}^{2}} \sum_{i}\sum_{j\in \nb{i}{m}} \Ep\bqty{\sum_{k_{1}\notin \nb{i}{4m}}\sum_{k_{2} \notin \nb{j}{4m}} H_{ik_{1}}\pqty{H_{jk_{2}} + H_{k_{2}j}}}} \nonumber                                                \\
    & = \frac{1}{s_{n}^{2}}\abs{s_{n}^{2} - \sum_{i}\sum_{j\in \nb{i}{m}} \sum_{k_{1}\notin \nb{i}{4m}}\sum_{k_{2} \notin \nb{j}{4m}} \Ep  \bqty{H_{ik_{1}}\pqty{H_{jk_{2}} + H_{k_{2}j}}}} = \frac{1}{s_{n}^{2}} \abs{s_{n}^{2} - A_{21}}.
\end{align*}
Recall that \(s_{n}^{2} = \sum_{i}\sum_{k_{1}\neq i}\sum_{j}\sum_{k_{2}\neq j} \Cov\pqty{H_{ik_{1}}, H_{jk_{2}}}\)\correction{.}

\correction{Define the following sets of vectors of indices that we are handling here; we take sums over the following sets:}
\begin{align*}
    \mathcal{I}_{1} & := \cqty{(i,k_{1},j,k_{2})\in \mathcal{I}^{4}: i\neq k_{1}, j\neq k_{2}},                                                           \\
    \mathcal{I}_{2} & := \cqty{(i,k_{1},j,k_{2}), (i,k_{1},k_{2},j)\in \mathcal{I}^{4}: j\in \nb{i}{m}, k_{1}\notin \nb{i}{4m}, k_{2} \notin \nb{j}{4m}},
\end{align*}
and we have \(\mathcal{I}_{2}\subset \mathcal{I}_{1}\). Define the following sets,
\begin{align*}
    \mathcal{I}_{3}  & := \cqty{(i,k_{1},j,k_{2}), (i,k_{1},k_{2},j)\in \mathcal{I}_{2}: i\in\nb{j}{m}, k_{1}\in \nb{k_{2}}{m}}, \\
    \mathcal{I}_{3}' & := \cqty{(i,k_{1},j,k_{2}), (i,k_{1},k_{2},j)\in \mathcal{I}_{1}: i\in\nb{j}{m}, k_{1}\in \nb{k_{2}}{m}}.
\end{align*}
Then we have \(\mathcal{I}_{3} \subset \mathcal{I}_{2}\subset \mathcal{I}_{1}\), \(\mathcal{I}_{3} \subset \mathcal{I}_{3}'\) and \( s_{n}^{2} = \sum_{\mathcal{I}_{1}} \Cov\pqty{H_{ik_{1}}, H_{jk_{2}}}\),
and
\begin{equation*}
    \correction{A_{21}=\sum_{\mathcal{I}_{2}} \Ep\bqty{H_{ik_{1}}H_{jk_{2}}}=\sum_{\mathcal{I}_{2}} \Cov\pqty{H_{ik_{1}}, H_{jk_{2}}}+B_{21},\qquad B_{21}:=\sum_{\mathcal I_2}\Ep H_{ik_1}\Ep H_{jk_2}.}
\end{equation*}
As a result, \correction{\(\abs{s_{n}^{2}-A_{21}}\leq\abs{\sum_{\mathcal{I}_{1}\setminus\mathcal{I}_{2}}\Cov\pqty{H_{ik_{1}},H_{jk_{2}}}}+\abs{B_{21}}\).}
\correction{Every pair in \(B_{21}\) has within-pair distance greater than \(4m\). The \(J=1\) case of \autoref{lemma:m-free} and \(\abs{\mathcal I_2}\lesssim n^4\) give \(\abs{B_{21}}\lesssim n^4\mathbf H_{2+\delta}^2\beta_4(m)^{\frac{\delta}{2+\delta}}\).}

Because \(\mathcal{I}_{1}\setminus \mathcal{I}_{2}\subset \mathcal{I}_{1}\setminus \mathcal{I}_{3} = (\mathcal{I}_{1} \setminus \mathcal{I}_{3}') \cup (\mathcal{I}_{3}' \setminus \mathcal{I}_3)\), the first term can be bounded by,
\begin{align*}
    \abs{\sum_{\mathcal{I}_{1}\setminus \mathcal{I}_{2}} \Cov\pqty{H_{ik_{1}}, H_{jk_{2}}}} & \leq \sum_{\mathcal{I}_{1} \setminus \mathcal{I}_{3}'} \abs{\Cov\pqty{H_{ik_{1}}, H_{jk_{2}}}} +\sum_{\mathcal{I}_{3}'\setminus \mathcal{I}_{3}} \abs{\Cov\pqty{H_{ik_{1}}, H_{jk_{2}}}} \\
    & = A_{22} + A_{23}.
\end{align*}
\(\mathcal{I}_{1}\setminus \mathcal{I}_{3}'\) consists of the vectors of indices that can be categorised into the following cases:
\begin{enumerate}
    \item At least one of the indices is \(m\)-free in \((i,k_{1},j,k_{2})\), the number of such summands is bounded by \(\hat{\tau}_{4}^{m} = \tau_{3,1}^{m} + \tau_{2,1,1}^{m} + \tau_{1^{4}}^{m}\)\correction{, and the \(J=2\) and \(J=1\) cases of \autoref{lemma:m-free} give}
        \begin{equation*}
            \abs{\Cov\pqty{H_{ik_{1}} , H_{jk_{2}}}} \lesssim \mathbf{H}_{2 + \delta}^{2}\beta(1,3,m)^{\frac{\delta}{2 + \delta}}.
        \end{equation*}
    \item \((i,k_{1},j,k_{2})\in T_{4}^{m}\). In this case the number of summands is bounded by \(\tau_{4}^{m}\) and each one can be bounded by
        \begin{equation*}
            \abs{\Cov\pqty{H_{ik_{1}} , H_{jk_{2}}}} \lesssim \mathbf{H}_{2}^{2}.
        \end{equation*}
    \item \((i,k_{1},j,k_{2})\in T_{2,2}^{m}\) and since \((i,k_{1}, j, k_{2})\notin \mathcal{I}_{3}\), it can only happen that \(\dist\pqty{i,k_{1}} \leq m\), \(\dist\pqty{j,k_{2}} \leq m\) while \(\dist\pqty{(i,k_{1}), (j,k_{2})} > m\). Then the number of such summands is bounded by \(\tau_{2,2}^{m}\) and each one can be bounded by
        \begin{equation*}
            \abs{\Cov\pqty{H_{ik_{1}} , H_{jk_{2}}}} \lesssim \mathbf{H}_{2 + \delta}^{2}\beta(2,2,m)^{\frac{\delta}{2 + \delta}}.
        \end{equation*}
\end{enumerate}

Overall, we will have
\begin{equation*}
    A_{22} \lesssim \hat{\tau}_{4}^{m} \mathbf{H}_{2 + \delta}^{2}\beta(1,3,m)^{\frac{\delta}{2 + \delta}} + \tau_{4}^{m} \mathbf{H}_{2}^{2} + \tau_{2,2}^{m} \mathbf{H}_{2 + \delta}^{2}\beta(2,2,m)^{\frac{\delta}{2 + \delta}}.
\end{equation*}

For any \((i,k_{1},j,k_{2})\in \mathcal{I}_{3}'\setminus \mathcal{I}_{3}\), because \(\mathcal{I}_{3}\subset \mathcal{I}_{3}'\), we know \((i,k_{1},j,k_{2})\) satisfies the restriction of \(\mathcal{I}_{3}'\) and additionally, either \(k_{1}\in \nb{i}{4m}\setminus\cqty{i}\) or \(k_{2}\in \nb{j}{4m}\setminus\cqty{j}\). As a result, \((i,k_{1},j,k_{2})\in T_{4}^{4m}\). Hence,
\begin{equation*}
    A_{23} \lesssim \tau_{4}^{4m} \mathbf{H}_{2}^{2}.
\end{equation*}

We can put everything together and have
\begin{equation}
    A_{2} \lesssim \frac{1}{s_{n}^{2}}\bqty{\tau_{4}^{4m} \mathbf{H}_{2}^{2} + n^{4}\mathbf{H}_{2 + \delta}^{2}\beta_{4}(m)^{\frac{\delta}{2 + \delta}}},
\end{equation}
where \(\beta_{q}(m) = \max_{q_{1} + q_{2} = q} \cqty{\beta(q_{1},q_{2},m)}\).

\paragraph{Bounds on \tmath{A_{3}}.}

For \(A_{3}\), we find bounds using the properties of \(\gga(x, y) = \Ep\bqty{H(X_{i}, x) H(X_{j}, y)}\) for \(j \in \nb{i}{m}\), which resemble the condition on \(\Ep\bqty{H(X_{i}, x) H(X_{i}, y)}\) in \cite{hall1984CentralLimit}. Recall \(\gga_{k_{1}k_{2}} = \gga(X_{k_{1}}, X_{k_{2}})\). To simplify the notation, we maintain that \(k_{1}\notin \nb{i}{4m}\) and \(k_{2}\notin \nb{j}{4m}\) and \(j \in \nb{i}{m}\) in the following, sometimes without explicitly stating the range for \((k_{1},k_{2},j)\).

Notice that \(\sigma(W_{ij}') \subset \sigma\pqty{X_{k}: k \in \mathcal{I} \setminus\pqty{\nb{i}{m} \cup \nb{j}{m}}} = : \mathcal{F}_{-ij}^{m}\), so we have by the symmetric kernel,
\begin{equation*}
    \Ep^{\mathcal{F}_{-ij}^{m}} \bqty{G_{i}D_{ij}} =\frac{2}{s_{n}^{2}} \sum_{k_{1}}\sum_{k_{2}} \Ep^{\mathcal{F}_{-ij}^{m}} \bqty{ H(X_{i}, X_{k_{1}}) H(X_{j}, X_{k_{2}})},
\end{equation*}
Notice that \(k_{1}\notin \nb{i}{4m}\), \(k_{2}\notin \nb{j}{4m}\) imply that \(X_{k_{1}}, X_{k_{2}}\in \mathcal{F}_{- ij}^{m}\) and \(\gga_{k_{1}k_{2}} = E^{\mathcal{F}_{- ij}^{m} } H_{\tilde{i}k_{1}}H_{\tilde{j}k_{2}}\), where \(\pqty{\tilde{X}_{i}, \tilde{X}_{j}}\) have the same distribution as \(\pqty{X_{i},X_{j}}\) and independent of \(\mathcal{F}^{m}_{- ij}\), we can bound
\begin{equation}\label{eq:gammak1k2}
    \Ep\abs{\Ep^{\mathcal{F}_{-ij}^{m}}H_{ik_{1}}H_{jk_{2}} - \gga_{k_{1}k_{2}}} \lesssim \beta(2, \infty, m)^{\frac{\delta}{2 + \delta}} \mathbf{H}_{2 + \delta}^{2},
\end{equation}
by \autoref{lemma:m-free}.
Therefore, we have
\begin{align*}
    A_{3}  &=\Ep\sum_{i}\sum_{j \in \nb{i}{m}}\abs{\Ep^{W_{ij}'} \pqty{M_{ij}- G_{i}D_{ij}}}  \\
    &\leq  \Ep\sum_{i}\sum_{j\in\nb{i}{m}}\Ep^{W_{ij}'} \abs{\bqty{\Ep^{\mathcal{F}_{-ij}^{m}} \pqty{G_{i}D_{ij} - \Ep G_{i}D_{ij}} - \frac{2}{s_{n}^{2}}\sum_{k_{1}}\sum_{k_{2}}\pqty{\gga_{k_{1}k_{2}} - \Ep\gga_{k_{1}k_{2}}}}}
    \\
    &\quad + \frac{2}{s_{n}^{2}}\Ep\sum_{i}\sum_{j\in\nb{i}{m}}\abs{\sum_{k_{1}}\sum_{k_{2}} \pqty{\gga_{k_{1}k_{2}} - \Ep\gga_{k_{1}k_{2}}}} = A_{31} + A_{32}.
\end{align*}

The first term can be bounded by
\begin{align*}
    \frac{4}{s_{n}^{2}} \sum_{i,j} \Ep \abs{ \sum_{k_{1},k_{2}} \Ep^{\mathcal{F}_{-ij}^{m}}\pqty{H_{ik_{1}}H_{jk_{2}}} - \gga_{k_{1}k_{2}}} \lesssim \frac{n^{2} \tau_{2}^{m}}{s_{n}^{2}} \mathbf{H}_{2 + \delta}^{2} \beta(2,\infty,m)^{\frac{\delta}{2 + \delta}}.
\end{align*}
with \eqref{eq:gammak1k2} and an estimate of the number of summands \(n^{2} \tau_{2}^{m}\).

The second term can be written as
\begin{align*}
    A_{32} \leq  \frac{2}{s_{n}^{2}} \sum_{i,j} \sqrt{\sum_{k_{1},k_{2},k_{3},k_{4}} \Cov\pqty{\gga_{k_{1}k_{2}}, \gga_{k_{3}k_{4}}}},
\end{align*}
where \(j\in\nb{i}{m}\) and \(k_{1}, k_{3}\notin \nb{i}{4m}\) and \(k_{2},k_{4}\notin \nb{j}{4m}\).

For a vector of four indices \((k_{1}, k_{2}, k_{3}, k_{4})\), we have the following partitions.
If \(\pqty{k_{1},k_{2},k_{3},k_{4}}\in T_{4}^{m}\), then we have
\begin{align*}
    \abs{\Cov\pqty{\gga_{k_{1}k_{2}}, \gga_{k_{3}k_{4}}}} &\leq \sqrt{\Ep(\gga_{k_{1}k_{2}})^{2} \Ep(\gga_{k_{3}k_{4}})^{2}} \nonumber \\
    &\leq \sqrt{\Ep\pqty{\Ep H_{\tilde{i}k_{1}}H_{\tilde{j}k_{2}} \mid k_{1},k_{2}}^{2} \Ep\pqty{\Ep H_{\tilde{i}k_{3}}H_{\tilde{j}k_{4}} \mid k_{3},k_{4}}^{2}} \nonumber \\
    &\leq \sqrt{\Ep H_{\tilde{i}k_{1}}^{2}H_{\tilde{j}k_{2}}^{2} \Ep H_{\tilde{i}k_{3}}^{2} H_{\tilde{j}k_{4}}^{2} } \leq \mathbf{H}_{4}^{4}
\end{align*}

If there is an \(m\)-free index in \(\pqty{k_{1},k_{2},k_{3},k_{4}}\), then by \autoref{lemma:gamma-cov}, we have
\begin{equation*}
    \abs{\Cov\pqty{\gga_{k_{1}k_{2}}, \gga_{k_{3}k_{4}}}} \lesssim \mathbf{H}_{4 + \delta}^{4} \beta(1,3,m)^{\frac{\delta}{4 + \delta}}.
\end{equation*}

If \(\pqty{k_{1},k_{2},k_{3},k_{4}}\in T_{2,2}^{m}\), by \autoref{lemma:gamma-cov}, if also \(\dist\pqty{k_{1},k_{2}}\leq m\), then
\begin{equation*}
    \abs{\Cov\pqty{\gga_{k_{1}k_{2}}, \gga_{k_{3}k_{4}}}}  \lesssim \mathbf{H}_{4 + \delta}^{4} \beta(2,2,m)^{\frac{\delta}{4 + \delta}},
\end{equation*}
Finally, if \(\pqty{k_{1},k_{2},k_{3},k_{4}}\in T_{2,2}^{m}\) and we have \(\dist\pqty{k_{1},k_{3}}\leq m\) or \(\dist\pqty{k_{1},k_{4}}\leq m\), then
\begin{equation*}
    \abs{\Cov\pqty{\gga_{k_{1}k_{2}}, \gga_{k_{3}k_{4}}}}  \lesssim  \mathbf{\Gamma}^{2}_{m, 2}  + \mathbf{H}_{4 + \delta}^{4} \beta(2,2,m)^{\frac{\delta}{4+\delta}}.
\end{equation*}
As a result,
\begin{equation*}
    A_{32} \lesssim \frac{\tau_{2}^{m}}{s_{n}^{2}} \sqrt{\tau_{4}^{m} \mathbf{H}_{4}^{4} + \hat{\tau}_{4}^{m} \mathbf{H}_{4 + \delta}^{4} \beta\pqty{1,3,m}^{\frac{\delta}{4 + \delta}} + \tau_{2,2}^{m} \mathbf{H}_{4 + \delta}^{4}\beta\pqty{2,2,m}^{\frac{\delta}{4 + \delta}} + \tau_{2,2}^{m} \mathbf{\Gamma}_{m,2}^{2}}.
\end{equation*}

Collecting all the terms, we have the following bound for \(A_{3}\), where we use the notation \(\hat{\tau}_{4}^{m} = \tau_{3,1}^{m} + \tau_{2,1,1}^{m} + \tau_{1^{4}}^{m}\) and \(\beta_{4}(m) = \max_{q_{1} +q_{2} = 4}\beta\pqty{q_{1},q_{2},m} \).
\begin{equation}
    A_{3} \lesssim \frac{\tau_{2}^{m}}{s_{n}^{2}}\pqty{n^{2}\mathbf{H}_{2+\delta}^{2} \beta(2,\infty, m)^{\frac{\delta}{2+\delta}} + \sqrt{\tau_{4}^{m} \mathbf{H}_{4}^{4} + \pqty{\hat{\tau}_{4}^{m} + \tau_{2,2}^{m}}\mathbf{H}_{4 + \delta}^{4}\beta_{4}(m)^{\frac{\delta}{4+\delta}} + \tau_{2,2}^{m}\mathbf{\Gamma}^{2}_{m,2} }}.
\end{equation}
In the independent case, we have \(\tau_{q_{1}, q_{2},\dots, q_{s}}^{m} = O\pqty{n^{s}}\) and \(\beta(n_{1},n_{2},m) = 0\) for \(m > 0\). Take \(m = 1\) and assume \(s_{n}^{2} \correction{\asymp} n^{2}\), we have
\begin{equation*}
    A_{3}\lesssim \frac{1}{n} \sqrt{n \mathbf{H}_{4}^{4} + n^{2} \mathbf{\Gamma}_{m,2}^{2}},
\end{equation*}
which will go to zero if \(\mathbf{H}_{4}^{4}/n + \mathbf{\Gamma}_{m,2}^{2} \to 0\), matching the condition in \cite{hall1984CentralLimit}.

\paragraph{Bounds on \tmath{A_{4}}.}

We can bound the item \(A_{4}\) in the following way.
\begin{align*}
    A_{4} &= \sum_{i}\sum_{j\in \nb{i}{m}} \Ep\abs{\pqty{M_{ij} - G_{i}D_{ij}} D_{ij}'} = \sum_{i}\sum_{j\in \nb{i}{m}} \Ep\abs{\pqty{G_{i}D_{ij} - \Ep G_{i}D_{ij}} D'_{ij}}\nonumber\\
    &\leq \sum_{i}\sum_{j\in \nb{i}{m}} \sqrt{\Ep G_{i}^{2}D_{ij}^{2} \Ep D_{ij}^{'2}} \leq \sqrt{\tau_{2}^{m}} \sqrt{\sum_{i}\sum_{j\in \nb{i}{m}}\Ep G_{i}^{2}D_{ij}^{2}\Ep D_{ij}^{'2}}\nonumber\\
    &\leq \frac{\sqrt{\tau_{2}^{m}}}{s_{n}^{3}} \sqrt{\sum_{i}\sum_{j\in \nb{i}{m}} \Ep\bqty{\bigg(\sum_{k\notin \nb{i}{4m}} H_{ik}\bigg)^{2}\bigg(\sum_{k'\notin \nb{j}{4m}} 2 H_{jk'}\bigg)^{2}} \Ep\bqty{\sum_{i'\in \nb{i}{m}\cup \nb{j}{m}}\sum_{k\notin \nb{i'}{4m}} 2 H_{i'k}}^{2}}\nonumber\\
    &= \frac{\sqrt{\tau_{2}^{m}}}{s_{n}^{3}}\sqrt{\sum_{i}\sum_{j\in \nb{i}{m}}A_{41}^{ij} A_{42}^{ij}}.
\end{align*}
As we have
\begin{equation*}
    A_{41}^{ij} = \correction{4}\Ep\bqty{\sum_{k_{1},k_{2}\notin\nb{i}{4m}}\sum_{k_{3},k_{4}\notin\nb{j}{4m}} H_{ik_{1}}H_{ik_{2}}H_{jk_{3}}H_{jk_{4}}},
\end{equation*}
if \((k_{1},k_{2},k_{3},k_{4})\in T_{4}^{m}\cup T_{2,2}^{m}\), then
\(\abs{\Ep H_{ik_{1}}H_{ik_{2}}H_{jk_{3}}H_{jk_{4}}} \leq \mathbf{H}_{4}^{4}\).
On the other hand, if there exists an \(m\)-free index in the vector \((k_{1},k_{2},k_{3},k_{4})\), without loss of generality, let us assume it is \(k_{1}\). By assumption \(\dist(k_{1}, i) > 4m\), and since \(j\in \nb{i}{m}\), \(\dist(k_{1},j) > m\), so \(k_{1}\) is also an \(m\)-free index in the vector \((i,j,k_{1},k_{2},k_{3},k_{4})\), and in this case, by \correction{the argument of} \autoref{lemma:m-free}, we have
\begin{equation*}
    \abs{\Ep H_{ik_{1}}H_{ik_{2}}H_{jk_{3}}H_{jk_{4}}} \lesssim  \mathbf{H}_{4 +\delta}^{4} \beta(1,5,m)^{\frac{\delta}{4 + \delta}}
\end{equation*}
So the first part, for all \(i\) and \(j\in \nb{i}{m}\), is bounded by
\begin{equation*}
    A_{41}^{ij}\lesssim {\pqty{\tau_{4}^{m} + \tau_{2,2}^{m}} \mathbf{H}_{4}^{4} + \pqty{\tau_{3,1}^{m} + \tau_{2,1,1}^{m} + \tau_{1^{4}}^{m}} \mathbf{H}_{4 + \delta}^{4}\beta(1,5,m)^{\frac{\delta}{4 + \delta}}}.
\end{equation*}

The second part is bounded by
\begin{align*}
    \sum_{i}\sum_{j\in \nb{i}{m}}A_{42}^{ij}
    &\correction{= \sum_{i}\sum_{j\in\nb{i}{m}}\Ep\bqty{\sum_{i'\in \nb{i}{m}\cup \nb{j}{m}}\sum_{k\notin \nb{i'}{4m}}2H_{i'k}}^{2}} \\
    &\correction{= 4\sum_{i}\sum_{j\in\nb{i}{m}}\Ep\bqty{\sum_{i',j'\in \nb{i}{m}\cup \nb{j}{m}} \sum_{k_{1}\notin\nb{i'}{4m}}\sum_{k_{2}\notin \nb{j'}{4m}} H_{i'k_{1}}H_{j'k_{2}}}} \\
    &\correction{=4}\Ep\sum_{i}\sum_{j\in \nb{i}{m}}\sum_{i',j'\in \nb{i}{m}\cup \nb{j}{m}}\bqty{\sum_{\dist\pqty{k_{1},k_{2}}\leq m} H_{i'k_{1}}H_{j'k_{2}} + \sum_{\dist\pqty{k_{1},k_{2}}> m} H_{i'k_{1}}H_{j'k_{2}}}.
\end{align*}
When \(\dist\pqty{k_{1},k_{2}}\leq m\), \(\dist((i',j') , (k_{1},k_{2})) > m\),
\begin{equation*}
    \abs{EH_{i'k_{1}}H_{j'k_{2}}} \leq \abs{\Ep H_{\tilde{i}' k_{1}} H_{\tilde{j}' k_{2}}} + \mathbf{H}_{2+\delta}^{2} \beta(2,2,m)^{\delta / (2 + \delta)} \leq \tilde{\mathbf{H}}_{2}^{2}+ \mathbf{H}_{2+\delta}^{2} \beta(2,2,m)^{\frac{\delta}{2+\delta}},
\end{equation*}
where \(\tilde{\mathbf{H}}_{2}^{2} = \sup_{i\neq k} \Ep H^{2}(X_{i}, \tilde{X}_{k})\),
and there are at most \(\tau_{4}^{m}\tau_{2}^{m}\) such summands. When \(\dist\pqty{k_{1},k_{2}}> m\), since \(\dist(k_{1}, i') > 4m\) and \(\dist(i', j')  \leq 3m\) by triangle inequality, \(\dist(k_{1}, j') > m\), which proves that \(k_{1}\) is an \(m\)-free index in \((k_{1},k_{2},i',j')\), and hence we have in this case \(\abs{EH_{i'k_{1}}H_{j'k_{2}}} \leq \mathbf{H}_{2 + \delta}^{2} \beta(1,3,m)^{\frac{\delta}{2 + \delta}}\), and the number of such summands is bounded by \(\tau_{4}^{m}n^{2}\).
Because \(\tau_{2}^{m} \leq n^{2}\),
\begin{equation*}
    \sum_{i,j\in \nb{i}{m}}A_{42}^{ij} \lesssim \tau_{4}^{m}\tau_{2}^{m} \tilde{\mathbf{H}}_{2}^{2} + \tau_{4}^{m}n^{2} \mathbf{H}_{2 + \delta}^{2} \beta_{4}(m)^{\frac{\delta}{2 + \delta}},
\end{equation*}
where \(\beta_{4}(m) = \max\cqty{\beta(2,2,m), \beta(1,3,m)}\).

Plug in the bounds we obtained, we have
\begin{equation}
    A_{4} \lesssim \frac{\sqrt{\tau_{2}^{m} \tau_{4}^{m}}}{s_{n}^{3}} \bqty{(\tau_{4}^{m} + \tau_{2,2}^{m}) \mathbf{H}_{4}^{4} + \hat{\tau}_{4}^{m} \mathbf{H}_{4 + \delta}^{4} \beta(1,5,m)^{\frac{\delta}{4 + \delta}}}^{\frac{1}{2}}
    \bqty{\tau_{2}^{m} \tilde{\mathbf{H}}_{2}^{2} + n^{2} \mathbf{H}_{2 + \delta}^{2}\beta_{4}(m)^{\frac{\delta}{2 + \delta}}}^{\frac{1}{2}}.
\end{equation}
In the independent case, again taking \(m = 1\), we have
\begin{equation*}
    A_{4} \lesssim \frac{1}{n^{2}}\pqty{n^{2} \mathbf{H}_{4}^{4}}^{\frac{1}{2}}\bqty{n \mathbf{H}_{2}^{2}}^{\frac{1}{2}} = \sqrt{\frac{\mathbf{H}_{4}^{4}\mathbf{H}_{2}^{2}}{n}}
\end{equation*}
which will go to zero, if \(\mathbf{H}_{2}^{2} = O(1)\) and \(\frac{1}{n} \mathbf{H}_{4}^{4}\to 0\), which is part of the condition in \eqref{eq:hall's condition}.

\paragraph{Bounds on \tmath{A_{5}}.}

The way we handle \(A_{5}\) is very similar to how we handled \(A_{4}\). We have, with the implicit assumption that \(k_{1},k_{2}\notin \nb{i}{4m}\), \(k_{3}\notin \nb{j}{4m}\) and \(k_{4}\notin \nb{j'}{4m}\) in the following,
\begin{equation*}
    A_{5} = \sum_{i} \Ep\abs{G_{i}D_{i}^{2}} \leq \sum_{i} \sqrt{\Ep G_{i}^{2}D_{i}^{2} \Ep D_{i}^{2}} = \frac{1}{s_{n}^{3}}\sum_{i} \sqrt{A_{51}^{i} A_{52}^{i}} \leq \frac{\sqrt{n}}{s_{n}^{3}} \sqrt{\sum_{i} A_{51}^{i}A_{52}^{i}}
\end{equation*}
where
\begin{align*}
    A_{51}^{i} &= \Ep\pqty{\sum_{k_{1}\notin \nb{i}{4m}} H_{ik_{1}}}^{2}\pqty{\sum_{j\in\nb{i}{m}}\sum_{k_{2}\notin \nb{j}{4m}} 2 H_{jk_{2}}}^{2}\\
    &= 4  \Ep \sum_{j,j'\in \nb{i}{m}} \sum_{k_{1},k_{2}\notin\nb{i}{4m}}\sum_{k_{3}\notin\nb{j}{4m}}\sum_{k_{4}\notin \nb{j'}{4m}} H_{ik_{1}}H_{ik_{2}}H_{jk_{3}}H_{j'k_{4}}.
\end{align*}

Notice that \(\dist(k_{l}, \cqty{i,j,j'})>m\) for \(l = 1,2,3,4\) and each \((k_{1},k_{2},k_{3},k_{4})\) falls into one of the following cases.
If \((k_{1},k_{2},k_{3},k_{4})\in T_{4}^{m}\cup T_{2,2}^{m}\), then
\begin{equation*}
    \Ep H_{ik_{1}}H_{ik_{2}}H_{jk_{3}}H_{j'k_{4}} \leq \mathbf{H}_{4}^{4}.
\end{equation*}
If there is an \(m\)-free index in the vector \(k_{1},k_{2},k_{3},k_{4}\), then it is also an \(m\)-free index in the larger vector \((i,j,j',k_{1},k_{2},k_{3},k_{4})\), in this case,
\begin{equation*}
    \Ep H_{ik_{1}}H_{ik_{2}}H_{jk_{3}}H_{j'k_{4}} \lesssim \mathbf{H}_{4 + \delta}^{4}\beta(1,6,m)^{\frac{\delta}{4 + \delta}}.
\end{equation*}
Therefore,
\begin{equation*}
    A_{51}^{i} \lesssim \sum_{j,j'\in\nb{i}{m}}\pqty{[\tau_{4}^{m} + \tau_{2,2}^{m}] \mathbf{H}_{4}^{4} + \bqty{\tau_{3,1}^{m} + \tau_{2,1,1}^{m} + \tau_{1^{4}}^{m}} \mathbf{H}_{4 + \delta}^{4} \beta(1,6,m)^{\frac{\delta}{4 + \delta}}}.
\end{equation*}

For \(A_{52}^{i}\), we have by similar reasoning to the treatment of \(A_{42}\),
\begin{align*}
    A_{52}^{i} &= \correction{4}\Ep \sum_{j,j'\in \nb{i}{m}} \sum_{k_{1},k_{2}} H_{jk_{1}}H_{j'k_{2}}\nonumber\\
    &\leq \correction{4}\sum_{j,j'} \bqty{\sum_{\dist\pqty{k_{1},k_{2}}\leq m} \Ep H_{jk_{1}}H_{j'k_{2}} + \sum_{\dist\pqty{k_{1},k_{2}} > m} \Ep H_{jk_{1}}H_{j'k_{2}}} \nonumber\\
    &\lesssim \sum_{j,j'}\bqty{\tau_{2}^{m}\tilde{ \mathbf{H}}_{2}^{2} + n^{2} \mathbf{H}_{2+\delta}^{2} \beta_4(m)^{\frac{\delta}{2+ \delta}}}.
\end{align*}
As a result,
\begin{equation*}
    A_{52}^{i} \lesssim \sum_{j,j'\in\nb{i}{m}} \bqty{\tau_{2}^{m} \tilde{\mathbf{H}}_{2}^{2} + n^{2} \mathbf{H}_{2+\delta}^{2} \beta_4(m)^{\frac{\delta}{2+ \delta}}}.
\end{equation*}
and
\begin{align*}
    \sum_{i} A_{51}^{i} A_{52}^{i} &\lesssim \sum_{i}\sum_{j_{1},j_{2},j_{3},j_{4}\in\nb{i}{m}}
    \pqty{[\tau_{4}^{m} + \tau_{2,2}^{m}] \mathbf{H}_{4}^{4} + \bqty{\tau_{3,1}^{m} + \tau_{2,1,1}^{m} + \tau_{1^{4}}^{m}} \mathbf{H}_{4 + \delta}^{4} \beta(1,6,m)^{\frac{\delta}{4 + \delta}}}\\
    &\cdot \bqty{\tau_{2}^{m} \tilde{\mathbf{H}}_{2}^{2} + n^{2} \mathbf{H}_{2+\delta}^{2} \beta_4(m)^{\frac{\delta}{2+ \delta}}}.
\end{align*}

As a result, we obtain the final bound for \(A_{5}\),
\begin{equation*}
    A_{5} \lesssim \frac{\sqrt{n\tau_{5}^{m}}}{s_{n}^{3}} \sqrt{\pqty{[\tau_{4}^{m} + \tau_{2,2}^{m}] \mathbf{H}_{4}^{4} + \hat{\tau}_{4}^{m} \mathbf{H}_{4 + \delta}^{4} \beta(1,6,m)^{\frac{\delta}{4 + \delta}}}  \pqty{\bqty{\tau_{2}^{m} \tilde{\mathbf{H}}_{2}^{2} + n^{2} \mathbf{H}_{2+\delta}^{2} \beta_4(m)^{\frac{\delta}{2+ \delta}}}}}.
\end{equation*}
\correction{To obtain \autoref{thm:degenerate}, let \(R_n\) denote the right-hand side of its bound. Under its assumptions, \autoref{lemma:sparsity-bound}, \(s_n\asymp n\), \(\tilde{\mathbf H}_2=O(1)\), and \(\mathbf H_{4+\delta}\geq\mathbf H_2\gtrsim1\) give \(A_0+A_1+A_3\lesssim R_n\), \(A_2\lesssim R_n^2\), and \(A_4+A_5\lesssim R_n+R_n^2\). Thus \(d_W(W_n,Z)\lesssim R_n+R_n^2\). Since \(d_W(W_n,Z)\leq\Ep|W_n|+\Ep|Z|<2\), considering \(R_n\leq1\) and \(R_n>1\) separately yields \(d_W(W_n,Z)\lesssim R_n\).}
\correction{\qed}

\subsection{Proofs for the HAC variance estimator}
For brevity, write \(\kappa_{ij}:=\kappa_{ij}(b_n)=\kappa\!\left(\frac{\mathbf d(i,j)}{b_n}\right)\) and \(\beta_4(m):=\max_{q_1+q_2=4}\beta(q_1,q_2,m)\).
Since \(\kappa(z)=0\) for \(z>1\), every row and every column of the weight matrix
\((\kappa_{ij})_{i,j\in I_n}\) has at most \(\eta_{b_n}\) non-zero entries.
For each integer \(r\ge 1\), define the set of pairs of indices that belong to each other's neighbourhood shells at distance \(r\) by \(\mathcal S_n(r):=\{(i,j)\in I_n^2:r-1<\mathbf d(i,j)\le r\}\).
By definition, \(|\mathcal S_n(r)|=\tau_2^r-\tau_2^{r-1}\).

\begin{lemma}[Weighted quadratic-form bound]\label{lem:rowsum-hac-weighted-bound}
    For any real arrays \((x_i)_{i\in I_n}\) and \((y_i)_{i\in I_n}\),
    \[
        \left|\sum_{i\in I_n}\sum_{j\in I_n} \kappa_{ij} x_i y_j\right|
        \le
        \|\kappa\|_\infty \, \eta_{b_n}
        \Big(\sum_i x_i^2\Big)^{1/2}
        \Big(\sum_j y_j^2\Big)^{1/2}.
    \]
    Consequently,
    \(
        \left|\sum_{i,j} \kappa_{ij} (x_i-\bar x)(y_j-\bar y)\right|
        \le
        4\|\kappa\|_\infty \, \eta_{b_n}
        \Big(\sum_i x_i^2\Big)^{1/2}
        \Big(\sum_j y_j^2\Big)^{1/2},
    \)
    where \(\bar x=n^{-1}\sum_i x_i\) and \(\bar y=n^{-1}\sum_i y_i\).
\end{lemma}

\begin{proof}
    Since \(\kappa_{ij}=0\) whenever \(\mathbf d(i,j)>b_n\), we have \(\sup_i \sum_j |\kappa_{ij}| \le \|\kappa\|_\infty\eta_{b_n}\) and \(\sup_j \sum_i |\kappa_{ij}| \le \|\kappa\|_\infty\eta_{b_n}\).
    The first display is therefore an immediate consequence of the Schur test. The
    second follows from the first and the elementary inequality
    \(\sum_i (x_i-\bar x)^2\le 4\sum_i x_i^2\).
\end{proof}

\begin{proposition}[Benchmark HAC form for the projection]\label{prop:benchmark-hac-projection}
    Assume the conditions of Theorem~\ref{thm:non-degenerate-clt} and Assumption~\ref{asmp:hac-non-degenerate}(i)--(iii). Then
    \(\sum_{i,j}\kappa_{ij}(h_i-\bar h)(h_j-\bar h)=\nu_n^2 + o_p(n)\), where \(\bar h:=n^{-1}\sum_i h_i\).
\end{proposition}

\begin{proof}
    Define the infeasible uncentered quadratic form \(\nu_{n,\kappa}^2 := \sum_{i\in I_n}\sum_{j\in I_n} \kappa_{ij} h_i h_j\), and let \(\gamma_{ij}:=E(h_i h_j)=\operatorname{Cov}(h_i,h_j)\). Then \[\nu_{n,\kappa}^2-\nu_n^2 = D_n+U_n+B_n,\] where \(D_n:=\sum_i \big(h_i^2-Eh_i^2\big)\), \(U_n:=\sum_{i\neq j}\kappa_{ij}\big(h_i h_j-\gamma_{ij}\big)\), and \(B_n:=\sum_{i\neq j}(\kappa_{ij}-1)\gamma_{ij}\).
    We handle the three terms separately.

    \medskip
    \noindent
    \textit{The bias term \(B_n\).}
    Since \(\kappa\) is nonincreasing and takes values in \([0,1]\), for every
    \((i,j)\in\mathcal S_n(r)\), we have \(0\le 1-\kappa_{ij} = 1-\kappa\!\left(\frac{\mathbf d(i,j)}{b_n}\right) \le 1-\kappa\!\left(\frac{r}{b_n}\right)\).
    Moreover, by Lemma~A.6 and the harmless convention \(\beta(1,1,0):=1\),
    \( |\gamma_{ij}| \lesssim \mathbf H_{2+\delta}^2\beta(1,1,r-1)^{\delta/(2+\delta)}\) for \((i,j)\in\mathcal S_n(r)\).
    Therefore,
    \begin{align*}
        |B_n|
        &\le
        \sum_{r=1}^\infty\sum_{(i,j)\in\mathcal S_n(r)}
        (1-\kappa_{ij})\,|\gamma_{ij}| \\
        &\lesssim
        \mathbf H_{2+\delta}^2
        \sum_{r=1}^\infty
        \left\{1-\kappa\!\left(\frac{r}{b_n}\right)\right\}
        \bigl(\tau_2^r-\tau_2^{r-1}\bigr)
        \beta(1,1,r-1)^{\delta/(2+\delta)}
        = o(n),
    \end{align*}
    by Assumption~\ref{asmp:hac-non-degenerate}(iii).

    \medskip
    \noindent
    \textit{The diagonal fluctuation \(D_n\).}
    Let \(g_i:=h_i^2-Eh_i^2\). Since \(\|h_i\|_{4+\delta}\le 2\mathbf H_{4+\delta}\),
    \(\|g_i\|_{2+\delta/2}\lesssim \mathbf H_{4+\delta}^2\). Hence, by Lemma~A.6,
    for \(\mathbf d(i,j)>b_n\), \( |E(g_i g_j)| \lesssim \mathbf H_{4+\delta}^4\beta(1,1,b_n)^{\delta/(4+\delta)} \le \mathbf H_{4+\delta}^4\beta_4(b_n)^{\delta/(4+\delta)}\).
    The number of ordered pairs \((i,j)\) such that \(\mathbf d(i,j)\le b_n\) is at most
    \(Cn\eta_{b_n}\). Therefore,
    \(\operatorname{Var}(D_n)= \sum_{i,j} E(g_i g_j) \lesssim n\eta_{b_n} \mathbf H_4^4 + n^2 \mathbf H_{4+\delta}^4 \beta_4(b_n)^{\delta/(4+\delta)}\).
    Since \(\eta_{b_n}^4=o(n)\) implies \(\eta_{b_n}=o(n)\), while
    \(n\beta_4(b_n)^{\delta/(4+\delta)}\to 0\) implies \(\beta_4(b_n)^{\delta/(4+\delta)}\to 0\), we obtain
    \(\operatorname{Var}(D_n)=o(n^2)\), hence \(D_n=o_p(n)\).

    \medskip
    \noindent
    \textit{The off-diagonal fluctuation \(U_n\).}
    Because \(\kappa_{ij}=0\) for \(\mathbf d(i,j)>b_n\), the term \(U_n\) only sums over
    \(i,j\) such that \(\dist(i,j) \leq \correction{b_n}\). Thus
    \[
        \correction{E(U_n^2) \le \sum_{(i,j)\in T_2^{b_n}}\sum_{(k,\ell)\in T_2^{b_n}} \abs{\Cov(h_i h_j,h_k h_\ell)}}.
    \]

    If the four indices \((i,j,k,\ell)\) lie in one \(b_n\)-connected component, then there
    are at most \(Cn\eta_{b_n}^3\) such quadruples, and each covariance is bounded by
    \(C \mathbf H_4^4\).

    Otherwise the four indices split into two \(b_n\)-connected pairs \(\{i,j\}\) and
    \(\{k,\ell\}\) with \(\mathbf d(\{i,j\},\{k,\ell\})>b_n\). Applying Berbee's lemma to
    couple \((X_i,X_j)\) away from \((X_k,X_\ell)\), and then using H\"older's inequality as
    in Lemma~A.8, yields \(\big|\operatorname{Cov}(h_i h_j,h_k h_\ell)\big| \lesssim \mathbf H_{4+\delta}^4\beta(2,2,b_n)^{\delta/(4+\delta)} \le \mathbf H_{4+\delta}^4\beta_4(b_n)^{\delta/(4+\delta)}\).
    The number of such quadruples is bounded by \(Cn^2\eta_{b_n}^2\).
    Hence \[E(U_n^2) \lesssim n\eta_{b_n}^3 \mathbf H_4^4 + n^2\eta_{b_n}^2 \mathbf H_{4+\delta}^4\beta_4(b_n)^{\delta/(4+\delta)}.\]
    Because \(\eta_{b_n}^4=o(n)\), we have \(\eta_{b_n}^3=o(n)\) and
    \(\eta_{b_n}^2/n\to 0\). Moreover,
    \(\eta_{b_n}^2\beta_4(b_n)^{\delta/(4+\delta)} = \frac{\eta_{b_n}^2}{n}\Bigl[n\beta_4(b_n)^{\delta/(4+\delta)}\Bigr] = o(1)\).
    It follows that \(E(U_n^2)=o(n^2)\), so \(U_n=o_p(n)\).

    Combining the three steps gives \(\nu_{n,\kappa}^2-\nu_n^2=o_p(n)\).

    It remains to account for centering. Since \(\sum_{i,j}\kappa_{ij}(h_i-\bar h)(h_j-\bar h)-\nu_{n,\kappa}^2 = -2\bar h\sum_{i,j}\kappa_{ij}h_i + \bar h^2\sum_{i,j}\kappa_{ij}\),
    and since \(\nu_n^2=\Theta(n)\) by Theorem~\ref{thm:non-degenerate-clt}, we have
    \(\bar h=O_p(n^{-1/2})\). Moreover,
    \(\left|\sum_{i,j}\kappa_{ij}h_i\right| \le \|\kappa\|_\infty\eta_{b_n}\sum_i |h_i| = O_p(n\eta_{b_n})\),
    because \(\sum_i|h_i|=O_p(n)\), and also
    \(\sum_{i,j}|\kappa_{ij}|\le \|\kappa\|_\infty n\eta_{b_n}\).
    Therefore, \(\left|\sum_{i,j}\kappa_{ij}(h_i-\bar h)(h_j-\bar h)-\nu_{n,\kappa}^2\right| = O_p(\eta_{b_n} n^{1/2}) + O_p(\eta_{b_n}) = o_p(n)\),
    since \(\eta_{b_n}^4=o(n)\) implies \(\eta_{b_n}=o(n^{1/4})\). This proves the claim.
\end{proof}

Define \(\zeta_i := \frac{1}{n-1}\sum_{k\neq i}\bigl(\hat H_i(X_k)-\theta_{ik}\bigr)\) and \(\rho_i := \frac{1}{n-1}\sum_{k\neq i} R_{ik}\).

\begin{lemma}[Aggregate row-average bounds]\label{lem:rowsum-hac-row-average-bounds}
    Under the assumptions of Proposition~\ref{prop:benchmark-hac-projection},
    \(\sum_{i\in I_n} \zeta_i^2 = O_p\!\left(\eta_{b_n} + n\beta(1,1,b_n)^{\delta/(2+\delta)}\right)\)
    and
    \(\sum_{i\in I_n} \rho_i^2 = O_p\!\left(\eta_{b_n} + n\beta(1,2,b_n)^{\delta/(2+\delta)}\right)\).
    In particular,
    \(\sum_{i\in I_n} (\zeta_i-\bar\zeta)^2 = o_p\!\left(\frac{n}{\eta_{b_n}^2}\right)\)
    and
    \(\sum_{i\in I_n} (\rho_i-\bar\rho)^2 = o_p\!\left(\frac{n}{\eta_{b_n}^2}\right)\),
    where \(\bar\zeta=n^{-1}\sum_i\zeta_i\) and \(\bar\rho=n^{-1}\sum_i\rho_i\).
\end{lemma}

\begin{proof}
    \textit{The reverse-projection average \(\zeta_i\).}
    Set \(g_{ik}:=\hat H_i(X_k)-\theta_{ik}\) for \(k\neq i\).
    Then \(Eg_{ik}=0\) and \(\|g_{ik}\|_{2+\delta}\lesssim \mathbf H_{2+\delta}\) uniformly in
    \((i,k)\). Hence \(E\zeta_i^2 = \frac{1}{(n-1)^2}\sum_{k\neq i}\sum_{\ell\neq i} E(g_{ik}g_{i\ell})\).
    For each fixed \(i\), at most \(Cn\eta_{b_n}\) ordered pairs \((k,\ell)\) satisfy
    \(\mathbf d(k,\ell)\le b_n\), and each corresponding summand is \(O(\mathbf H_2^2)\). For the
    remaining pairs, Lemma~A.6 gives \(|E(g_{ik}g_{i\ell})| \lesssim \mathbf H_{2+\delta}^2\beta(1,1,b_n)^{\delta/(2+\delta)}\).
    Therefore, \(E\zeta_i^2 \lesssim \frac{n\eta_{b_n}}{n^2}\mathbf H_2^2 + \mathbf H_{2+\delta}^2\beta(1,1,b_n)^{\delta/(2+\delta)}\),
    and summing over \(i\) proves the first display.

    \medskip

    \textit{The remainder average \(\rho_i\).}
    Using \(\rho_i=(n-1)^{-1}\sum_{k\neq i}R_{ik}\), we have \(E\rho_i^2 = \frac{1}{(n-1)^2}\sum_{k\neq i}\sum_{\ell\neq i} E(R_{ik}R_{i\ell})\).
    For each fixed \(i\), the number of ordered pairs \((k,\ell)\) for which at least one of
    \(\mathbf d(i,k)\le b_n\), \(\mathbf d(i,\ell)\le b_n\), or \(\mathbf d(k,\ell)\le b_n\) holds
    is bounded by \(Cn\eta_{b_n}\), and each summand is \(O(\mathbf H_2^2)\).

    On the complement of this set, all three pairwise distances exceed \(b_n\). In
    particular, \(k\) is a \(b_n\)-free index in \((i,k,\ell)\). Coupling \(X_k\) to an
    independent copy by Berbee's lemma, and using Lemma~A.7 exactly as in the proof of
    Theorem~\ref{thm:lln for non-degenerate}, yields \(|E(R_{ik}R_{i\ell})| \lesssim \mathbf H_{2+\delta}^2\beta(1,2,b_n)^{\delta/(2+\delta)}\).
    Hence \(E\rho_i^2 \lesssim \frac{n\eta_{b_n}}{n^2}\mathbf H_2^2 + \mathbf H_{2+\delta}^2\beta(1,2,b_n)^{\delta/(2+\delta)}\),
    and summing over \(i\) proves the second display.

    Finally, centering can increase the sum of squares by at most a constant factor:
    \(\sum_i(\zeta_i-\bar\zeta)^2\le 4\sum_i\zeta_i^2\) and likewise for \(\rho_i\). Since
    \(\eta_{b_n}^4=o(n)\) implies \(\eta_{b_n}^3=o(n)\) and \(n/\eta_{b_n}^2\to\infty\), while
    \(\beta(1,1,b_n)\le\beta_4(b_n)\) and \(\beta(1,2,b_n)\le\beta_4(b_n)\), we obtain
    \(\eta_{b_n} + n\beta(1,1,b_n)^{\delta/(2+\delta)} = o\!\left(\frac{n}{\eta_{b_n}^2}\right)\),
    and the same for the term involving \(\beta(1,2,b_n)\). The last display follows.
\end{proof}

\begin{proof}[Proof of Theorem~\ref{thm:rowsum-hac-consistency}]
    Because \(\mu_i=\mu\) for all \(i\in I_n\),
    \(Q_i = \mu + h_i + \zeta_i + \rho_i\). Hence, with \(r_i := (\zeta_i-\bar\zeta) + (\rho_i-\bar\rho)\), we have \(Q_i-\bar Q = (h_i-\bar h) + r_i\). Therefore, \(\hat\nu_n^2 = \sum_{i,j}\kappa_{ij}(h_i-\bar h)(h_j-\bar h) + 2A_n + B_n\), where \(A_n := \sum_{i,j}\kappa_{ij}(h_i-\bar h)r_j\) and \(B_n := \sum_{i,j}\kappa_{ij}r_i r_j\).
    By Proposition~\ref{prop:benchmark-hac-projection},
    \(\sum_{i,j}\kappa_{ij}(h_i-\bar h)(h_j-\bar h) = \nu_n^2 + o_p(n)\).
    It remains to show that \(A_n=o_p(n)\) and \(B_n=o_p(n)\).

    Lemma~\ref{lem:rowsum-hac-row-average-bounds} gives
    \(\sum_i (\zeta_i-\bar\zeta)^2 = o_p\!\left(\frac{n}{\eta_{b_n}^2}\right)\) and \(\sum_i (\rho_i-\bar\rho)^2 = o_p\!\left(\frac{n}{\eta_{b_n}^2}\right)\).
    Hence \(\sum_i r_i^2 = o_p\!\left(\frac{n}{\eta_{b_n}^2}\right)\).
    Also, \(E\sum_i (h_i-\bar h)^2 \le 4\sum_i Eh_i^2 \lesssim n\),
    so \(\sum_i (h_i-\bar h)^2=O_p(n)\).

    Applying Lemma~\ref{lem:rowsum-hac-weighted-bound}, we obtain
    \[
        |A_n|
        \lesssim
        \eta_{b_n}
        \Big(\sum_i (h_i-\bar h)^2\Big)^{1/2}
        \Big(\sum_i r_i^2\Big)^{1/2}
        = o_p(n),
    \]
    and \(|B_n| \lesssim \eta_{b_n}\sum_i r_i^2 = o_p(n)\).
    Therefore \(\hat\nu_n^2 = \nu_n^2 + o_p(n)\).
    Since \(\nu_n^2=\Theta(n)\) under Theorem~\ref{thm:non-degenerate-clt}, ratio consistency follows:
    \(\frac{\hat\nu_n^2}{\nu_n^2} \to_p 1\).
\end{proof}

Assumption~\ref{asmp:hac-non-degenerate}(iv) only enters through the identity
\(Q_i-\bar Q=(h_i-\bar h)+(\zeta_i-\bar\zeta)+(\rho_i-\bar\rho)\), so that the
constant deterministic row mean cancels after centering. The rest of the proof uses the
same coupling and counting arguments as in the proofs of the non-degenerate limit
theorems.

\subsection{Proof of \autoref{thm:degenerate u variance}}% For simplicity, we consider \(d = 1\); the multivariate case is analogous because we work with product kernels, see \cite{zheng1996ConsistentTest} and \cite{fan1999CentralLimit}.

\begin{proof}[Proof of \autoref{thm:degenerate u variance}]
    Let \(H_{ik} = H(X_{i}, X_{k})\) and \(H_{i\tilde{k}} = H(X_{i}, \tilde{X}_{k})\), where \(\tilde{X}_{k}\) is an independent copy of \(X_{k}\). Writing \(S_{n} = \sum_{i}\sum_{k\neq i} H_{ik}\), we decompose the variance as
    \begin{align*}
        s_{n}^{2}
        &= \Var\pqty{S_{n}} \\
        &= 2\sum_i \sum_{k\neq i}\Var(H_{ik})
        + 4\sum_{i}\sum_{k_{1}\neq i}\sum_{k_{2}\neq i,k_{1}} \Cov\pqty{H_{ik_{1}}, H_{ik_{2}}}
        + \sum_{(i_{1},i_{2},k_{1},k_{2})\in \mathcal D_{4}} \Cov\pqty{H_{i_{1}k_{1}}, H_{i_{2}k_{2}}} \\
        &=: 2\sigma_{n,1}^{2} + 4\mathbf{R}_{1} + \mathbf{R}_{2},
    \end{align*}
    where \(\mathcal D_{4}\) denotes the set of quadruples \((i_{1},i_{2},k_{1},k_{2})\) with all indices distinct.

    For the first term,
    \begin{align*}
        \sigma_{n,1}^{2}
        &= \sum_i \sum_{k\neq i} \Var(H_{ik}) \\
        &= \sum_i \sum_{k\neq i} \Ep H_{i\tilde{k}}^{2}
        + \sum_i \sum_{k\neq i} \Ep\pqty{H_{ik}^{2} - H_{i\tilde{k}}^{2}}
        - \sum_i \sum_{k\neq i} \abs{\Ep H_{ik}}^{2} \\
        &= \sigma_{n}^{2} + \mathbf{I}_{1} - \mathbf{I}_{0},
    \end{align*}
    where \(\sigma_{n}^{2} := \sum_i \sum_{k\neq i} \Ep H_{i\tilde{k}}^{2}\), \(\mathbf{I}_{1} := \sum_i \sum_{k\neq i} \Ep\pqty{H_{ik}^{2} - H_{i\tilde{k}}^{2}}\), and \(\mathbf{I}_{0} := \sum_i \sum_{k\neq i} \abs{\Ep H_{ik}}^{2}\).
    By splitting the index pairs according to \(T_{2}^{m}\) and \(T_{1,1}^{m}\), Berbee's lemma and H\"older's inequality give
    \begin{equation}\label{eq:nonpar-I1-bound}
        \abs{\mathbf{I}_{1}}
        \lesssim
        \tau_{2}^{m} \mathbf{H}_{2}^{2}
        + \tau_{1,1}^{m} \mathbf{H}_{2 + \delta}^{2} \beta(1,1,m)^{\frac{\delta}{2 + \delta}}.
    \end{equation}
    Since \(\Ep H_{i\tilde{k}} = 0\), the same argument also yields
    \begin{equation}\label{eq:nonpar-I0-bound}
        \mathbf{I}_{0}
        \lesssim
        \tau_{2}^{m} \mathbf{H}_{1}^{2}
        + \tau_{1,1}^{m} \mathbf{H}_{1+\delta}^{2} \beta(1,1,m)^{\frac{2\delta}{1 + \delta}}.
    \end{equation}

    We now bound \(\mathbf{R}_{1}\). The summation runs over vectors of three different indices, so we split according to \(T_{3}^{m}\), \(T_{2,1}^{m}\), and \(T_{1^{3}}^{m}\). For \((i,k_{1},k_{2})\in T_{3}^{m}\), Cauchy--Schwarz gives \(\abs{\Cov\pqty{H_{ik_{1}},H_{ik_{2}}}} \leq \sqrt{\Ep H_{ik_{1}}^{2}\, \Ep H_{ik_{2}}^{2}} \leq \mathbf{H}_{2}^{2}\).
    For \((i,k_{1},k_{2})\in T_{2,1}^{m}\), there are two cases. If the repeated index is \(i\), equivalently \(d(i,\{k_{1},k_{2}\})>m\), then
    \begin{align*}
        \abs{\Cov\pqty{H_{ik_{1}}, H_{ik_{2}}}}
        &\leq \abs{\Ep H_{ik_{1}}H_{ik_{2}} - \Ep H_{\tilde{i}k_{1}}H_{\tilde{i}k_{2}}}
        + \abs{\Ep H_{\tilde{i}k_{1}}H_{\tilde{i}k_{2}}}
        \correction{+\abs{\Ep H_{ik_1}\Ep H_{ik_2}}} \\
        &\lesssim \mathbf{H}_{2+\delta}^{2} \beta(1,2,m)^{\frac{\delta}{2 + \delta}} + \gamma_{m,1},
    \end{align*}
    by \correction{Berbee's lemma, \autoref{lemma:gamma-prod}, and the \(J=1\) case of \autoref{lemma:m-free}}. If one of \(k_{1}\) or \(k_{2}\) is an \(m\)-free index, then \correction{the argument of \autoref{lemma:m-free}} yields \(\abs{\Cov\pqty{H_{ik_{1}}, H_{ik_{2}}}} \lesssim \mathbf{H}_{2 + \delta}^{2} \beta(1,2,m)^{\frac{\delta}{2 + \delta}}\).
    Therefore,
    \begin{equation}\label{eq:nonpar-R1-bound}
        \mathbf{R}_{1}
        \lesssim \tau_{3}^{m} \mathbf{H}_{2}^{2}
        + \pqty{\tau_{2,1}^{m} + \tau_{1^{3}}^{m}}\mathbf{H}_{2+\delta}^{2} \beta(1,2,m)^{\frac{\delta}{2 + \delta}}
        + \tau_{2,1}^{m}\gamma_{m,1}.
    \end{equation}

    For \(\mathbf{R}_{2}\), the summation runs over vectors of four different indices. For \((i_{1},i_{2},k_{1},k_{2}) \in T_{4}^{m}\), \(\abs{\Cov\pqty{H_{i_{1}k_{1}}, H_{i_{2}k_{2}}}} \leq \mathbf{H}_{2}^{2}\). For \((i_{1},i_{2},k_{1},k_{2})\in T_{3,1}^{m}\cup T_{2,1,1}^{m} \cup T_{1^{4}}^{m}\), \correction{the \(J=2\) and \(J=1\) cases of \autoref{lemma:m-free} give} \(\abs{\Cov\pqty{H_{i_{1}k_{1}}, H_{i_{2}k_{2}}}} \lesssim \mathbf{H}_{2 + \delta}^{2} \beta(1,3,m)^{\frac{\delta}{2 + \delta}}\).
    Finally, consider \((i_{1},i_{2},k_{1},k_{2}) \in T^{m}_{2,2}\). If, after relabeling, either \(d(i_{1},i_{2}) \leq m\) or \(d(i_{1},k_{2}) \leq m\), then \(\abs{\Cov\pqty{H_{i_{1}k_{1}}, H_{i_{2}k_{2}}}} \lesssim \mathbf{H}_{2 +\delta}^{2}\beta(2,2,m)^{\frac{\delta}{2 + \delta}} + \gamma_{m,1}\)\correction{, using the \(J=1\) case of \autoref{lemma:m-free} for the product of means}.
    Otherwise the two close pairs are \((i_{1},k_{1})\) and \((i_{2},k_{2})\), and the pair of pairs is \(m\)-separated, so \(\abs{\Cov\pqty{H_{i_{1}k_{1}}, H_{i_{2}k_{2}}}} \lesssim \mathbf{H}_{2 + \delta}^{2} \beta(2,2,m)^{\frac{\delta}{2+\delta}}\).
    Hence
    \begin{equation}\label{eq:nonpar-R2-bound}
        \mathbf{R}_{2}
        \lesssim \tau_{4}^{m} \mathbf{H}_{2}^{2}
        + \hat{\tau}_{4}^{m}\mathbf{H}_{2+\delta}^{2}\beta(1,3,m)^{\frac{\delta}{2+ \delta}}
        + \tau_{2,2}^{m}\bqty{\mathbf{H}_{2 + \delta}^{2}\beta(2,2,m)^{\frac{\delta}{2 + \delta}} + \gamma_{m,1}}.
    \end{equation}

    We now compare these remainder terms with the theorem assumption. By monotonicity of moments and mixing coefficients, \(\mathbf{H}_{1} \leq \mathbf{H}_{2}\), \(\mathbf{H}_{1+\delta} \leq \mathbf{H}_{2+\delta}\), and definition of \(\beta_{4}(m)\). Since the shorter index profiles embed into the longer ones, \(\tau_{2}^{m},\tau_{3}^{m} \leq \tau_{4}^{m}\), \(\tau_{1,1}^{m},\tau_{2,1}^{m} \leq \tau_{2,2}^{m}\), and \(\tau_{1^{3}}^{m} \leq \hat{\tau}_{4}^{m}\). Moreover, for \(0\leq \beta\leq 1\), \(\beta^{\frac{2\delta}{1+\delta}} \leq \beta^{\frac{\delta}{2+\delta}}\).
    Therefore \eqref{eq:nonpar-I1-bound}--\eqref{eq:nonpar-R2-bound} imply
    \[
        \abs{\mathbf{I}_{1}} + \mathbf{I}_{0} + \abs{\mathbf{R}_{1}} + \abs{\mathbf{R}_{2}}
        \lesssim
        \tau_{4}^{m} \mathbf{H}_{2}^{2}
        + \pqty{\hat{\tau}_{4}^{m} + \tau_{2,2}^{m}} \mathbf{H}_{2+\delta}^{2}\beta_{4}(m)^{\frac{\delta}{2 + \delta}}
        + \tau_{2,2}^{m} \gamma_{m,1}
        = o(\sigma_{n}^{2}).
    \]
    Consequently, \(s_{n}^{2} = 2\sigma_{n}^{2} + 2\mathbf{I}_{1} - 2\mathbf{I}_{0} + 4\mathbf{R}_{1} + \mathbf{R}_{2} = (2 + o(1))\sigma_{n}^{2}\),
    as claimed.
\end{proof}

\subsection{Limiting distribution for the test statistic}

Now we prove \autoref{thm:nonpar test} using \autoref{thm:lln for non-degenerate}, \autoref{thm:degenerate}, and \autoref{thm:degenerate u variance}.
\begin{proof}[Proof of \autoref{thm:nonpar test}]
    Write \(B = b^{d}\). Under \(\mathbb H_{0}\), let \(u_{i} = Y_{i} - g(Z_{i}, \gamma_{0})\), \(I_{n1} := I_{n}^{\circ} = B^{- \frac{1}{2}}\sum_{i}\sum_{j\neq i} u_{i}u_{j} K_{ij}\), and \(\bar{s}_{n}^{2} = 2B^{-1}\sum_{i}\sum_{j\neq i} u_{i}^{2}u_{j}^{2} K_{ij}^{2}\).
    We decompose the feasible statistic as
    \begin{equation}\label{eq:nonpar-decomp-In}
        I_{n}
        = I_{n1} - 2I_{n2} + I_{n3},
    \end{equation}
    where \(I_{n2} := B^{-\frac12}\sum_{i\neq j} u_{i}\delta_{j}K_{ij}\), \(I_{n3} := B^{-\frac12}\sum_{i\neq j} \delta_{i}\delta_{j}K_{ij}\), and \(\delta_{j} := g(Z_{j},\hat\gamma)-g(Z_{j},\gamma_{0})\).
    Choose any \(\delta \in \pqty{0, \frac{4\xi}{1-\xi}}\).
    We will prove the following four claims:
    \begin{enumerate}[(a)]
        \item \(s_{n}^{2} = (2+o(1))\sigma_{n}^{2} \asymp n^{2}\), where
            \(\sigma_{n}^{2} := B^{-1}\sum_{i\neq j}\Ep\bqty{u_{i}^{2}\tilde u_{j}^{2}K_{i\tilde j}^{2}}\);
        \item \(I_{n1}/s_{n} \rightsquigarrow N(0,1)\);
        \item \(I_{n2} = o_{p}(n)\) and \(I_{n3} = o_{p}(n)\);
        \item \(\hat s_{n}^{2} = (2+o_{p}(1))\sigma_{n}^{2}\).
    \end{enumerate}
    Since \(s_{n}\asymp n\) by part (a), these four claims imply \(\frac{I_{n1}}{s_{n}} \rightsquigarrow N(0,1)\), \(\frac{I_{n}-I_{n1}}{s_{n}} \to_{p} 0\), and \(\frac{\hat s_{n}^{2}}{s_{n}^{2}} \to_{p} 1\),
    and therefore \(\mathbf T_{n} = I_{n}/\hat s_{n} \rightsquigarrow N(0,1)\) by Slutsky's theorem.

    We first record the standard kernel-overlap bound. For every \(p \geq 1\), by \autoref{asmp:z} and \autoref{asmp:nonpar kernel},
    \begin{equation}\label{eq:kernel-Lp-bounds}
        \norm{K_{ij}}_{p}^{p}
        = \Ep\abs{K\pqty{\frac{Z_{i} - Z_{j}}{b}}}^{p}
        = B\int \abs{K(v)}^{p} f_{i,j}(z + bv, z)\dd z\dd v
        = O(B),
    \end{equation}
    and similarly \(\norm{K_{i\tilde{j}}}_{p}^{p} = O(B)\). For every \(p \geq 1\) such that \(p(1-\xi) < 4\), H\"older's inequality, \autoref{asmp:nonpar moments}, and \eqref{eq:kernel-Lp-bounds} give
    \begin{equation}\label{eq:nonpar-Hp-rate}
        \mathbf{H}_{p} = O\pqty{B^{\lambda_{p}}},
        \qquad
        \lambda_{p} := \frac{4 - 2p - p(1-\xi)}{4p} = \frac{4 - 3p + p\xi}{4p}.
    \end{equation}
    In particular, because \(\delta < \frac{4\xi}{1-\xi}\), the bound applies at \(p = 2, 2+\delta, 4\), and \(4+\delta\).

    \textit{Claim (a).}
    Consider the kernel \(H_{n}(x,y) = B^{- \frac{1}{2}} x_{1}y_{1} K\pqty{\frac{x_{2} - y_{2}}{b}}\).
    Under \(\mathbb H_{0}\), this kernel is degenerate. Indeed, for an independent copy \(\tilde{X}_{k} = (\tilde{u}_{k}, \tilde{Z}_{k})\),
    \begin{align*}
        \Ep\bqty{H_{n}(X_{i}, \tilde{X}_{k}) \mid X_{i}}
        &= B^{- \frac{1}{2}} u_{i} \Ep\bqty{\tilde{u}_{k} K\pqty{\frac{Z_{i} - \tilde{Z}_{k}}{b}} \mid X_{i}} \\
        &= B^{- \frac{1}{2}} u_{i} \Ep\bqty{\Ep\pqty{\tilde{u}_{k} \mid \tilde{Z}_{k}, X_{i}} K\pqty{\frac{Z_{i} - \tilde{Z}_{k}}{b}} \mid X_{i}} = 0.
    \end{align*}
    Moreover,
    \begin{align*}
        \Ep H_{i\tilde{k}}^{2}
        &= \frac{1}{B} \Ep\bqty{u_{i}^{2}\tilde{u}_{k}^{2} K_{i\tilde{k}}^{2}} \\
        &= \frac{1}{B} \Ep\bqty{\Ep\bqty{u_{i}^{2} \mid Z_{i}} \, \Ep\bqty{\tilde{u}_{k}^{2} \mid \tilde{Z}_{k}} \, K^{2}\pqty{\frac{Z_{i} - \tilde{Z}_{k}}{b}}} \\
        &= \int \mu_{i,2}(z + bv) \mu_{k,2}(z) K(v)^{2} f_{i}(z + bv) f_{k}(z)\dd z\dd v.
    \end{align*}
    The displayed integral is uniformly bounded, so \(\sigma_n^2 \lesssim n^2\). For the lower bound, let \(\underline{\mu}, \underline{f} > 0\) be the constants from Assumptions~\ref{asmp:nonpar moments} and~\ref{asmp:z}. Choose a bounded measurable set \(\mathcal{Z}_0 \subset \mathcal{Z}\) with positive measure. Since \(K\geq 0\) and \(\int K(v)\dd v = 1\), there exist a bounded measurable set \(A \subset \mathbb{R}^{d}\) with positive measure and a constant \(\underline{K} > 0\) such that \(K(v)^2 \geq \underline{K}\) for all \(v \in A\). Because \(A\) is bounded and translations are continuous in \(L^1\),
    \[
        \inf_{v\in A} \int \indicator{z+bv \in \mathcal{Z}_0}\indicator{z \in \mathcal{Z}_0}\dd z
        \to
        \int \indicator{z \in \mathcal{Z}_0}\dd z
        > 0
    \]
    as \(b\to 0\). Hence, for all sufficiently large \(n\) and all \(i\neq k\),
    \[
        \Ep H_{i\tilde{k}}^{2}
        \geq
        \underline{\mu}^{2}\underline{f}^{2}\underline{K}
        \int_A \int \indicator{z+bv \in \mathcal{Z}_0}\indicator{z \in \mathcal{Z}_0}\dd z\dd v
        \geq c
    \]
    for some constant \(c>0\). Therefore,
    \begin{equation}\label{eq:nonpar-sigma-order}
        \sigma_{n}^{2} \asymp n^{2}.
    \end{equation}
    \correction{Thus \(\liminf_{n\to\infty}\mathbf H_2\geq\sqrt c>0\).}

    Let \(m = C\log n\), so \(\eta_{m} = O(m^{C_{1}})\). By \eqref{eq:nonpar-Hp-rate}, \(\mathbf{H}_{2}^{2} = O\pqty{B^{2\lambda_{2}}} = O\pqty{B^{\frac{\xi-1}{2}}}\) and \(\mathbf{H}_{2+\delta}^{2} = O\pqty{B^{2\lambda_{2+\delta}}}\).
    Since \(\tau_{4}^{m} \lesssim n\eta_{m}^{3}\), we obtain
    \[
        \tau_{4}^{m}\mathbf{H}_{2}^{2}
        = O\pqty{n\eta_{m}^{3} B^{\frac{\xi-1}{2}}}
        = o(n^{2}),
    \]
    because \(\alpha < (2-\xi)^{-1}\) and \(\eta_{m}\) grows only polylogarithmically. Likewise,
    \[
        \pqty{\hat\tau_{4}^{m}+\tau_{2,2}^{m}} \mathbf{H}_{2+\delta}^{2}\beta_{4}(m)^{\frac{\delta}{2+\delta}}
        = o(n^{2})
    \]
    for a sufficiently large choice of \(C\), since \correction{\(\beta_4(m)\leq\beta(2,\infty,m)=\beta(m)=O(\rho^m)\) by symmetry and mixing monotonicity}.

    \correction{Finally, for \(k_1\neq k_2\), let
        \((\tilde X_i,\tilde X_j)\) have the law of \((X_i,X_j)\) and be
        independent of \((X_{k_1},X_{k_2})\). A change of variables and the
        uniform density bounds give \(\Ep\abs{K_{\tilde i k_1}K_{\tilde j k_2}}^q=O(B^2)\) for every \(q\geq1\). Hence, for
        \(\xi_1\in(1/2,1)\), H\"older's inequality gives
        \begin{align*}
            \Ep\abs{\gga_{k_1k_2}}
            &\leq \frac1B
            \Ep\abs{\tilde u_i\tilde u_j u_{k_1}u_{k_2}
            K_{\tilde i k_1}K_{\tilde j k_2}}\\
            &\leq\frac1B
            \norm{\tilde u_i\tilde u_j u_{k_1}u_{k_2}}_{\frac1{1-\xi_1}}
            \norm{K_{\tilde i k_1}K_{\tilde j k_2}}_{\frac1{\xi_1}}
            =O\pqty{B^{-1+2\xi_1}}.
    \end{align*}}
    Therefore \(\tau_{2,2}^{m}\gamma_{m,1} \lesssim n^{2}\eta_{m}^{2} B^{-1 + 2\xi_{1}} = o(n^{2})\),
    because \(2\xi_{1}-1>0\). The conditions of \autoref{thm:degenerate u variance} are now verified, so
    \begin{equation}\label{eq:oracle-variance-order}
        s_{n}^{2} = (2+o(1))\sigma_{n}^{2} \asymp n^{2}.
    \end{equation}

    \textit{Claim (b).}
    We verify the conditions of \autoref{thm:degenerate} for the kernel \(H_{ij}=B^{-1/2}u_{i}u_{j}K_{ij}\). First, \(\correction{\abs{\Ep I_{n1}} = \abs{\sum_{i\neq j}\Ep H_{ij}}} \lesssim \tau_{2}^{m}\mathbf{H}_{1} + n^{2}\mathbf{H}_{1+\delta}\beta(1,1,m)^{\frac{\delta}{1+\delta}}\),
    and the right-hand side is \(o(s_{n})\) by \eqref{eq:oracle-variance-order} and the same choice of \(m=C\log n\).

    Next, \(\Ep H_{ij}^{4} = \Ep\bqty{\frac{1}{B^{2}}u_{i}^{4}u_{j}^{4}K_{ij}^{4}} \leq \frac{1}{B^{2}}\norm{u_{i}^{4}u_{j}^{4}}_{\frac{1}{1-\xi}}\norm{K_{ij}^{4}}_{\xi^{-1}} = O(B^{\xi - 2})\), and therefore \(\frac{1}{n}\mathbf{H}_{4}^{4} = O\pqty{n^{-1+\alpha(2-\xi)}} = o(1)\),
    because \(\alpha < (2-\xi)^{-1}\).

    To bound \(\mathbf\Gamma_{m,2}\), first consider the case \(i=j\) with \(X_{k_{1}},X_{k_{2}},X_{i}\) mutually independent. Then
    \begin{align*}
        \Ep \pqty{\gga_{k_{1}k_{2}}}^{2}
        &= \Ep\bqty{\bqty{\Ep\pqty{H_{\tilde{i}k_{1}}H_{\tilde{i}k_{2}}\mid X_{k_{1}},X_{k_{2}}}}^{2}} \\
        &= \frac{1}{B^{2}}\Ep\bqty{u_{k_{1}}^{2}u_{k_{2}}^{2}\pqty{\Ep^{k_{1}k_{2}}\pqty{u_{i}^{2}K_{ik_{1}}K_{ik_{2}}}}^{2}} \\
        &\lesssim B \int f_{k_{1}}(z_{1})f_{k_{2}}(z_{1}-bv_{2})
        \pqty{\int K(v_{1})K(v_{1}+v_{2})f_{i}(z_{1}+bv_{1})\dd v_{1}}^{2} \dd z_{1}\dd v_{2}
        + o(B) \\
        &= O(B).
    \end{align*}
    For \(i\neq j\), with \(\sigma(X_{i},X_{j})\), \(\sigma(X_{k_{1}})\), and \(\sigma(X_{k_{2}})\) mutually independent,
    \begin{align*}
        \Ep\pqty{\gga_{k_{1}k_{2}}}^{2}
        &= \frac{1}{B^{2}} \Ep\bqty{u_{k_{1}}^{2}u_{k_{2}}^{2}\pqty{\Ep^{k_{1}k_{2}} \pqty{u_{i}u_{j} K_{ik_{1}} K_{jk_{2}} }}^{2}} \\
        &\leq \frac{1}{B^{2}} \Ep\bqty{u_{k_{1}}^{2}u_{k_{2}}^{2}
        \pqty{\norm{u_{i}u_{j}}_{\frac{4}{1-\xi}} \pqty{B^{2}\int K^{\frac{4}{3+\xi}}(v_{1}) K^{\frac{4}{3+\xi}}(v_{2}) f_{i,j}(Z_{k_{1}}+bv_{1},Z_{k_{2}}+bv_{2})\dd v_{1}\dd v_{2}}^{\frac{3+\xi}{4}}}^{2}} \\
        &= O(B^{1+\xi}).
    \end{align*}
    Hence \(\mathbf\Gamma_{m,2}^{2} = O(B)\), and therefore \(\eta_m^4 \mathbf\Gamma_{m,2}^{2} = O\pqty{(\log n)^{4C_1} n^{-\alpha}} = o(1)\).
    Next, since \(\mathbf{H}_{4}^{4} = O(B^{\xi-2})\),
    \(\frac{\eta_m^7 + \eta_{4m}^3}{n}\mathbf{H}_{4}^{4} = O\pqty{\frac{(\log n)^{7C_1}}{n} B^{\xi-2}} = o(1)\),
    because \(B^{-1} = O(n^{\alpha})\) with \(\alpha < (2-\xi)^{-1}\). Finally, \(\mathbf{H}_{4+\delta}^{4} = O(B^{4\lambda_{4+\delta}})\), and with a sufficiently large choice of \(C\),
    \(n^{2}\eta_m^2\mathbf{H}_{4+\delta}^{4}\beta(m)^{\frac{\delta}{4+\delta}} = O\pqty{n^{2}(\log n)^{2C_1} B^{4\lambda_{4+\delta}} \rho^{C\log n\frac{\delta}{4+\delta}}} = o(1)\).
    Therefore \autoref{thm:degenerate} yields \(\frac{I_{n1} - \Ep I_{n1}}{s_{n}} \rightsquigarrow N(0,1)\).
    Since \(\Ep I_{n1} = o(s_{n})\), we conclude that
    \begin{equation}\label{eq:oracle-clt-nonpar}
        \frac{I_{n1}}{s_{n}} \rightsquigarrow N(0,1).
    \end{equation}

    \textit{Claim (c).}
    We want to show $I_{n2}=o_p(n)$ and $I_{n3}=o_p(n)$.
    % The only delicate part is the linear term inside $I_{n2}$.
    % A direct use of Theorem 3.1 is too weak here: after symmetrization, Theorem 3.1 only gives $I_{n21}^{(r)}=o_p(n^2)$, whereas we need $I_{n21}^{(r)}=o_p(n^{3/2})$ in order to multiply by $\|\hat\gamma-\gamma_0\|=O_p(n^{-1/2})$ and still obtain $o_p(n)$.
    Let $\delta_\gamma:=\hat\gamma-\gamma_0$, $\Delta(z):=\partial_\gamma g(z,\gamma_0)$, and $\Delta_j:=\Delta(Z_j)$. By Taylor's theorem,
    \[
        \delta_j=g(Z_j,\hat\gamma)-g(Z_j,\gamma_0)=\Delta_j^\top \delta_\gamma+\mathrm{Rem}_j,
        \qquad
        |\mathrm{Rem}_j|\le C M_g(Z_j)\|\delta_\gamma\|^2.
    \]
    Hence
    \[
        I_{n2}=I_{n21}^\top \delta_\gamma+R_{n,2},
        \qquad
        I_{n21}:=B^{-1/2}\sum_{i\neq j}u_i\Delta_j K_{ij},
    \]
    with
    \[
        |R_{n,2}|
        \le C\|\delta_\gamma\|^2 B^{-1/2}\sum_{i\neq j}|u_i|M_g(Z_j)|K_{ij}|.
    \]
    Since $\sqrt n \delta_\gamma=O_p(1)$ and using \autoref{asmp:nonpar moments}, \autoref{asmp:z} and \eqref{eq:kernel-Lp-bounds}, $R_{n,2}=o_p(n)$.

    It remains to bound $I_{n21}$. Fix a coordinate $r$ and define
    \[
        J^{(r)}(x,y):=\frac12 B^{-1/2}\bigl[x_1\Delta^{(r)}(y_2)+y_1\Delta^{(r)}(x_2)\bigr]K\!\left(\frac{x_2-y_2}{b}\right).
    \]
    Then $I_{n21}^{(r)}=\sum_{i\neq j}J^{(r)}(X_i,X_j)$. Write the Hoeffding decomposition as
    \[
        I_{n21}^{(r)}=\hat I_{n21}^{(r)}+I_{n21}^{*(r)},
        \qquad
        \hat I_{n21}^{(r)}=2(n-1)\sum_i q_i^{(r)},
    \]
    where $q_i^{(r)}=(n-1)^{-1}\sum_{k\neq i}\hat J_k^{(r)}(X_i)$ and $\hat J_k^{(r)}(x):=\mathbb E[J^{(r)}(X_k,x)]$.

    For $x=(u,z)$,
    \[
        \hat J_k^{(r)}(u,z)
        =\frac12 B^{-1/2}\mathbb E\!\left[u_k\Delta^{(r)}(z)K\!\left(\frac{Z_k-z}{b}\right)+u\,\Delta^{(r)}(Z_k)K\!\left(\frac{Z_k-z}{b}\right)\right].
    \]
    The first term is zero because $\mathbb E(u_k\mid Z_k)=0$. Hence only the second term remains, and after the change of variables $Z_k=z+bv$,
    \[
        \hat J_k^{(r)}(u,z)=\frac12 B^{1/2}u\,\psi_{k,r}(z),
        \qquad
        \psi_{k,r}(z):=\int \Delta^{(r)}(z+bv)K(v)f_k(z+bv)\,dv.
    \]

    % Therefore
    % \[
    %     \hat J_k^{(r)}(X_i)=\frac12 B^{1/2}u_i\psi_{k,r}(Z_i),
    %     \qquad
    %     \theta_{ik}^{J,r}:=\mathbb E[\hat J_k^{(r)}(X_i)]=0,
    % \]
    % and the first-order term is of size $B^{1/2}$ rather than order one. Under \autoref{asmp:z}, we have \(\sup_{i,k}\norm{\psi_{k,r}(Z_{i})}_{2+\delta_{0}} < \infty\), this implies $\sup_i\|q_i^{(r)}\|_{2+\delta_0}\lesssim B^{1/2}$.
    Let $p:=2+\delta_0$, where $\delta_0>0$ is chosen small enough so that
    \(
        \delta_0<\frac{2\xi}{2-\xi}.
    \)
    In particular, since $\xi\in(0,1)$, we have \correction{\(\delta_0<2\) and therefore \(p<4\). More generally, for every \(t\in[1,4]\),}
    \[
        \correction{|\psi_{k,r}(z)|^t}
        \lesssim
        \int M_g(z+bv)^{\correction{t}} K(v)f_k(z+bv)\,dv,
    \]
    integrating against $f_i(z)\,dz$, changing variables $w=z+bv$, and using
    $\sup_i\sup_z f_i(z)<\infty$, \correction{we~get}
    \[
        \correction{\|\psi_{k,r}(Z_i)\|_t^t}
        \lesssim
        \int M_g(w)^{\correction{t}} f_k(w)\,dw \le
        \Big(\int M_g(w)^4 f_k(w)\,dw\Big)^{\correction{t/4}}.
    \]

    By \autoref{asmp:z}(b), \correction{\(\sup_{i,k}\|\psi_{k,r}(Z_i)\|_t<\infty\) for every \(t\in[1,4]\)}.

    Now write
    \[
        q_i^{(r)}=\frac12 B^{1/2}u_i\bar\psi_i,
        \qquad
        \bar\psi_i:=\frac{1}{n-1}\sum_{k\neq i}\psi_{k,r}(Z_i).
    \]
    Choose $s\le 4$ so that
    \(
        \frac{1}{2+\delta_0}=\frac{1}{4/(1-\xi)}+\frac{1}{s}.
    \)
    By Assumption~4.2, $\sup_i\|u_i\|_{4/(1-\xi)}<\infty$. Therefore Hölder's inequality gives
    \(
        \sup_i\|q_i^{(r)}\|_{2+\delta_0}
        \lesssim
        B^{1/2}.
    \)

    Since \(q_{i}^{(r)}\) is centered, for \(j\notin \nb{i}{m}\), \autoref{lemma:aux-one-index} gives \(\abs{\Ep\bqty{q_{i}^{(r)} q_{j}^{(r)}}} \lesssim B \beta(1,1,m)^{\frac{\delta_{0}}{2+\delta_{0}}}\), while for \(j\in \nb{i}{m}\), \(\abs{\Ep\bqty{q_{i}^{(r)} q_{j}^{(r)}}} \lesssim B\).
    So that with the choice \(m = C\log n\) with sufficiently large \(C\),
    \(
        \hat I_{n21}^{(r)}=O_p\bigl(n^{3/2}B^{1/2}\eta_m^{1/2}\bigr)=o_p(n^{3/2}).
    \)
    \correction{For the degenerate remainder, let
        \[
            \mathbf J_t:=\sup_{i\neq j}\left\{\norm{J^{(r)}(X_i,X_j)}_t,\norm{J^{(r)}(X_i,\tilde X_j)}_t\right\}.
        \]
        The H\"older calculation used in \eqref{eq:nonpar-Hp-rate} gives
        \(\mathbf J_t=O(B^{\kappa_t})\), where
        \(\kappa_t=t^{-1}-(4-\xi)/4\), for \(t=2,2+\delta_0\).
        The second-moment argument in the proof of
        \autoref{thm:degenerate u variance} therefore gives
        \[
            \Ep\abs{I_{n21}^{*(r)}}^2
            \lesssim
            (\tau_4^m+\tau_{2,2}^m)\mathbf J_2^2
            +n^4\mathbf J_{2+\delta_0}^2
            \beta(1,3,m)^{\frac{\delta_0}{2+\delta_0}}
            =o(n^3).
        \]
        Here the first term uses \(\alpha(2-\xi)<1\), and the second is
        \(o(n^3)\) for a sufficiently large \(C\). Hence
    \(I_{n21}^{*(r)}=o_p(n^{3/2})\).}
    Hence
    \(
        I_{n21}^{(r)}=o_p(n^{3/2}).
    \)
    Since the parameter dimension is finite, $\|I_{n21}\|=o_p(n^{3/2})$, and therefore
    \[
        |I_{n21}^\top \delta_\gamma|
        \le \|I_{n21}\|\,\|\delta_\gamma\|
        =o_p(n^{3/2})O_p(n^{-1/2})
        =o_p(n).
    \]
    Combining this with $R_{n,2}=o_p(n)$, we obtain $I_{n2}=o_p(n)$.

    Finally,
    \(
        |I_{n3}|
        \le C\|\delta_\gamma\|^2 B^{-1/2}\sum_{i\neq j}M_g(Z_i)M_g(Z_j)|K_{ij}|.
    \)
    The sum has expectation $O(n^2B^{1/2})$, so
    \(
        I_{n3}=O_p(n^{-1})O_p(n^2B^{1/2})=O_p(nB^{1/2})=o_p(n).
    \)

    \textit{Claim (d).}
    The oracle variance estimator is \(\bar{s}_{n}^{2} = 2B^{-1}\sum_{i\neq j} u_{i}^{2}u_{j}^{2}K_{ij}^{2}\). Consider the non-degenerate kernel \(G_{n}(x,y) = 2B^{-1}x_{1}^{2}y_{1}^{2} K\pqty{\frac{x_{2} - y_{2}}{b}}^{2}\). Let \(\mathbf{G}_{p} := \sup_{i\neq k}\cqty{\norm{G_{n}(X_{i},X_{k})}_{p}, \norm{G_{n}(X_{i},\tilde{X}_{k})}_{p}}\). Its mean satisfies \(\theta_{ik}^{G} := \Ep G_{n}(X_{i}, \tilde{X}_{k}) = 2B^{-1}\Ep\bqty{u_{i}^{2}\tilde{u}_{k}^{2} K_{i\tilde{k}}^{2}}\),
    and therefore \(\sum_i \sum_{k\neq i} \theta_{ik}^{G} = 2\sigma_{n}^{2}\). Set \(\delta_{G} := \delta/2\). Since \(G_{n} = 2H_{n}^{2}\),
    \(\mathbf{G}_{2}^{2} = 4\mathbf{H}_{4}^{4}\) and \(\mathbf{G}_{2+\delta_{G}}^{2} = 4\mathbf{H}_{4+\delta}^{4}\).
    Using \eqref{eq:nonpar-Hp-rate} and \(\tau_{2}^{m} \lesssim n\eta_{m}\),
    \(\tau_{2}^{m}\mathbf{G}_{2}^{2} \lesssim n\eta_{m} B^{\xi-2} = o(n^{2})\),
    because \(B^{-1} = O(n^{\alpha})\) with \(\alpha < (2-\xi)^{-1}\). Moreover, for the same choice of \(m = C\log n\) as in Claim (b),
    \(\mathbf{G}_{2+\delta_{G}}^{2}\beta(1,3,m)^{\frac{\delta_{G}}{2+\delta_{G}}} = 4\mathbf{H}_{4+\delta}^{4}\beta(1,3,m)^{\frac{\delta}{4+\delta}} = o(1)\).
    Hence the assumptions of \autoref{thm:lln for non-degenerate} are satisfied for the kernel \(G_{n}\), and
    \begin{equation}\label{eq:oracle-variance-estimator}
        \bar{s}_{n}^{2} = \sum_{i\neq k} \theta_{ik}^{G} + o_{p}(n^{2}) = 2\sigma_{n}^{2} + o_{p}(n^{2}).
    \end{equation}
    Combining \eqref{eq:oracle-variance-estimator} with \eqref{eq:oracle-variance-order}, we obtain
    \begin{equation}\label{eq:oracle-variance-consistency}
        \frac{\bar{s}_{n}^{2}}{s_{n}^{2}} \to_{p} 1.
    \end{equation}

    It remains to compare \(\hat{s}_{n}^{2}\) and \(\bar{s}_{n}^{2}\). Let \(D_{i} = \hat{u}_{i}^{2} - u_{i}^{2}\). Then \(D_{i} = -2u_{i}\delta_{i} + \delta_{i}^{2}\), and \(\hat{s}_{n}^{2} - \bar{s}_{n}^{2} = 2B^{-1}\sum_{i\neq j} \bqty{D_{i}u_{j}^{2} + u_{i}^{2}D_{j} + D_{i}D_{j}} K_{ij}^{2}\).
    Each \(D_{i}\) is bounded by a constant multiple of
    \(\norm{\delta_\gamma}\abs{u_{i}}M_{g}(Z_{i}) + \norm{\delta_\gamma}^{2}M_{g}(Z_{i})^{2}\).
    Hence every term in \(\hat{s}_{n}^{2} - \bar{s}_{n}^{2}\) is dominated by a constant multiple of either
    \[
        \norm{\delta_\gamma}
        B^{-1}\sum_{i\neq j} \abs{u_{i}}^{a}\abs{u_{j}}^{b}M_{g}(Z_{i})^{c}M_{g}(Z_{j})^{d}K_{ij}^{2},
        \qquad a+b+c+d\leq 4,
    \]
    or the same quantity multiplied by \correction{additional factors of} \(\norm{\delta_\gamma}\). \correction{For a term containing \(\norm{\delta_\gamma}^q\), let \(U\) and \(V\) be the total powers of \(\abs u\) and \(M_g\), respectively. H\"older's inequality and Assumptions~\ref{asmp:nonpar moments}--\ref{asmp:z} give expectation \(O\pqty{n^2B^{-r_{U,V}}}\), where \(r_{U,V}:=\pqty{U(1-\xi)+V}/4\), for its nonnegative prefactored sum; Markov's inequality makes the sum \(O_p(n^2B^{-r_{U,V}})\). The expansion gives \((q,U,V)\in\{(1,3,1),(2,2,2),(3,1,3),(4,0,4)\}\). For each tuple,
        \(n^{-q/2}B^{-r_{U,V}}=o(1)\) under
    \(\alpha<(2-\xi)^{-1}\).}
    Since \(\norm{\delta_\gamma} = O_{p}(n^{-1/2})\), all correction terms are \(o_{p}(n^{2})\). Therefore \(\hat{s}_{n}^{2} - \bar{s}_{n}^{2} = o_{p}(n^{2}).\)
    Together with \eqref{eq:oracle-variance-consistency}, we conclude that
    \begin{equation}\label{eq:nonpar-feasible-variance-consistency}
        \frac{\hat{s}_{n}^{2}}{s_{n}^{2}} \to_{p} 1.
    \end{equation}
    \correction{Consequently,} \(\mathbf T_{n} = \frac{I_{n}}{\hat s_{n}} \rightsquigarrow N(0,1).\)
\end{proof}

\section{Additional Details and Robustness Checks for the Empirical Application}\label{sec:empirical_appendix}

This appendix records the sample construction, residual-dependence diagnostics, and supplementary robustness checks for the Dublin Airbnb application.

\subsection{OLS Regression Setup}

The analysis uses the Dublin listings snapshot on Inside Airbnb scraped on September 16, 2025. Restricting attention to \textit{Entire home/apt} listings with positive nightly prices leaves 2,958 observations. Dropping rows with missing controls yields the final regression sample of 2,494 listings.

The dependent variable is \(Y_i=\log(\text{price}_i)\). The baseline scalar location index is \(Z_i=\log(d_i+0.25)\), where \(d_i\) is straight-line distance in kilometers from the Spire of Dublin at \((53.3498,-6.2603)\). Count-like controls with substantial right skewness, namely accommodates, minimum nights, number of reviews, and host listings count, enter as log terms. Binary indicators and bounded rates remain in levels. The baseline regression includes property-type fixed effects.
Table~\ref{tab:dublin_appendix_variables} reports summary statistics for the variables used in the baseline regression.

\begin{table}[htbp]
    \centering
    \caption{Regression Variables and Summary Statistics}
    \label{tab:dublin_appendix_variables}
    \footnotesize
    \setlength{\tabcolsep}{3pt}
    \begin{tabular}{p{0.23\linewidth}p{0.27\linewidth}rrrrr}
        \hline
        Variable & Coding & $N$ & Mean & Std. dev. & Min & Max \\
        \hline
        Log price & \(\log(\text{price})\) & 2494 & 5.357 & 0.598 & 3.892 & 9.860 \\
        Log-distance to city center & \(\log(d+0.25)\), \(d\) in km from the Spire & 2494 & 1.206 & 0.957 & -1.301 & 3.414 \\
        Accommodates (log) & \(\log(1+\text{accommodates})\) & 2494 & 1.578 & 0.384 & 0.693 & 2.833 \\
        Bedrooms & Parsed bedroom count & 2494 & 1.959 & 1.195 & 0.000 & 20.000 \\
        Bathrooms & Parsed bathroom count & 2494 & 1.513 & 0.927 & 0.000 & 20.000 \\
        Beds & Parsed bed count & 2494 & 2.436 & 1.748 & 0.000 & 21.000 \\
        Minimum nights (log) & \(\log(1+\text{minimum nights})\) & 2494 & 1.442 & 0.736 & 0.693 & 5.903 \\
        Superhost & Indicator & 2494 & 0.343 & 0.475 & 0.000 & 1.000 \\
        Identity verified & Indicator & 2494 & 0.972 & 0.166 & 0.000 & 1.000 \\
        Instant bookable & Indicator & 2494 & 0.231 & 0.421 & 0.000 & 1.000 \\
        Response rate & Host response rate (\%) & 2494 & 92.582 & 21.159 & 0.000 & 100.000 \\
        Acceptance rate & Host acceptance rate (\%) & 2494 & 80.977 & 27.315 & 0.000 & 100.000 \\
        Number of reviews (log) & \(\log(1+\text{reviews})\) & 2494 & 2.788 & 1.780 & 0.000 & 7.539 \\
        Reviews per month & Parsed reviews per month & 2494 & 1.897 & 2.252 & 0.000 & 32.400 \\
        Rating & Median-imputed review score & 2494 & 4.762 & 0.332 & 1.000 & 5.000 \\
        Rating missing & Missing-rating indicator & 2494 & 0.140 & 0.347 & 0.000 & 1.000 \\
        Listings count (log) & \(\log(1+\text{host listings})\) & 2494 & 1.506 & 1.171 & 0.693 & 4.984 \\
        Property-type fixed effects & 15 most common types, plus \textit{Other}; one omitted & -- & -- & -- & -- & -- \\
        \hline
    \end{tabular}
    \vspace{0.2em}

    \parbox{0.92\linewidth}{\footnotesize Notes: Summary statistics are computed on the complete-case regression sample used in the baseline Dublin estimation. The variable labels match Table~\ref{tab:dublin_main}. For indicator variables, the mean is the sample share equal to one. The underlying construction follows \textit{empirical/airbnb\_data.py}.}
\end{table}

\subsection{Residual Dependence Diagnostics}

Figure~\ref{fig:dublin_appendix_residual_dependence_plot} summarizes residual dependence by geographic distance using equal-width 0.25 km bins. The orange bars report the distribution of unordered listing pairs over distance, while the blue curve reports a Gaussian-kernel estimate of the mean residual cross-product \(E[\hat u_i \hat u_j \mid d_{ij}]\). The pattern is consistent with local spatial dependence: the average residual cross-product is \(0.005491\) for pairs within 1 km, but falls to \(-0.000669\) for pairs more than 10 km apart. This short-range dependence motivates the use of dependence-robust procedures based on neighbourhood blocks, grid blocks, and spatial multiplier scales.

\begin{figure}[htbp]
    \centering
    \includegraphics[width=0.7\linewidth]{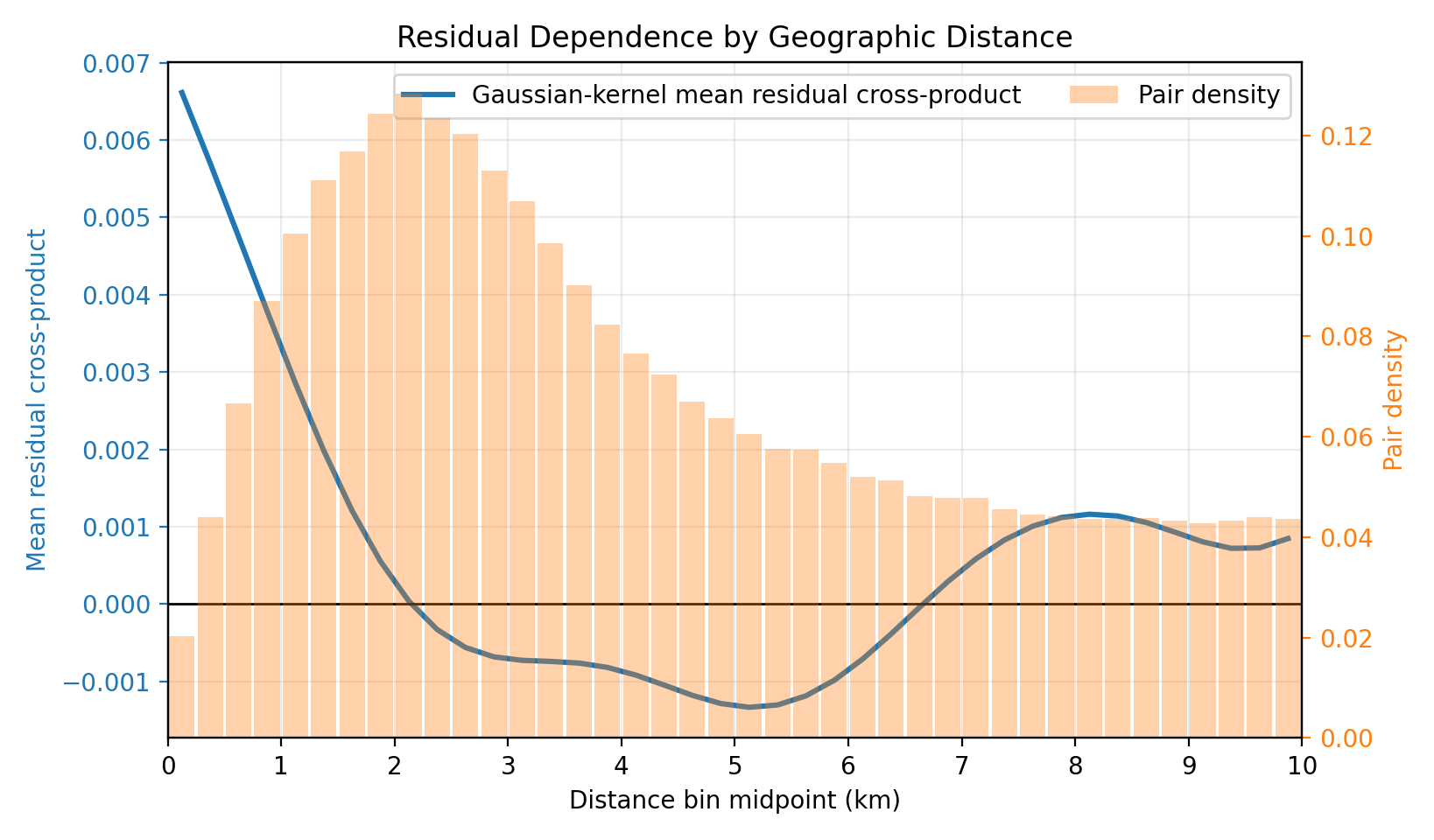}
    \caption{Pair Density and Smoothed Mean Residual Cross-Products by Distance}
    \label{fig:dublin_appendix_residual_dependence_plot}
    \vspace{0.2em}
\end{figure}

\subsection{Additional Robustness Checks}

\correction{Table~\ref{tab:dublin_appendix_robustness} reports supplementary checks using a cross-validated bandwidth, trimming extreme values of the predictor, and omitting the constant from the log-distance transformation:}
\begin{itemize}
    \item \textit{Grouped spatial cross-validation.} To assess whether the benchmark bandwidth is driving the non-rejection, listings are projected into local planar coordinates, partitioned into 2 km grid cells, and assigned to five folds, which yields 163 spatial groups in the Dublin sample. For each candidate bandwidth on a restricted grid around the rule-of-thumb pilot, \correction{we} partial out the non-distance controls on the training folds only, fit a local-linear scalar-\(Z\) smoother to the resulting partial residuals, and evaluate out-of-fold mean squared prediction error on the held-out folds. The reported specification uses the minimum-MSE choice.
    \item \textit{\correction{Trimming extreme predictor values}.} The trimming exercise drops the bottom and top 1\% of the baseline log-distance \(Z_i=\log(d_i+0.25)\).
    \item \textit{Alternative distance transform.} We \correction{replace} \(Z_i=\log(d_i+0.25)\) with \(Z_i=\log d_i\) when constructing the log-distance.
\end{itemize}

Taken together, these results support the same qualitative conclusion as the main text. Across different specifications, the specification test for linearity remains non-rejecting.

\begin{table}[htbp]
    \centering
    \caption{Supplementary Robustness Checks}
    \label{tab:dublin_appendix_robustness}
    \small
    \begin{tabular}{lccccc}
        \hline
        Specification & Sample size & Bandwidth & Simple $T$ & Asymptotic $p$ & 0.5 km grid $p$ \\
        \hline
        Baseline rule-of-thumb & 2494 & 0.2121 & 0.855 & 0.196 & 0.552 \\
        Grouped spatial CV & 2494 & 0.1962 & 1.019 & 0.154 & 0.530 \\
        1\% / 99\% trim in \(Z\) & 2445 & 0.2046 & 1.037 & 0.150 & 0.520 \\
        No-offset log distance & 2494 & 0.2393 & 0.688 & 0.246 & 0.604 \\
        \hline
    \end{tabular}
    \vspace{0.2em}

    \parbox{0.92\linewidth}{\footnotesize Notes: The grouped-spatial-CV bandwidth is selected by five-fold cross-validation over 163 projected 2 km grid groups. The validation criterion is out-of-fold mean squared prediction error from a local-linear scalar-\(Z\) smoother fitted to residuals after partialling out the non-distance controls, and the reported choice uses the minimum-CV-MSE rule on a theory-constrained candidate set. The trimming specification drops the bottom and top 1\% of the baseline index \(Z_i=\log(d_i+0.25)\). The no-offset specification uses \(Z_i=\log d_i\). All dependence-robust \(p\)-values in this table are based on the studentized 0.5 km grid block bootstrap.}
\end{table}

\section{Extension to allow for random distance}\label{s2: conditional berbee}
There are recently growing interests in allowing for random distance in the literature, see \cite{kojevnikov2020LimitTheorems}, where they allow for the distance to depend on a randomly generated network \(\mathcal{G}\). Our results can be similarly extended to allow for conditioning on the random distance. Recall that our random elements are defined on a probability space \(\pqty{\Omega, \mathcal{F}, \Prob}\).

Let now our distance on \(\mathcal{I}_{n}\) be \(\dist: \mathcal{I}_{n}\times \mathcal{I}_{n} \times \Omega\), where for \(\Prob\)-almost surely  \(\omega\in \Omega\), \(\dist(\cdot,\cdot,\omega)\) is a distance on \(\mathcal{I}_{n}\) and \(\dist\pqty{i, j, \cdot}\) viewed a random distance function is measurable with respect to some sub-\(\sigma\)-algebra \(\mathcal{H}_{n}\subset \mathcal{F}\). For example, in \cite{kojevnikov2020LimitTheorems}, the distance \(\dist\) is a measurable function of a random network \(\mathcal{G}_{n}\pqty{\omega}\), and \(\mathcal{G}(\omega)\) is measurable with respect to \(\mathcal{H}_{n}\).

\begin{assumption}
    There is a \textit{regular conditional probability} \(\rcd{A}: \Omega\times \mathcal{F} \to [0,1]\) given \(\mathcal{H}\), such that
    \begin{enumerate}
        \item For fixed \(A\), \(\rcd{A}= \Prob\pqty{A \mid \mathcal{H}}\), almost surely;
        \item For \(\Prob\)-almost surely \(\omega\), \(\rcd{\cdot}\) defines a probability measure on \((\Omega,\mathcal{F})\).
    \end{enumerate}
\end{assumption}

Then we can define the conditional \(\beta\)-mixing coefficient for sub-\(\sigma\)-algebras \(\mathcal{F}_{1}, \mathcal{F}_{2}\), see conditional \(\alpha\)-mixing coefficients in \cite{prakasarao2009ConditionalIndependence}.
\begin{definition}
    The conditional \(\beta\)-mixing coefficient for sub-\(\sigma\)-algebras \((\mathcal{F}_{1}, \mathcal{F}_{2})\) is a \(\mathcal{H}\)-measurable random variable \(\beta^{\mathcal{H}}\pqty{\mathcal{F}_{1}, \mathcal{F}_{2}}\), such that, almost surely,
    \begin{equation*}
        \beta^{\mathcal{H}} \pqty{\mathcal{F}_{1},\mathcal{F}_{2}}(\omega) =  \sup_{\mathcal{A}, \mathcal{B}} \frac{1}{2} \sum_{A_{i}\in \mathcal{A}}\sum_{B_{j}\in \mathcal{B}}\abs{\rcd{A_{i}\cap B_{j}} - \rcd{A_{i}} \rcd{B_{j}}},
    \end{equation*}
    where the supremum is taken over all finite partitions \(\mathcal{A} =\cqty{A_{i}: i \leq n_{A}}\) and \(\mathcal{B} = \cqty{B_{j}: j \leq n_{B}}\) of \(\Omega\), such that \( \mathcal{A} \subset \mathcal{F}_{1}\) and \(\mathcal{B}\subset \mathcal{F}_{2}\).
\end{definition}

Assume the following conditional mixing condition holds,
\begin{assumption}\label{asmp:conditional mixing}
    There exist random variables \(\beta_{n}(n_{1},n_{2}, m) \leq 1\) that are \(\mathcal{H}_{n}\)-measurable, such that,
    \begin{equation}
        \sup_{(\mathcal{I}_{1}, \mathcal{I}_{2}) \in \mathcal{P}_{n}(n_{1},n_{2},m)}\beta^{\mathcal{H}_{n}}\pqty{\sigma\pqty{X_{i}: i \in \mathcal{I}_{1}}, \sigma\pqty{X_{i} : i \in \mathcal{I}_{2}}} \leq \beta_{n}(n_{1}, n_{2}, m).
    \end{equation}
    and let \(\beta\pqty{n_{1},n_{2},m} = \sup_{n} \beta_{n}\pqty{n_{1},n_{2},m}\).
\end{assumption}

\correction{Subject to the corresponding conditional moment, sparsity, rate, and normalization conditions, our results extend conditionally on \(\mathcal{H}\) once Berbee's lemma is extended to the conditional case.}
\begin{lemma}[{Conditional Berbee's Lemma}]
    Assume \(\pqty{\Omega, \mathcal{F}, \Prob}\) admits a regular conditional probability given \(\mathcal{H}\). Let \(\mathcal{H}, \mathcal{A}\) be sub-\(\sigma\)-algebras of \(\mathcal{F}\), and let \(Y: \pqty{\Omega,\mathcal{F},\Prob} \to \pqty{\mathbb{R}^{d}, \mathcal{B}_{\mathbb{R}^{d}}}\) be a random vector. Possibly after enlarging the probability space, there exists a random vector \(\tilde{Y}\) such that
    \begin{equation*}
        \tilde{Y} \perp\!\!\!\perp \mathcal{A} \mid \mathcal{H}, \qquad \mathcal{L}\pqty{\tilde{Y} \mid \mathcal{H}} = \mathcal{L}\pqty{Y \mid \mathcal{H}},
    \end{equation*}
    and
    \begin{equation*}
        \Prob\pqty{Y \neq \tilde{Y} \mid \mathcal{H}} = \beta^{\mathcal{H}}\pqty{\mathcal{A}, \sigma(Y)} \qquad \text{a.s.}
    \end{equation*}
\end{lemma}
\begin{proof}
    For each Borel set \(B \subset \mathbb{R}^{d}\), define the probability kernels
    \begin{equation*}
        Q(\omega, B) = \Prob\pqty{Y\in B \mid \mathcal{A} \vee \mathcal{H}}(\omega), \qquad \mu(\omega, B) = \Prob\pqty{Y\in B \mid \mathcal{H}}(\omega).
    \end{equation*}
    Since \(\mathbb{R}^{d}\) is Polish, there exist jointly measurable Radon-Nikodym derivatives of \(Q(\omega, \cdot)\) and \(\mu(\omega, \cdot)\) with respect to \(\nu(\omega, \cdot) = Q(\omega, \cdot) + \mu(\omega, \cdot)\). Denote them by
    \begin{equation*}
        q(\omega,y) = \frac{\dd Q(\omega,\cdot)}{\dd \nu(\omega,\cdot)}(y), \qquad m(\omega,y) = \frac{\dd \mu(\omega,\cdot)}{\dd \nu(\omega,\cdot)}(y),
    \end{equation*}
    and set
    \begin{equation*}
        h(\omega,y) = q(\omega,y) \wedge m(\omega,y), \qquad r(\omega) = \int h(\omega,y)\,\nu(\omega,\dd y).
    \end{equation*}
    For each \(\omega\), define a probability measure \(K_{\omega}\) on \(\mathbb{R}^{d}\times \mathbb{R}^{d}\) by
    \begin{align*}
        K_{\omega}(\dd y, \dd z)
        &= h(\omega,y)\,\nu(\omega,\dd y)\,\delta_{y}(\dd z) \\
        &\quad + \indicator{r(\omega)<1}\frac{(q-m)^{+}(\omega,y)(m-q)^{+}(\omega,z)}{1-r(\omega)}\,\nu(\omega,\dd y)\,\nu(\omega,\dd z).
    \end{align*}
    A direct calculation shows that the marginals of \(K_{\omega}\) are \(Q(\omega,\cdot)\) and \(\mu(\omega,\cdot)\), and that
    \begin{equation*}
        K_{\omega}\pqty{\cqty{(y,z): y\neq z}} = 1-r(\omega) = \norm{Q(\omega,\cdot)-\mu(\omega,\cdot)}_{\mathrm{TV}}.
    \end{equation*}
    Thus \(K_{\omega}\) is a measurable maximal coupling of the two conditional laws.

    Because \(\mathbb{R}^{d}\) is Polish, \(K_{\omega}\) admits a measurable disintegration
    \begin{equation*}
        K_{\omega}(\dd y,\dd z) = Q(\omega,\dd y)\,k(\omega,y,\dd z)
    \end{equation*}
    for some probability kernel \(k\). Enlarge the probability space so that it supports an auxiliary \(U\sim \mathrm{Unif}[0,1]\) independent of \(\mathcal{F}\). By the standard randomization theorem for kernels on Polish spaces, there exists a jointly measurable map \(\Phi\) such that, conditional on \((\omega,y)\), \(\Phi(\omega,y,U)\) has law \(k(\omega,y,\cdot)\). Define \(\tilde{Y}(\omega,u) = \Phi(\omega, Y(\omega), u)\) on the enlarged space \(\pqty{\Omega\times [0,1], \mathcal{F}\otimes \mathcal{B}([0,1]), \Prob\otimes \mathrm{Leb}}\). Then the conditional law of \((Y,\tilde{Y})\) given \(\mathcal{A}\vee \mathcal{H}\) is exactly \(K_{\omega}\).

    Therefore, for every Borel set \(B\),
    \begin{equation*}
        \Prob\pqty{\tilde{Y}\in B \mid \mathcal{A}\vee \mathcal{H}} = K_{\omega}(\mathbb{R}^{d}, B) = \mu(\omega,B),
    \end{equation*}
    and since the right-hand side is \(\mathcal{H}\)-measurable, we also have
    \begin{equation*}
        \Prob\pqty{\tilde{Y}\in B \mid \mathcal{H}} = \mu(\omega,B) = \Prob\pqty{Y\in B \mid \mathcal{H}}.
    \end{equation*}
    Hence \(\tilde{Y}\) has the same conditional law as \(Y\) given \(\mathcal{H}\).

    Moreover, for any \(A_{0}\in \mathcal{A}\),
    \begin{align*}
        \Prob\pqty{A_{0}\cap \cqty{\tilde{Y}\in B} \mid \mathcal{H}}
        &= \Ep\bqty{\indicator{A_{0}}\Prob\pqty{\tilde{Y}\in B \mid \mathcal{A}\vee \mathcal{H}} \mid \mathcal{H}} \\
        &= \Ep\bqty{\indicator{A_{0}}\mu(\omega,B) \mid \mathcal{H}} \\
        &= \mu(\omega,B)\Prob\pqty{A_{0}\mid \mathcal{H}} \\
        &= \Prob\pqty{\tilde{Y}\in B\mid \mathcal{H}}\Prob\pqty{A_{0}\mid \mathcal{H}}.
    \end{align*}
    Hence \(\tilde{Y} \perp\!\!\!\perp \mathcal{A} \mid \mathcal{H}\).

    Finally, \(\Prob\pqty{Y\neq \tilde{Y} \mid \mathcal{A}\vee \mathcal{H}} = K_{\omega}\pqty{\cqty{(y,z): y\neq z}} = \norm{Q(\omega,\cdot)-\mu(\omega,\cdot)}_{\mathrm{TV}}.\)
    Taking conditional expectation given \(\mathcal{H}\) yields
    \begin{equation*}
        \Prob\pqty{Y\neq \tilde{Y}\mid \mathcal{H}} = \Ep\bqty{\norm{Q(\omega,\cdot)-\mu(\omega,\cdot)}_{\mathrm{TV}} \mid \mathcal{H}}.
    \end{equation*}
    For \(\Prob\)-almost every \(\omega\), the usual total-variation representation of the \(\beta\)-mixing coefficient applied under the regular conditional probability \(\rcd{\cdot}\) gives
    \begin{equation*}
        \beta^{\mathcal{H}}\pqty{\mathcal{A}, \sigma(Y)}(\omega)
        = \int \norm{Q(\omega',\cdot)-\mu(\omega,\cdot)}_{\mathrm{TV}}\,\lambda^{\mathcal{H}}_{\omega}(\dd \omega')
        = \Ep\bqty{\norm{Q(\omega,\cdot)-\mu(\omega,\cdot)}_{\mathrm{TV}} \mid \mathcal{H}}(\omega).
    \end{equation*}
    Combining the last two displays proves the claim.
\end{proof}

\section{Extension to a mixed regime}\label{s2: mixed regime}
In this appendix we consider the regime in which both the projection and the degenerate parts of the Hoeffding decomposition contribute.
Recall the dependence-adapted Hoeffding decomposition from \eqref{eq:hoeffding-depdendent}:
\begin{equation}\label{eq:ortho-mixed-hoeffding}
    \correction{S_n - \Ep S_n = \hat S_n + S_n^{*,c},\qquad S_n^{*,c}:=S_n^*-\Ep S_n^*,}
\end{equation}
where
\[
    \hat S_n = 2(n-1)\sum_{i\in\mathcal I_n} h_i,
    \qquad
    h_i := \frac{1}{n-1}\sum_{k\neq i}\bigl(\hat H_k(X_i)-\theta_{ik}\bigr),
\]
and
\[
    S_n^* = \sum_{i\in\mathcal I_n}\sum_{k\neq i} R_{ik},
    \qquad
    R_{ik} := H(X_i,X_k)-\hat H_i(X_k)-\hat H_k(X_i)+\theta_{ik}.
\]
We write
\begin{equation}\label{eq:ortho-remainder-kernel}
    R_{ik}(x,z)
    :=
    H(x,z)-\hat H_i(z)-\hat H_k(x)+\theta_{ik},
\end{equation}
so that $R_{ik}=R_{ik}(X_i,X_k)$.

We also write
\(
    a_n^2 := \Var(\hat S_n),
    b_n^2 := \Var(S_n^*),
    c_n := \Cov(\hat S_n,S_n^*),
    \correction{v_n^2 := \Var(\hat S_n+S_n^*) = a_n^2+b_n^2+2c_n.}
\)
As in Appendix~A.2, Jensen's inequality gives, for every $p\ge 1$,
\begin{equation}\label{eq:ortho-moment-bounds}
    \|h_i\|_p \lesssim \mathbf H_p,
    \qquad
    \|R_{ik}\|_p \lesssim \mathbf H_p.
\end{equation}

The regime of interest is the \emph{asymptotically orthogonal mixed regime}
\begin{equation}\label{eq:ortho-main-orthogonality}
    c_n = o(a_n b_n).
\end{equation}
We first propose a sufficient condition for \eqref{eq:ortho-main-orthogonality}.
This is the dependence-adapted analogue of orthogonality between the linear and degenerate components in the classical independent theory.

%% ================================================================

The next lemma is the basic estimate for the cross terms between the projection part and the degenerate part.

\begin{lemma}[Cross product bound]\label{lem:mixed-bound}
    Suppose Assumptions~\ref{asmp:no-clustering-point}--\ref{asmp:kernel} hold.
    Let $m\ge 1$ and $\delta>0$ \correction{with \(\mathbf H_{2+\delta}<\infty\)}.

    \begin{enumerate}
        \item If \correction{\(j\neq k\) and} one of the indices $i,j,k$ is $m$-free in the ordered triple
            $(i,j,k)$, then
            \begin{equation}\label{eq:ortho-mixed-triple-bound}
                |\Ep[h_i R_{jk}]|
                \lesssim
                \mathbf H_{2+\delta}^2\,\beta(1,2,m)^{\delta/(2+\delta)}.
            \end{equation}

        \item Fix $j\in\nb{i}{m}$ and let $(\tilde X_i,\tilde X_j)$ be \correction{a Berbee copy of
            \((X_i,X_j)\) independent of \(\mathcal F_{-ij}^m\) and satisfying
            \(\Prob\pqty{(X_i,X_j)\neq(\tilde X_i,\tilde X_j)}\leq\beta(2,\infty,m)\)}. For
            $k\notin\nb{j}{4m}$ define
            \begin{equation}\label{eq:ortho-Psi}
                \Psi_{ij,k}
                :=
                \Ep\Bigl[
                    \tilde h_i\Bigl\{
                        R_{jk}(\tilde X_j,X_k)+R_{kj}(X_k,\tilde X_j)
                    \Bigr\}
                    \Bigm| X_k
                \Bigr],
            \end{equation}
            where
            \[
                \tilde h_i
                :=
                \frac{1}{n-1}\sum_{\ell\neq i}
                \bigl(\hat H_\ell(\tilde X_i)-\theta_{i\ell}\bigr).
            \]
            Then $\Psi_{ij,k}$ is $\mathcal F_{-ij}^m$-measurable and
            \begin{equation}\label{eq:ortho-A3-repl-1}
                \Ep\Bigl|
                \Ep\bigl[h_i(R_{jk}+R_{kj})\mid \mathcal F_{-ij}^m\bigr]
                - \Psi_{ij,k}
                \Bigr|
                \lesssim
                \mathbf H_{2+\delta}^2\,\beta(2,\infty,m)^{\delta/(2+\delta)}.
            \end{equation}

        \item Fix $j\in\nb{i}{m}$ and\correction{, using the same copy,} for $k\notin\nb{i}{4m}$, define
            \begin{equation}\label{eq:ortho-Psitilde}
                \widetilde\Psi_{ij,k}
                :=
                \Ep\Bigl[
                    R_{ik}(\tilde X_i,X_k)\,\tilde h_j
                    \Bigm| X_k
                \Bigr],
            \end{equation}
            where
            \[
                \tilde h_j
                :=
                \frac{1}{n-1}\sum_{\ell\neq j}
                \bigl(\hat H_\ell(\tilde X_j)-\theta_{j\ell}\bigr).
            \]
            Then $\widetilde\Psi_{ij,k}$ is $\mathcal F_{-ij}^m$-measurable and
            \begin{equation}\label{eq:ortho-A3-repl-2}
                \Ep\Bigl|
                \Ep\bigl[R_{ik}h_j\mid \mathcal F_{-ij}^m\bigr]
                - \widetilde\Psi_{ij,k}
                \Bigr|
                \lesssim
                \mathbf H_{2+\delta}^2\,\beta(2,\infty,m)^{\delta/(2+\delta)}.
            \end{equation}
    \end{enumerate}
\end{lemma}

\begin{proof}
    For part~(i), if $i$ is $m$-free in $(i,j,k)$, couple $X_i$ to an
    independent copy $\tilde X_i$ with
    $\Prob(X_i\neq \tilde X_i)\le \beta(1,2,m)$. Then $\tilde h_i$ is
    independent of $(X_j,X_k)$, has mean zero, and therefore
    $\Ep[\tilde h_i R_{jk}]=0$. Hence,
    \[
        \Ep[h_iR_{jk}]
        =
        \Ep\bigl[(h_i-\tilde h_i)R_{jk}\,\mathbf 1\{X_i\neq \tilde X_i\}\bigr],
    \]
    and H\"older's inequality together with \eqref{eq:ortho-moment-bounds}
    yields \eqref{eq:ortho-mixed-triple-bound}.

    If instead $j$ is $m$-free in $(i,j,k)$, couple $X_j$ to $\tilde X_j$
    independent of $(X_i,X_k)$ with the same bound on the coupling error. By
    \autoref{lemma:aux-remainder-centered},
    $\Ep[R_{jk}(\tilde X_j,X_k)\mid X_k]=0$ a.s., so
    $\Ep[h_iR_{jk}(\tilde X_j,X_k)]=0$. Therefore,
    \[
        \Ep[h_iR_{jk}]
        =
        \Ep\Bigl[
            h_i\bigl\{R_{jk}-R_{jk}(\tilde X_j,X_k)\bigr\}
            \mathbf 1\{X_j\neq \tilde X_j\}
        \Bigr],
    \]
    and the same H\"older bound applies. The case where $k$ is $m$-free is
    identical.

    For part~(ii), let
    \[
        Y_{jk} := R_{jk}+R_{kj},
        \qquad
        \widetilde Y_{jk} := R_{jk}(\tilde X_j,X_k)+R_{kj}(X_k,\tilde X_j).
    \]
    Since $j\in\nb{i}{m}$ and $k\notin\nb{j}{4m}$, we have
    $X_k\in\mathcal F_{-ij}^m$. Conditional on $X_k$, the pair
    $(\tilde h_i,\widetilde Y_{jk})$ is measurable with respect to
    $(\tilde X_i,\tilde X_j,X_k)$, and
    \[
        \Psi_{ij,k} = \Ep[\tilde h_i\widetilde Y_{jk}\mid X_k].
    \]
    On the event $\{(X_i,X_j)=(\tilde X_i,\tilde X_j)\}$ we have
    $h_i\,Y_{jk}=\tilde h_i\,\widetilde Y_{jk}$. Hence, by H\"older's
    inequality and \eqref{eq:ortho-moment-bounds},
    \[
        \Ep\Bigl|
        \Ep[h_iY_{jk}\mid \mathcal F_{-ij}^m]-\Psi_{ij,k}
        \Bigr|
        \lesssim
        \mathbf H_{2+\delta}^2\,\beta(2,\infty,m)^{\delta/(2+\delta)}.
    \]
    This proves \eqref{eq:ortho-A3-repl-1}. Part~(iii) is identical.
\end{proof}

\begin{proposition}[Cross-covariance bound]\label{prop:ortho-cross-cov-small}
    Suppose Assumptions~\ref{asmp:no-clustering-point}--\ref{asmp:kernel} hold.
    Let $m\ge 1$ and $\delta>0$ \correction{with \(\mathbf H_{2+\delta}<\infty\) and \(\nu_n b_n>0\)}. Let $\tau_3^m$ denote the number of ordered
    triples whose indices lie in a single $m$-connected component, and let
    \(
        \hat\tau_3^m := \tau_{2,1}^m + \tau_{1,1,1}^m
    \)
    denote the number of ordered triples with an $m$-free index. Then
    \begin{equation}\label{eq:ortho-cross-cov-rho}
        \left|\frac{c_n}{a_n b_n}\right|
        \lesssim
        \frac{\tau_3^m}{\nu_n b_n}\,\mathbf H_2^2
        +
        \frac{\hat\tau_3^m}{\nu_n b_n}\,
        \mathbf H_{2+\delta}^2\,\beta(1,2,m)^{\delta/(2+\delta)},
    \end{equation}
    where $a_n = 2(n-1)\nu_n$. Equivalently,
    \begin{equation}\label{eq:ortho-cross-cov-cn}
        |c_n|
        \lesssim
        2(n-1)
        \Bigl[
            \tau_3^m\,\mathbf H_2^2
            +
            \hat\tau_3^m\,\mathbf H_{2+\delta}^2\,
            \beta(1,2,m)^{\delta/(2+\delta)}
        \Bigr].
    \end{equation}
    Consequently, if for some sequence $m=m_n$,
    \begin{equation}\label{eq:ortho-cross-cov-cond}
        \frac{\tau_3^{m_n}}{\nu_n b_n}\,\mathbf H_2^2
        +
        \frac{\hat\tau_3^{m_n}}{\nu_n b_n}\,
        \mathbf H_{2+\delta}^2\,\beta(1,2,m_n)^{\delta/(2+\delta)}
        \longrightarrow 0,
    \end{equation}
    then $c_n=o(a_n b_n)$.
\end{proposition}

\begin{proof}
    \correction{Since \(\Ep h_i=0\), \(\Cov(h_i,R_{jk})=\Ep[h_iR_{jk}]\). Thus}
    \begin{equation}\label{eq:ortho-cross-cov-start}
        c_n
        =
        \Cov\!\left(2(n-1)\sum_{i\in\mathcal I_n} h_i,\,\sum_{j\neq k}R_{jk}\right)
        =
        2(n-1)\sum_{i\in\mathcal I_n}\sum_{j\neq k}\Ep[h_iR_{jk}].
    \end{equation}
    If $(i,j,k)\in T_3^m$, then Cauchy--Schwarz and
    \eqref{eq:ortho-moment-bounds} give
    $|\Ep[h_iR_{jk}]|\lesssim \mathbf H_2^2$. Summing over all such triples,
    \begin{equation}\label{eq:ortho-cross-cov-local}
        \sum_{\substack{(i,j,k)\in T_3^m\\\correction{j\neq k}}}|\Ep[h_iR_{jk}]|
        \lesssim
        \tau_3^m\,\mathbf H_2^2.
    \end{equation}

    If $(i,j,k)$ has an $m$-free index, \autoref{lem:mixed-bound}(i) yields
    \begin{equation}\label{eq:ortho-cross-cov-nonlocal}
        \sum_{\substack{(i,j,k)\in T_{2,1}^m\cup T_{1,1,1}^m\\\correction{j\neq k}}}|\Ep[h_iR_{jk}]|
        \lesssim
        \hat\tau_3^m\,\mathbf H_{2+\delta}^2\,
        \beta(1,2,m)^{\delta/(2+\delta)}.
    \end{equation}
    Combining \eqref{eq:ortho-cross-cov-start},
    \eqref{eq:ortho-cross-cov-local}, and
    \eqref{eq:ortho-cross-cov-nonlocal} proves
    \eqref{eq:ortho-cross-cov-cn}. Dividing by $a_nb_n=2(n-1)\nu_n b_n$
    gives \eqref{eq:ortho-cross-cov-rho}. The last claim is immediate.
\end{proof}

Fix \correction{\(m\ge 1\) and \(\delta>0\) with \(v_n>0\) and \(\mathbf H_{2+\delta}<\infty\)}. Define the truncated quadratic row sums by
\[
    r_i^{(m)} := \sum_{k\notin\nb{i}{4m}} R_{ik},
    \qquad
    \correction{\bar r_j^{(m)} := \sum_{k\notin\nb{j}{4m}}(R_{jk}+R_{kj}),}
\]
and set
\[
    \correction{S_n^{*,\circ}(m) := \sum_{i\in\mathcal I_n}r_i^{(m)},}
    \qquad
    \correction{S_n^{*,\backslash}(m) := S_n^*-S_n^{*,\circ}(m).}
\]
We then write
\[
    \correction{W_n := \frac{\hat S_n+S_n^{*,c}}{v_n},}
    \qquad
    \correction{W_n^\circ(m) := \frac{\hat S_n+S_n^{*,\circ}(m)}{v_n}.}
\]
\correction{Using the variance of the deleted block gives}
\begin{align}
    \correction{\dist_W\bigl(W_n,W_n^\circ(m)\bigr)}
    &\correction{\le
    \frac{\sqrt{\Var(S_n^{*,\backslash}(m))}}{v_n}
    +\frac{\abs{\Ep S_n^{*,\circ}(m)}}{v_n}}\nonumber\\
    &\correction{\lesssim
    \frac{1}{v_n}
    \Bigl[
        \tau_4^{4m}\mathbf H_2^2+
        \tau_{2,2}^{4m}\mathbf H_{2+\delta}^2
        \beta_4(4m)^{\delta/(2+\delta)}
    \Bigr]^{1/2}
    +\frac{\abs{\Ep S_n^{*,\circ}(m)}}{v_n}}\nonumber\\
    &\correction{=:A_0(m)+\frac{\abs{\Ep S_n^{*,\circ}(m)}}{v_n}},
    \label{eq:ortho-A0}
\end{align}
\correction{where \(\beta_4(m):=\max\{\beta(q_1,q_2,m):q_1+q_2=4\}\).}

\correction{Set \(\ell_i:=2(n-1)h_i\). The Stein coupling in \autoref{lemma:stein} can be applied formally to \(W_n^\circ(m)\). For \(i\in\mathcal I_n\), define}
\begin{equation}\label{eq:ortho-Gi}
    G_i := \frac{1}{v_n}\bigl(\ell_i+r_i^{(m)}\bigr),
\end{equation}
\correction{so that \(\sum_iG_i=W_n^\circ(m)\). For \(j\in\nb{i}{m}\), let}
\begin{equation}\label{eq:ortho-Dij}
    D_{ij} := \frac{1}{v_n}\bigl(\ell_j+\bar r_j^{(m)}\bigr),
    \qquad
    \correction{D_i := \sum_{j\in\nb{i}{m}}D_{ij},}
\end{equation}
and define
\begin{equation}\label{eq:ortho-Dijprime}
    D_{ij}'
    :=
    \frac{1}{v_n}
    \sum_{u\in\nb{i}{m}\cup\nb{j}{m}}
    \bigl(\ell_u+\bar r_u^{(m)}\bigr).
\end{equation}
\correction{Then \(W_i':=W_n^\circ(m)-D_i\) is \(\mathcal F_{-i}^m\)-measurable and \(W_{ij}':=W_n^\circ(m)-D_{ij}'\) is \(\mathcal F_{-ij}^m\)-measurable. Hence}
\autoref{lemma:stein} yields
\begin{equation}\label{eq:ortho-stein}
    \correction{\dist_W\bigl(W_n^\circ(m),Z\bigr)}
    \le
    \correction{2A_1+\sqrt{\frac{2}{\pi}}A_2
    +\sqrt{\frac{2}{\pi}}A_3+2A_4+A_5},
    \qquad Z\sim N(0,1),
\end{equation}
\correction{with \(A_1,\dots,A_5\) as in \autoref{lemma:stein}. Moreover, \(\abs{\Ep S_n^{*,\circ}(m)}/v_n=\abs{\Ep W_n^\circ(m)}\leq A_1\), so the triangle inequality gives \(\dist_W(W_n,Z)\lesssim A_0(m)+A_1+A_2+A_3+A_4+A_5\).}

\correction{The pure projection terms are controlled as before with the normalizing standard deviation \(a_n=2(n-1)\nu_n\) replaced by \(v_n\); the pure degenerate terms are also as before except with \(b_n\) replaced by \(v_n\). The cross-product terms are new and are controlled by \autoref{lem:mixed-bound}.}

\end{document}